\documentclass[11pt]{article}

\usepackage{amsfonts,amsmath,amsthm,amssymb}
\usepackage[usenames,dvipsnames]{xcolor}
\usepackage{tikz}
\usetikzlibrary{backgrounds,decorations.pathreplacing,fadings, calc}
\usepackage{float}
\definecolor{DarkBlue}{rgb}{0.1,0.1,0.5}
\usepackage[colorlinks,citecolor=ForestGreen,linkcolor=blue,urlcolor=DarkBlue,pagebackref]{hyperref}
\usepackage{xspace}
\usepackage[margin=1in]{geometry}
\usepackage{graphicx} 
\usepackage[capitalise,nameinlink]{cleveref}
\usepackage{thm-restate}
\usepackage{framed}
\usepackage{tikz}
\usetikzlibrary{backgrounds,decorations.pathreplacing,fadings}
\usepackage{float}
\usepackage{multirow}
\usepackage{halloweenmath}

\definecolor{midnight}{rgb}{0.0, 0.2, 0.4}
\definecolor{dblue}{rgb}{0.12, 0.56, 1.0}
\definecolor{amber}{rgb}{1.0, 0.49, 0.0}
\definecolor{dorange}{rgb}{1.0, 0.55, 0.0}
\definecolor{dtang}{rgb}{1.0, 0.66, 0.07}
\definecolor{dsblue}{rgb}{0.0, 0.75, 1.0}
\definecolor{denim}{rgb}{0.08, 0.38, 0.74}
\definecolor{pblue}{rgb}{0.2, 0.2, 0.6}
\definecolor{gblue}{rgb}{0.0, 0.58, 0.71}
\definecolor{cdblue}{rgb}{0.16, 0.32, 0.75}
\definecolor{dlavender}{rgb}{0.45, 0.31, 0.59}
\definecolor{alizarin}{rgb}{0.82, 0.1, 0.26}
\definecolor{dorchid}{rgb}{0.6, 0.2, 0.8}
\definecolor{dimgray}{rgb}{0.41, 0.41, 0.41}
\definecolor{skobeloff}{rgb}{0.0, 0.48, 0.45}
\definecolor{cgreen}{rgb}{0.0, 0.8, 0.6}

\tikzset{
    ncbar angle/.initial=90,
    ncbar/.style={
        to path=(\tikztostart)
        -- ($(\tikztostart)!#1!\pgfkeysvalueof{/tikz/ncbar angle}:(\tikztotarget)$)
        -- ($(\tikztotarget)!($(\tikztostart)!#1!\pgfkeysvalueof{/tikz/ncbar angle}:(\tikztotarget)$)!\pgfkeysvalueof{/tikz/ncbar angle}:(\tikztostart)$)
        -- (\tikztotarget)
    },
    ncbar/.default=0.1cm,
}
\tikzset{square left brace/.style={ncbar=0.1cm}}
\tikzset{square right brace/.style={ncbar=-0.1cm}}

\newcommand{\NP}{\textsf{NP}\xspace}
\newcommand{\set}[1]{\left\{#1\right\}}
\newcommand{\grp}[1]{\left(#1\right)}
\newcommand{\pair}[1]{\langle #1\rangle}
\DeclareMathOperator{\poly}{poly}
\DeclareMathOperator{\disj}{disj}
\DeclareMathOperator{\early}{early}
\newcommand{\dia}{\diamond}
\DeclareMathOperator{\dist}{dist}
\newcommand{\ca}{\mathcal}
\renewcommand{\t}{\text}
\newcommand{\vin}{V_{\t{in}}}
\newcommand{\vout}{V_{\t{out}}}
\newcommand{\vmix}{V_{\t{mix}}}
\newcommand{\outdeg}{\deg_{\text{out}}}
\usepackage{old-arrows}
\newcommand{\rev}{\overleftarrow}

\renewcommand{\bar}{\overline}
\newcommand{\sub}{\subseteq}

\newcommand{\agree}{\text{agree}}

\newcommand{\Fagree}{F_{\agree}}
\newcommand{\Frev}{F_{\text{dis}}}
\newcommand{\Fdis}{\mathcal{F}_{\text{dis}}}

\DeclareMathOperator{\pos}{pos}

\theoremstyle{definition}
\newtheorem{theorem}{Theorem}
\newtheorem{proposition}[theorem]{Proposition}
\newtheorem{lemma}[theorem]{Lemma}

\newtheorem{corollary}[theorem]{Corollary}
\newtheorem{definition}[theorem]{Definition}
\newtheorem{hypothesis}[theorem]{Hypothesis}
\newtheorem{remark}[theorem]{Remark}

\newtheorem{claim}[theorem]{Claim}
\crefname{claim}{claim}{claims}

\newtheorem{open}{Open Problem}

\NewDocumentEnvironment{claimproof}{ }{%
  \proof
}{%
  \endproof
}

\definecolor{amethyst}{rgb}{0.6, 0.4, 0.8}

\title{Detecting Disjoint Shortest Paths in Linear Time and More}
\author{
Shyan Akmal\thanks{\url{naysh@mit.edu}. Partially supported by Virginia Vassilevska Williams' Simons Investigator Award.}
\\MIT 
\and 
Virginia Vassilevska Williams\thanks{\url{virgi@mit.edu}. Supported by NSF Grant CCF-2330048, BSF Grant 2020356, and a Simons Investigator Award.}
\\MIT
\and
Nicole Wein\thanks{\url{nswein@umich.edu}.}\\University of Michigan}
\date{\vspace{-3ex}}

\begin{document}

\maketitle

\vspace{-4ex}
\begin{abstract}
    In the \textsf{$k$-Disjoint Shortest Paths ($k$-DSP)} problem, we are given a graph $G$ (with positive edge weights) on $n$ nodes and $m$ edges with specified source vertices $s_1, \dots, s_k$, and target vertices $t_1, \dots, t_k$, and are tasked with determining if $G$ contains vertex-disjoint $(s_i,t_i)$-shortest paths.
    For any constant $k$, it is known that \textsf{$k$-DSP} can be solved in polynomial time over undirected graphs and directed acyclic graphs (DAGs).
    % [\textcolor{magenta}{Lochet, SODA 2021}].
     However, the \textit{exact} time complexity of \textsf{$k$-DSP} remains mysterious,
    with large gaps between the fastest known algorithms and best conditional lower bounds.
    In this paper, we obtain faster algorithms for important cases of \textsf{$k$-DSP}, and present better conditional lower bounds for \textsf{$k$-DSP} and its variants.
    
    Previous work solved \textsf{2-DSP} over weighted undirected graphs in $O(n^7)$ time, 
    % [\textcolor{magenta}{Akhmedov, CSR 2020}], 
    and weighted DAGs in $O(mn)$ time. 
    % [\textcolor{magenta}{Fortune, Hopcroft, and Wyllie, Theoretical Computer Science 1980}].
    For the main result of this paper, we present optimal \emph{linear time}  algorithms for solving \textsf{2-DSP} on weighted undirected graphs and DAGs.
    Our linear time algorithms are algebraic however, and so only solve the detection rather than search version of \textsf{2-DSP} (we show how to solve the search version in $O(mn)$ time, which is  faster than the previous best runtime in weighted undirected graphs, but only matches the previous best runtime for DAGs).

    We also obtain a faster algorithm for \textsf{$k$-Edge Disjoint Shortest Paths ($k$-EDSP)} in DAGs, the variant of \textsf{$k$-DSP} where one seeks edge-disjoint instead of vertex-disjoint paths between sources and their corresponding targets.
    Algorithms for \textsf{$k$-EDSP} on DAGs from previous work take  $\Omega(m^k)$ time.
    We show that \textsf{$k$-EDSP} can be solved over DAGs in $O(mn^{k-1})$ time, matching the fastest known runtime for solving \textsf{$k$-DSP} over DAGs.

    Previous work established conditional lower bounds for solving \textsf{$k$-DSP} and its variants via reductions from detecting cliques in graphs.
    Prior work implied 
    that \textsf{$k$-Clique} can be reduced to \textsf{$2k$-DSP} in DAGs and undirected graphs with $O((kn)^2)$ nodes.
    We improve this reduction, by showing how to reduce from \textsf{$k$-Clique} to  \textsf{$k$-DSP} in DAGs and undirected graphs with $O((kn)^2)$ nodes (halving the number of paths needed in the reduced instance).
    A variant of \textsf{$k$-DSP} is the \textsf{$k$-Disjoint Paths ($k$-DP)} problem, where the solution paths no longer need to be shortest paths.
    Previous work reduced from \textsf{$k$-Clique} to 
    \textsf{$p$-DP} in DAGs with $O(kn)$ nodes, for $p= k  + k(k-1)/2$.
    % [\textcolor{magenta}{Slivkins, SIAM Journal of Computing 2010}]. 
We improve this by showing a reduction from \textsf{$k$-Clique} to \textsf{$p$-DP}, for $p=k + \lfloor k^2/4\rfloor$.

Under the \textsf{$k$-Clique Hypothesis} from fine-grained complexity, our results establish better conditional lower bounds for \textsf{$k$-DSP} for all $k\ge 4$, and better conditional lower bounds for \textsf{$p$-DP} for all $p\le 4031$.
Notably, our work gives the first nontrivial conditional lower bounds \textsf{4-DP} in DAGs and \textsf{4-DSP} in undirected graphs and DAGs.
Before our work, nontrivial conditional lower bounds were only known for \textsf{$k$-DP} and \textsf{$k$-DSP} on such graphs when $k\ge 6$.

\end{abstract}

% \vfill

\paragraph*{Acknowledgements}
We thank anonymous reviewers for pointing out issues with a proof in a previous version of this work, and simplifying our construction of covering families.
We thank Malte Renken and Andr\'{e} Nichterlein for answering questions about previous work.
We thank Matthias Bentert for pointing out that  \cite{BentertFominGolovach2024} independently
proves \Cref{lbdspboth}. 
The first author additionally thanks Ryan Williams and Zixuan Xu for conversations about early versions of  ideas presented here, and Yinzhan Xu for nice discussions about the arguments in this paper.

% Remove page number for abstract page.
 % \newpage
\pagenumbering{gobble}
\newpage
 \pagenumbering{arabic}

\section{Introduction}
\label{sec:intro}

Routing disjoint paths in graphs is a classical problem in computer science.
For positive integer $k$, in the \textsf{$k$-Disjoint Paths ($k$-DP)} problem, we are given a graph $G$ with $n$ vertices and $m$ edges, with specified source nodes $s_1, \dots, s_k$ and target nodes $t_1, \dots, t_k$, and are tasked with determining if $G$ contains $(s_i,t_i)$-paths which are internally vertex-disjoint.
Beyond being a natural graph theoretic problem to study, \textsf{$k$-DP}  is important because of its deep connections with various computational tasks from the Graph Minors project \cite{RobertsonSeymour1995}.

Following a long line of research, the polynomial-time complexity of \textsf{$k$-DP} has essentially been settled: in directed graphs the \textsf{2-DP} problem is \NP-hard \cite[Lemma 3]{FortuneHopcroftWyllie1980}, and so is unlikely to admit a polynomial time algorithm, while in undirected graphs \textsf{$k$-DP} can be solved in $\tilde{O}(m+n)$ time for $k=2$ \cite{Tholey2005}, and in $O(n^2)$ time or $m^{1+o(1)}$ time for any constant $k\ge 3$ \cite{KawarabayashiKobayashiReed2012,KPS2024}. 

In this work we study an optimization variant of \textsf{$k$-DP}, the \textsf{$k$-Disjoint Shortest Paths ($k$-DSP)} problem.
In \textsf{$k$-DSP} we are given the same input as in \textsf{$k$-DP}, but are now tasked with determining if the input  contains \emph{$(s_i,t_i)$-shortest} paths which are internally vertex-disjoint.
This problem is interesting both because it is a natural graph algorithms question to investigate from the perspective of combinatorial optimization, and because understanding the complexity of \textsf{$k$-DSP} could lead to a deeper understanding of the interaction between shortest paths structures in graphs (analogous to how studying \textsf{$k$-DP} helped develop the rich theory surrounding forbidden minors in graphs). 

Compared to \textsf{$k$-DP}, not much is known about the exact time complexity of \textsf{$k$-DSP}.
In directed graphs, \textsf{2-DSP} can be solved in polynomial time \cite{BercziKobayashi2017}, but no polynomial-time algorithm (or \NP-hardness proof) is known for \textsf{$k$-DSP} for \emph{any} constant $k\ge 3$.
In undirected graphs, it was recently shown that for any constant $k$, \textsf{$k$-DSP} can be solved in polynomial time \cite{Lochet2021}. 
However,
the current best algorithms for \textsf{$k$-DSP} in undirected graphs run in 
$n^{O(k\cdot k!)}$ time, so in general this polynomial runtime is quite large for $k\ge 3$.
For example, the current fastest algorithm for \textsf{3-DSP} in undirected graphs takes $O(n^{292})$ time \cite{Geometric-Lens-conf}.

Significantly faster algorithms are known for detecting $k=2$ disjoint shortest paths.
The paper which first introduced the \textsf{$k$-DSP} problem in 1998 also presented an $O(n^8)$ time algorithm for solving \textsf{2-DSP} in weighted\footnote{Throughout, we assume that weighted graphs have positive edge weights.} undirected graphs \cite{Tzoreff1998}.
The first improvement for this problem came over twenty years later in \cite{Akhmedov2020}, 
which presented an algorithm solving \textsf{2-DSP} in weighted undirected graphs in $O(n^7)$ time, and in unweighted undirected graphs in $O(n^6)$ time.
Soon after, \cite[Theorem 1]{Geometric-Lens-conf} presented an even faster $O(mn)$ time algorithm for solving \textsf{2-DSP} in the special case of unweighted undirected graphs.\footnote{It seems plausible that the method of \cite{Geometric-Lens-conf} could be adapted to handle weighted undirected graphs as well, but such a generalization does not appear to currently be known.}

The main result of our work is an optimal algorithm for \textsf{2-DSP} in weighted undirected graphs.

\begin{restatable}{theorem}{twodspundir}
    \label{2dsp-undir}
    The \textsf{2-DSP} problem can be solved in weighted undirected graphs in $O(m+n)$ time.
\end{restatable}

% Unlike previous approaches to the \textsf{2-DSP} problem, our algorithm from \Cref{2dsp-undir} is algebraic in nature, and works by efficiently evaluating a polynomial whose terms encode all pairs of disjoint shortest paths.
% Once consequence of this approach is that our algorithm is randomized, and solves the problem with high probability.
% We discuss some limitations of this approach in \Cref{subsec:limit}, and compare the ideas used in our algorithm to previous algebraic graph algorithms in \Cref{subsec:prev-alg}. 
% \nicole{I think it would be clearer to just state the limitations here instead of having a section about it later. E.g. we could just put the word ``randomized" into the theorem statement and I'm not sure it's necessary to include the description for why it's randomized. And for detection vs search, I think we can add a sentence here saying that it only works for detection and also add a Corollary with the O(mn) search version.}

This result pins down the true time complexity of \textsf{2-DSP} in undirected graphs, and (up to certain limitations of our algorithm, which we discuss later) closes the line of research for this specific problem, initiated twenty-five years ago in \cite{Tzoreff1998}.

As discussed previously, over  directed graphs the  \textsf{2-DP} problem is \NP-hard, and the complexity of \textsf{$k$-DSP}  is open even for $k=3$.
This lack of algorithmic progress in general directed graphs has motivated researchers to characterize the complexity of disjoint path problems in restricted classes of directed graphs.
In this context, studying algorithms for routing disjoint paths in directed acyclic graphs (DAGs) has proved to be particularly fruitful.
For example, the only known polynomial time algorithm for \textsf{2-DSP} on general directed graphs works by reducing to several instances of \textsf{2-DP} on DAGs \cite{BercziKobayashi2017}.
Similarly, the fastest known algorithm for \textsf{$k$-DSP} on undirected graphs works by reducing to several instances of disjoint paths on DAGs \cite{Geometric-Lens-conf}.

It is known that \textsf{2-DP} in DAGs can be solved in linear time \cite{Tholey12}.
More generally, since 1980 it has been known that \textsf{$k$-DP} in DAGs can be solved in $O(mn^{k-1})$ time, and this remains the fastest known algorithm for these problems for all $k\ge 3$ \cite[Theorem 3]{FortuneHopcroftWyllie1980}.

As observed in \cite[Proposition 10]{BercziKobayashi2017}, the algorithm of \cite{FortuneHopcroftWyllie1980} for \textsf{$k$-DP} on DAGs can be  modified to solve \textsf{$k$-DSP} in weighted DAGs in the same $O(mn^{k-1})$ runtime.
This is the fastest known runtime for \textsf{$k$-DSP} in DAGs.
In particular, the fastest algorithm for \textsf{2-DSP} from previous work runs in $O(mn)$ time.

% \nicole{I would add here that in particular the fastest for 2-DSP is $O(mn)$ because otherwise it feels abrupt to go from talking about k-DSP to stating a result for 2-DSP}
% for all positive integers $k$.

The second result of our work is an optimal algorithm for \textsf{2-DSP} in weighted DAGs.

\begin{restatable}{theorem}{twodspDAG}
    \label{2dsp-DAG}
    The \textsf{2-DSP} problem can be solved in weighted DAGs in $O(m+n)$ time.
\end{restatable}

This settles the  time complexity of \textsf{2-DSP} in DAGs, and marks the first improvement over the $O(mn)$ time algorithm implied by \cite{FortuneHopcroftWyllie1980} from over thirty years ago.
The \textsf{2-DSP} problem in weighted DAGs generalizes the \textsf{2-DP} problem in DAGs, and so \Cref{2dsp-DAG} also offers an alternate linear time algorithm for \textsf{2-DP} in DAGs, which is arguably simpler than the previous approaches leading up to \cite{Tholey12}, many of which involved  tricky case analyses and carefully constructed data structures.

Our algorithms for solving \textsf{2-DSP} in undirected graphs and DAGs are algebraic, and work by testing whether certain polynomials, whose terms encode pairs of disjoint shortest paths in the input graph, are nonzero.
As a consequence, the algorithms establishing \Cref{2dsp-undir,2dsp-DAG} are randomized, and solve \textsf{2-DSP} with high probability.
Moreover, these algorithms determine whether the input graph has a solution, but do not explicitly return solution paths.
So while our algorithms solve the decision problem \textsf{2-DSP}, they do not solve the search problem of returning two disjoint shortest paths if they exist.
This is a common limitation for algebraic graph algorithms.

The \textsf{2-DSP} problem does admit a search to decision reduction -- with $O(m)$ calls to an algorithm which detects whether a graph contains two disjoint shortest paths, we can actually find two disjoint shortest paths if they exist.
Thanks to the algebraic nature of our algorithms, we can get a slightly better reduction, and find two disjoint shortest paths when they exist with only $O(n)$ calls to the algorithms from \Cref{2dsp-undir,2dsp-DAG}.

\begin{restatable}{theorem}{search}
    \label{thm:search}
    We can solve \textsf{2-DSP} over weighted DAGs and undirected graphs, and find a solution if it exists, in $O(mn)$ time.
\end{restatable}

So we can \emph{find} two disjoint shortest paths in weighted undirected graphs in $O(mn)$ time (which still beats the previous fastest $O(n^7)$ time algorithm for weighted undirected graphs, and matches the previous fastest algorithm for \emph{unweighted} undirected graphs), and in weighted DAGs in $O(mn)$ time (which only matches, rather than beats, the previous fastest runtime for \textsf{2-DSP} in DAGs).

Finally, one can also consider \emph{edge-disjoint} versions of all the problems discussed thus far.
The \textsf{$k$-Edge Disjoint Paths ($k$-EDP)} and \textsf{$k$-Edge Disjoint Shortest Paths ($k$-EDSP)} problems are the same as the respective \textsf{$k$-DP} and \textsf{$k$-DSP} problems, except the solutions paths merely need to be edge-disjoint instead of internally vertex-disjoint. 
% By subdividing each vertex in a graph, we can reduce from \textsf{$k$-DP} and \textsf{$k$-DSP} on $n$ nodes and $m$ edges to \textsf{$k$-EDP} and \textsf{$k$-EDSP} respectively on $2n$ nodes and $m+n$ edges.
% So the vertex-disjoint problems are not harder than their edge-disjoint variants.

For any constant $k$, there is a simple reduction from \textsf{$k$-EDSP} on $n$ nodes and $m$ edges to \textsf{$k$-DSP} on $O(m+n)$ nodes and $O(m)$ edges (see \Cref{app:unnecessary} for the details).
Combining this reduction with \Cref{2dsp-undir,2dsp-DAG}, we get that we can solve \textsf{2-EDSP} over weighted DAGs and undirected graphs in linear time as well.

\begin{corollary}
    \label{corr:2-EDSP}
    We can solve \textsf{2-EDSP} over weighted DAGs and undirected graphs in $O(m+n)$ time.
\end{corollary}

More generally, for all $k\ge 3$ the fastest known algorithms for \textsf{$k$-EDSP} on DAGs in the literature work by reducing this problem to \textsf{$k$-DSP} using the reduction mentioned in the previous paragraph.
Consequently, the current fastest algorithm for \textsf{$k$-EDSP} in DAGs runs in $O(m^k)$ time, which in dense graphs is much slower than the $O(mn^{k-1})$ time algorithm known for \textsf{$k$-DSP}.
For the same reason, the fastest known algorithm for \textsf{$k$-EDP} in DAGs for $k\ge 3$ runs in $O(m^k)$ time.

Our final algorithmic result is that we can solve \textsf{$k$-EDSP} in weighted DAGs as quickly as the fastest known algorithms for \textsf{$k$-DSP}.

\begin{restatable}{theorem}{edsp}
    \label{thm:edsp}
        The \textsf{$k$-EDSP} problem can be solved in weighted DAGs in $O(mn^{k-1})$ time.    
\end{restatable}

Since \textsf{$k$-EDSP} in weighted DAGs generalizes the \textsf{$k$-EDP} problem in DAGs (see \Cref{app:unnecessary} for details),
\Cref{thm:edsp} also implies faster algorithms for this latter problem.
Our algorithm is simple and employs the same general approach used by previous routines \cite{FortuneHopcroftWyllie1980,BercziKobayashi2017} for this problem.

\subsection*{Lower Bounds}

For $k\ge 3$, the known $O(mn^{k-1})$ algorithms for \textsf{$k$-DP} and \textsf{$k$-DSP} in DAGs have resisted any improvements over the past three decades.
Thus, it is natural to wonder whether there is complexity theoretic evidence that solving these problems significantly faster would be difficult. 
Researchers have presented some evidence in this vein, in the form of reductions from the conjectured hard problem of detecting  cliques in graphs.

Let $k = \Theta(1)$ be a positive integer.
A \emph{$k$-clique} is a collection of $k$ mutually adjacent vertices in a graph.
In the \textsf{$k$-Clique} problem,\footnote{This problem is sometimes referred to in the literature as \textsf{$k$-Multicolored Clique}. A folklore argument reduces from detecting a $k$-clique in an arbitrary $n$-node graph to the \textsf{$k$-Clique} problem as defined here, by making $k$ copies of the input graph, and only including edges between different copies, e.g. as in  \cite[Proof of Theorem 13.7]{parameterized-complexity-book}. }
 we are given a $k$-partite graph $G$ with vertex set $V_1\sqcup \cdots \sqcup V_k$, where each part $V_i$ has $n$ vertices,
and
are tasked with determining if $G$ contains a $k$-clique.

We can of course solve \textsf{$k$-Clique} in $O(n^k)$ time, just by trying out all possible $k$-tuples of vertices.
Better algorithms for \textsf{$k$-Clique} are known, which employ fast matrix multiplication.
Let $\omega$ denote the exponent of matrix multiplication (i.e., $\omega$ is the smallest real such that two $n\times n$ matrices can be multiplied in $n^{\omega+o(1)}$ time).
Given positive reals $a,b,c$, we more generally write $\omega(a,b,c)$ to denote 
the smallest real such that we can compute the product of an $n^a\times n^b$ matrix and an $n^b\times n^c$ matrix in $n^{\omega(a,b,c)+o(1)}$ time.
Then it is known that \textsf{$k$-Clique} can be solved in \[C(n,k) = \Theta(n^{\omega(\lfloor k/3\rfloor, \lceil (k-1)/3\rceil, \lceil k/3\rceil)})\]
time \cite{EisenbrandGrandoni2004}.
% \cite{nevsetvril1985complexity}
% When $3\mid k$, the above expression simplifies to $C(n,k) = \Theta(n^{\omega k/3})$.
% If $\omega = 2$, then $C(n,k) = O(n^{2k/3})$ for all $k\ge 3$ with $3\mid k$.
The current fastest matrix multiplication algorithms yield $\omega < 2.37156$ \cite{even-faster-matrix-mult}.
A popular fine-grained hardness hypothesis posits (e.g., in  \cite{williams2018some,dalirrooyfard2022induced}) that current algorithms for \textsf{$k$-Clique} are optimal.

\begin{hypothesis}[\textsf{$k$-Clique Hypothesis}] For any integer constant $k \geq 3$,  solving \textsf{$k$-Clique} requires at least $C(n,k)^{1-o(1)}$ time.
\end{hypothesis}

%This hypothesis is used as the basis for conditional hardness results in fine-grained complexity.
In this context, previous work  provided reductions from \textsf{$k$-Clique} to disjoint path problems. 
For example,
\cite{Geometric-Lens-conf} presents a reduction from \textsf{$k$-Clique} to \textsf{$2k$-DSP} on undirected graphs with $O((kn)^2)$ vertices (and this reduction easily extends to DAGs).
Our first conditional lower bound improves this result for DAGs, by halving the number of paths needed in the reduction.
% \nicole{Should we combine the following two into one theorem statement?}

\begin{restatable}{theorem}{lbdspboth}
    \label{lbdspboth}
    There is a reduction from \textsf{$k$-Clique}  to \textsf{$k$-DSP} on unweighted DAGs with $O((kn)^2)$ vertices, that runs in $O((kn)^2)$ time.
\end{restatable}

% \begin{restatable}{theorem}{lbdsp}
%     \label{lbdsp}
%     There is a reduction from \textsf{$k$-Clique} on a $kn$-vertex graph to \textsf{$k$-DSP} on a DAG with $O((kn)^2)$ vertices, that runs in $O((kn)^2)$ time.
% \end{restatable}

%A small modification to this reduction shows hardness for \textsf{$k$-DSP} in undirected graphs as well. 

% \begin{restatable}{theorem}{lbdspundir}
%     \label{lbdspundir}
%     There is a reduction from \textsf{$k$-Clique} on a $kn$-vertex graph to \textsf{$k$-DSP} on an undirected graph with $O((kn)^2)$ vertices, that runs in $O((kn)^2)$ time.
% \end{restatable}

% One can plug these values into the \textsf{$k$-Clique} Hypothesis to obtain a complicated conditional lower bound statement; for simplicity, we state a corollary for the case of $3\mid k$.

\begin{corollary}
\label{cor:dsp}
    Assuming the \textsf{$k$-Clique Hypothesis}, \textsf{$k$-DSP} requires $C(n^{1/2},k)^{1-o(1)}$ time to solve on unweighted DAGs.
\end{corollary}

The previous reduction of \cite{Geometric-Lens-conf}  yields a  weaker bound of $C(n^{1/2},\lfloor k/2\rfloor )^{1-o(1)}$ for the time needed to solve \textsf{$k$-DSP}, assuming the \textsf{$k$-Clique} Hypothesis.
If $\omega > 2$, this earlier result only gives nontrivial (that is, superquadratic) lower bounds for $k\ge 10$, and if $\omega  = 2$ is only nontrivial for $k\ge 14$.
In comparison, if $\omega > 2$, \Cref{cor:dsp} is nontrivial for $k\ge 5$, and if $\omega = 2$, \Cref{cor:dsp} is still nontrivial for $k\ge 7$.
See \Cref{table:DSP-small-values} for the concrete lower bounds we achieve for small $k$.

As mentioned before (and as shown in detail in \Cref{app:unnecessary}), the \textsf{$k$-DSP} problem in weighted DAGs generalizes \textsf{$k$-DP} in DAGs.
However, the current fastest algorithms for \textsf{$k$-DP} have the same time complexity as the current best algorithms for the more general \textsf{$k$-DSP} problem.
To explain this behavior, it is desirable to show conditional lower bounds for \textsf{$k$-DP} in DAGs, which are similar in quality to the known lower bounds for \textsf{$k$-DSP} in DAGs.

Such lower bounds have been shown by \cite{Chitnis2021}.
In particular, \cite{Chitnis2021} together with standard reductions in parameterized complexity \cite[Proofs of Theorems 14.28 and 14.30]{parameterized-complexity-book}
implies that there is a reduction from \textsf{$k$-Clique} to  \textsf{$8k$-EDSP} on graphs with $O((kn)^4)$ nodes.
One can easily modify this reduction, using the idea in the construction from \cite[Section 6]{Geometric-Lens}, to instead reduce from \textsf{$k$-Clique} to \textsf{$8k$-DSP} on graphs with $O((kn)^4)$ nodes.

This reduction implies that \textsf{$k$-DSP} requires $C(n^{1/4},\lfloor k/8\rfloor)^{1-o(1)}$ time to solve on DAGs, assuming the \textsf{$k$-Clique Hypothesis}.
For large $k$, this is the current best  conditional lower bound for \textsf{$k$-DP} in DAGs.
However, the blow-up in the graph size and parameter value in this reduction makes this lower bound irrelevant for small  $k$
(in fact, the reduction only yields nontrivial lower bounds under the \textsf{$k$-Clique} Hypothesis for $k\ge 96$).

For small values of $k$, better conditional lower bounds are known for \textsf{$k$-DP}.
In particular, \cite{Slivkins2010} presents reductions from \textsf{$k$-Clique} to \textsf{$p$-DP} and \textsf{$p$-DSP} on DAGs with $O(kn)$ vertices, for parameter value $p = k + \binom{k}{2}$.
For our final conditional lower bound, we improve this reduction, by reducing the number of paths needed to $k+ \lfloor k^2/4\rfloor$.

\begin{restatable}{theorem}{lbdp}
    \label{lbdp}
    Let $k\ge 3$ be a constant integer, and set $p=k+ \lfloor k^2/4\rfloor$.
    There are $O((kn)^2)$ time reductions from \textsf{$k$-Clique} to \textsf{$p$-DP} and \textsf{$p$-DSP} on unweighted DAGs with $O(kn)$ vertices.
\end{restatable}

\begin{table}[t]

    \centering
    \begin{tabular}{|c|c|c|}
    \hline
         \multirow{2}{*}{$k$} & \multicolumn{2}{|c|}{\textsf{$k$-DSP} Exponent Lower Bound}  \\ % \cline{2-3}
         & (for current $\omega$) & (if $\omega = 2)$ \\
         \hline
         $5$ & $2.042$ & Trivial \\

         $6$ &  $2.371$ & Trivial \\
         $7$ & $2.794$ & $2.5$ \\
         $8$ & $3.198$ & $3$ \\
         $9$ & $3.557$ & $3$ \\
         \hline
    \end{tabular}
    \caption{A list of lower bounds implied by \Cref{cor:dsp} for \textsf{$k$-DSP} when $5\le k\le 9$. 
    Each row corresponds to a value of $k$.
    An entry of $\alpha$ in the left column of the row for a given $k$ value indicates that solving \textsf{$k$-DSP} in $O(n^{\alpha-\delta})$ time for any constant $\delta > 0$ would require refuting the \textsf{$k$-Clique Hypothesis} or designing faster matrix multiplication algorithms. 
    An entry of $\beta$ in the right column in the row for a given $k$ value indicates that assuming the \textsf{$k$-Clique Hypothesis}, \textsf{$k$-DSP} requires $n^{\beta - o(1)}$ time to solve.
    \textbf{The previous reduction of \cite{Geometric-Lens-conf} gave no nontrivial lower bound for \textsf{$\boldsymbol{k}$-DSP} for any value of $\boldsymbol{k}$  in this table}, and the reduction of \cite{Slivkins2010} matches our lower bound for $k=6$, but is  worse everywhere else.
    Table entry values are based off rectangular matrix multiplication exponents from \cite[Table 1]{even-faster-matrix-mult}.
    }
\label{table:DSP-small-values}
\end{table}

For each integer $p\ge 5$, we can find the largest integer $k\ge 3$ such that $k + \lfloor k^2/4\rfloor \le p$, and then apply \Cref{lbdp} to obtain conditional lower bounds for \textsf{$p$-DP} and \textsf{$p$-DSP} on DAGs.

\begin{corollary}\label{cor:dp}
Assuming the \textsf{$k$-Clique Hypothesis}, the \textsf{$p$-DSP} and \textsf{$p$-DP} problems require 
\[\max\left(C(n,k_{\text{even}}(p)), C(n,k_{\text{odd}}(p))\right)^{1-o(1)}\]
time to solve on unweighted DAGs for all integers $p\ge 5$, where 
    \[k_{\text{even}}(p) = 2\lfloor\sqrt{p+1}\rfloor -2\]
and
    \[k_{\text{odd}}(p) = 2\left\lfloor\frac{\sqrt{p+5}-1}{2}\right\rfloor - 1\]
    are the largest even and odd integers $k$ such that $k + \lfloor k^2/4\rfloor \le p$ respectively.
\end{corollary}

% \Cref{cor:dp} implies better lower bounds for \textsf{$k$-DSP} on DAGs than \Cref{cor:dsp} for certain small values of $k$.
% In particular, assuming 
Assuming the \textsf{$k$-Clique Hypothesis}, \Cref{cor:dp} shows that \textsf{5-DSP} requires at least $n^{\omega - o(1)}$ time and \textsf{8-DSP} requires at least $C(n,4)^{1-o(1)}$ time to solve.
For the current value of $\omega$, these yield lower bounds of $n^{2.371-o(1)}$ for \textsf{5-DSP} and $n^{3.198-o(1)}$ for \textsf{8-DSP}, which are better than the lower bounds implied by \Cref{cor:dsp} (see \Cref{table:DSP-small-values}).
If $\omega = 2$ however, \Cref{cor:dp} does not yield better lower bounds than \Cref{cor:dsp} for \textsf{$k$-DSP}.

Previous reductions give nontrivial lower bounds for \textsf{$p$-DP} only when $p\ge 6$ if $\omega > 2$, and  $p\ge 10$ if $\omega = 2$.
In comparison,
\Cref{cor:dp} yields nontrivial lower bounds for \textsf{$p$-DP} under the \textsf{$k$-Clique Hypothesis} for $p\ge 5$ if $\omega > 2$, and $p\ge 8$ if $\omega = 2$.

Previously, the reduction of \cite{Slivkins2010} yielded the best lower bounds for \textsf{$p$-DP} for $p\le 2016$, and otherwise the reduction of \cite{Chitnis2021} yielded better lower bounds.
In comparison, \Cref{cor:dp} yields lower bounds matching the reduction from \cite{Slivkins2010}  for $p\in\set{6,7,10}$, and  otherwise, for $\omega > 2$, yields strictly better lower bounds for \textsf{$p$-DP} for all $p \ge 5$.
Moreover, for $\omega = 2$,  \Cref{cor:dp} yields the best lower bounds for \textsf{$p$-DP} for all $p\le 4031$ (with \cite{Chitnis2021} yielding better lower bounds only for larger $p$). 

To see quantitatively how \Cref{cor:dp} improves the best conditional lower bounds for \textsf{$p$-DP} from previous work at various concrete values of $p$, see \Cref{table:dp-various-values}.

\begin{table}[t]
    \centering
    \begin{tabular}{|c|c|c|c|}
    \hline
         \multirow{2}{*}{$p$} & \multicolumn{3}{|c|}{\textsf{$p$-DP} Exponent Lower Bound (if $\omega = 2)$} \\ %\cline{2-4}
         &  From \Cref{cor:dp} & Reduction of \cite{Slivkins2010} & Reduction of \cite{Chitnis2021} \\
         \hline
         $9$ & $3$ & Trivial & Trivial \\ %k=4
         $24$ & $6$ & $4$ & Trivial \\ %k=8, 6, 
         $89$ & $12$ & $8$ & Trivial \\ %k=17, 12
         $239$ & $20$ & $14$ & $5$ \\ %k = 29, 21, 30 
         $929$ & $40$ & $28$ & $19.5$ \\ %k=59, 42, 117
         $2016$ & $58$ & $42$ & $42$ \\ %k=87, 63, 252
         $2969$ & $72$ & $51$ & $62$\\ %k=107, 76, 371
         $4031$ & $84 $& $60$ & $84$ \\ %k=125, 89, 503
         \hline
    \end{tabular}
    \caption{A list of lower bounds implied by \Cref{cor:dp} (and previous work) for \textsf{$p$-DP} at various values of $p$.  Rows correspond to values of $p$. For a given such row, the entries  $\alpha, \beta, \gamma$ in the three columns collected under the heading of ``\textsf{$p$-DP} Exponent Lower Bound,'' read from left to right, indicate that \Cref{cor:dp}, the reduction of \cite{Slivkins2010}, and the reduction of \cite{Chitnis2021} imply that \textsf{$p$-DP} requires $n^{\alpha - o(1)}$, $n^{\beta  -o(1)}$, and $n^{\gamma - o(1)}$ time to solve respectively, assuming the \textsf{$k$-Clique Conecture}.}
    \label{table:dp-various-values}
\end{table}

% \subsection{Limitations of Algebraic Algorithms}
% \label{subsec:limit}

% Our algorithms for \textsf{2-DSP} in undirected graphs and DAGs work by checking that a certain polynomial, whose monomials correspond uniquely to pairs of disjoint shortest paths in the input graph, is nonzero.
% Compared to previous work, our approach has two main limitations.

% \subsubsection*{Deterministic versus Randomized}

% To check if the polynomial encoding pairs of disjoint shortest paths is nonzero, rather than explicitly construct the polynomial (which would be far too slow), we efficiently evaluate the polynomial at a uniform random point of some sufficiently large ground field.
% This leads to a randomized algorithm, while  previous algorithms for \textsf{2-DSP} in general undirected graphs and DAGs are deterministic.
% This reliance on randomness is common for algebraic graph algorithms.

\subsection{Comparison with Previous Algorithms}
\label{subsec:prev-alg}

% \nicole{How about an introductory sentence of two to restate that this is the first work to bring algebraic algorithms into these problems. But algebraic algorithms have been used for some related problems, and we describe some examples, and compare our work to them.} 

Previous algorithms for \textsf{2-DSP} and \textsf{2-DP} in DAGs are combinatorial in nature: they observe certain structural properties of candidate solutions, and then leverage these observations to build up pairs of disjoint paths.
In the special case of \emph{unweighted} undirected graphs, \cite{BjorklundHusfeldtKaski2022} presented an algebraic algorithm for solving a generalization of \textsf{2-DSP}, but all other prior algorithms for \textsf{2-DSP} and \textsf{2-DP} in undirected graphs are combinatorial.
Our work is the first to employ algebraic methods to tackle the general weighted \textsf{2-DSP} problem: our algorithms for \textsf{2-DSP} on undirected graphs and DAGs work by checking that a certain polynomial, whose monomials correspond uniquely to pairs of disjoint shortest paths in the input graph, is nonzero. 
To obtain the fast runtimes in \Cref{2dsp-undir,2dsp-DAG}, we evaluate this polynomial over a field of characteristic two, and crucially exploit certain symmetries 
% in this polynomial 
which make efficient evaluation possible when working modulo two.

Such ``mod 2 vanishing'' methods have appeared previously in the literature for algebraic graph algorithms, but the symmetries we exploit in our algorithms for \textsf{2-DSP} differ in interesting ways from those of previous approaches. 
For example, previous methods tend to work exclusively in undirected graphs (relying on the ability to traverse cycles in both the forwards and backwards directions to produce terms in polynomials which cancel modulo 2), while our approach is able to handle \textsf{2-DSP} in both undirected graphs and DAGs.
It is also interesting that our algorithms solve \textsf{2-DSP} in \emph{weighted} graphs without any issue, since the previous algebraic graph algorithms we are aware of are efficient in unweighted graphs, but in weighted graphs have a runtime which depends polynomially on the value of the maximum edge weight.

Below, we compare our techniques to previous algebraic algorithms  in the literature.

\paragraph{Two Disjoint Paths with Minimum Total Length}

The most relevant examples of algebraic graph algorithms in the literature to our work are previous algorithms for the \textsf{MinSum 2-DP} problem:
in this problem, we are given a graph $G$ on $n$ vertices, with specified sources $s_1,s_2$ and targets $t_1,t_2$, and are tasked with finding internally vertex-disjoint paths $P_i$ from $s_i$ to $t_i$, such that the sum of the lengths of $P_1$ and $P_2$ is minimized, or reporting that no such paths exists.
% This problem generalizes the \textsf{2-DSP} problem.

In \emph{unweighted} undirected graphs, \cite{BjorklundHusfeldt-SDP}  showed that \textsf{MinSum 2-DP} can be solved in 
polynomial time, with 
\cite[Section 6]{BjorklundHusfeldtKaski2022} providing a faster implementation of this approach running in  $\tilde{O}(n^{4+\omega})$ time.
Similar to our work, these algorithms
 check if a certain polynomial enumerating disjoint pairs of paths in $G$ is nonzero or not.
These methods rely on $G$ being undirected, and are based off computing determinants of $n\times n$ matrices.

Our approach for \textsf{2-DSP} differs from these arguments because we seek linear time algorithms, and so \textbf{avoid computing determinants} (which would yield $\Omega(n^\omega)$ runtimes).
We instead directly enumerate pairs of intersecting paths and subtract them out.
This alternate approach also allows us to obtain algorithms which apply to both  undirected graphs and DAGs, whereas the cycle-reversing arguments of \cite{BjorklundHusfeldtKaski2022} do not appear to extend to DAGs.

\paragraph{Paths and Linkages with Satisfying Length Conditions}
Given sets $S$ and $T$ of $p$ source and target vertices respectively, an $(S,T)$-linkage is a set of $p$ vertex-disjoint paths, beginning at different nodes in $S$ and ending at different nodes in $T$.
The length of such a linkage is the sum of the lengths of the paths it contains.
Recent work has presented algorithms for the problem of finding $(S,T)$-linkages in undirected graphs of length at least $k$, fixed-parameter tractable in $k$.
In particular, \cite[Section 4]{FPT-Max-Colored-Path} presents an algorithm solving this problem in $2^{k+p}\poly(n)$ time.
Their algorithm enumerates collections of $p$ walks beginning at different nodes in $S$ and ending at different nodes in $T$.
They then argue that all terms in this enumeration with intersecting walks cancel modulo 2, leaving only the  $(S,T)$-linkages.
One idea used in the above cancellation argument is that if two paths $P$ and $Q$ in a collection intersect at a vertex $v$, then we can pair this collection with a new collection obtained by swapping the suffixes of $P$ and $Q$ after vertex $v$.

In the \textsf{2-DSP} problem, solution paths must connect sources $s_i$ to corresponding targets $t_i$ instead of to arbitrary targets, and so we cannot use the above suffix-swapping argument to get cancellation.
So to enumerate disjoint shortest paths in our algorithms, we employ somewhat trickier cancellation arguments than what was previously used.

More recently, \cite[Section 6]{Det-Sieving} presented an algorithm solving the linkage problem discussed above in $2^k\poly(n)$ time (with runtime independent of $p$).
Their approach uses determinants to enumerate $(S,T)$-linkages.
As mentioned previously, we explicitly avoid using determinants so that we can obtain linear time algorithms.

\paragraph{Additional Related Work}
There are many additional examples of algebraic graph algorithms in the literature.
For example, \cite{BjorklundHusfeldtTaslaman} presents an efficient algorithm for finding shortest cycles through specified subsets of vertices, \cite{CyganGabowSankowski2015} presents algorithms for finding shortest cycles and perfect matchings in essentially matrix multiplication time, and \cite{BjorklundHusfeldtKaski2022} presents a polynomial time algorithm for finding a shortest cycle of even length in a directed graph.
Even more examples of algebraic methods in  parameterized algorithms are listed in \cite[Table 1]{Det-Sieving}.

\paragraph{Bibliographic Remark}
While the current paper was under submission, the work \cite{BentertFominGolovach2024}
 was posted online.
The reduction used to establish \cite[Theorem 1]{BentertFominGolovach2024} is essentially the same as the reduction we use to prove \Cref{lbdspboth},
so this result was independently shown by \cite{BentertFominGolovach2024}.

\subsection*{Organization}
In \Cref{sec:prelim} we introduce notation and recall useful facts about graphs and polynomials used in our results. 
In \Cref{sec:overview} we provide some intuition for the proofs of our results.
In \Cref{sec:2dsp} we present our linear time algorithms for \textsf{2-DSP} in weighted undirected graphs and DAGs. In \cref{sec:edge} we present our algorithm for \textsf{$k$-EDSP}. In \cref{sec:lb} we present our lower bounds.
Finally, we conclude in \Cref{sec:conclusion} by highlighting some open problems motivated by this work.

See \Cref{app:unnecessary} for proofs of some simple reductions between variants of disjoint path problems which were mentioned previously.

\section{Preliminaries}
\label{sec:prelim}

\subsubsection*{General Notation}

Given a positive integer $a$, we let $[a] = \set{1, \dots, a}$ denote the set of the first  $a$ positive integers.
Given positive integers $a$ and $b$, we let $[a,b] = \set{a, \dots, b}$ denote the set of consecutive integers from $a$ to $b$ inclusive (if $a > b$, then $[a,b]$ is the empty set).

Throughout, we let $k$ denote a constant positive integer parameter.

\subsubsection*{Graph Notation and Assumptions}

Throughout, we let $G$ denote the input graph on $n$ vertices and $m$ edges.
We let $s_1, \dots, s_k$ denote the source vertices of $G$, and $t_1, \dots, t_k$ denote the target vertices of $G$.
A \emph{terminal} is a source or target node.
We assume without loss of generality that $G$ is weakly connected (we can do this because we only care about solving disjoint path problems on $G$, and if terminals of $G$ are in separate weakly connected components, we can solve smaller disjoint path problems on each component separately).

Given an edge $e = (u,v)$, we let $\ell(u,v)$ denote the weight of $e$ in $G$.
We assume all edge weights are positive.
We let $\dist(u,v)$ denote the distance of a shortest path (i.e., the sum of the weights of the edges used in a shortest path) from $u$ to $v$. 
When we write ``path $P$ traverses edge $(u,v)$'' we mean that $P$ first enters $u$, then immediately goes to $v$.

We represent paths $P = \langle v_0, \dots, v_r\rangle$ as sequences of vertices.
If the path $P$ passes through vertices $u$ and $v$ in that order, we let $P[u,v]$ denote the subpath of $P$ which begins at $u$ and ends at $v$. 
We let $\rev{P}$ denote the \emph{reverse path} of $P$, which traverses the vertices of $P$ in reverse order.
Given two paths $P$ and $Q$ such that the final vertex of $P$ is the same as the first vertex of $Q$, we let $P\dia Q$ denote the concatenation of $P$ and $Q$.

\subsubsection*{Shortest Path DAGs}

Given a graph $G$ and specified vertex $s$, the \emph{$s$-shortest paths DAG of $G$} is the graph with the same vertex set as $G$, which includes edge $(u,v)$ if and only if $(u,v)$ is an edge traversed by an $(s,v)$-shortest path in $G$.
From this definition, it is easy to see that a sequence of vertices is an $(s,v)$-shortest path of $G$ if and only if it is an $(s,v)$-path in the $s$-shortest paths DAG of $G$.
Indeed, every edge of an $(s,v)$-shortest path in $G$ is contained in the $s$-shortest paths DAG by definition, and so forms a path in this graph.
Conversely, if the sequence of vertices $P = \langle v_0, \dots, v_r\rangle$ is an $(s,v)$-path in the $s$-shortest paths DAG of $G$, then we can inductively show that $P[s,v_i]$ is a shortest path in $G$ for each index $i$.

We observe that shortest paths DAGs can be constructed in linear time.

\begin{proposition}[Shortest Path DAGs]
    \label{prop:sp-DAG}
    Let $G$ be a weighted DAG or undirected graph with distinguished vertex $s$.
    Then we can construct the $s$-shortest paths DAG of $G$ in linear time.
\end{proposition}
\begin{proof}
    By definition, an edge $(u,v)$ is in the $s$-shortest paths DAG of $G$ if and only if $(u,v)$ is the last edge of some  $(s,v)$-shortest path in $G$.
    This is equivalent to the condition that $(u,v)$ is an edge in $G$, and 
        \begin{equation}
        \label{eq:dsv}
        \dist(s,v) = \dist(s,u) + \ell(u,v).
        \end{equation}
    So, we can construct the $s$-shortest paths DAG of $G$ by computing the values of $\dist(s,v)$ for all vertices $v$, and then going through each edge $(u,v)$ in $G$ (if $G$ is undirected, we try out both ordered pairs $(u,v)$ and $(v,u)$ of an edge $\set{u,v}$) and checking if \cref{eq:dsv} holds.

    So to prove the claim, it suffices to compute $\dist(s,v)$ for all vertices $v$ in linear time.

    When $G$ is a weighted DAG, we can compute a topological order of $G$ in linear time, and then perform dynamic programming over the vertices in this order to compute $\dist(s,v)$ for all vertices $v$ in linear time (this procedure is just a modified breadth-first search routine).

    When $G$ is an undirected graph, we instead use Thorup's linear-time algorithm for single-source shortest paths in weighted undirected graphs \cite{Thorup1997} to compute $\dist(s,v)$ for all vertices $v$.
\end{proof}

\subsubsection*{Finite Fields}

Our algorithms for \textsf{2-DSP} in undirected graphs and DAGs involve working over a finite field $\mathbb{F}_{2^q}$ of characteristic two, where $q = O(\log n)$.
We work in the Word-RAM model with words of size $O(\log n)$, so that 
addition and multiplication over this field take constant time.

We make use of the following classical result, which shows that we can test if a polynomial is nonzero by evaluating it at a random point of a sufficiently large finite field. 

\begin{proposition}[Schwartz-Zippel Lemma]
\label{obvious}
    Let $f$ be a nonzero polynomial of degree at most $d$.
    Then a uniform random evaluation of $f$ over $\mathbb{F}$ is nonzero with probability at least $1-d/|\mathbb{F}|$.
\end{proposition}

% \subsubsection*{Multicolored Clique}

% For both of our lower bounds we will reduce from \textsf{$k$-Clique} on $kn$ vertices. In this problem the input is a $k$-vertex-colored graph, and the output is whether or not there is a $k$-clique with one vertex of each color. There is a standard reduction from $k$-\textsf{Clique} on $n$ vertices to \textsf{$k$-Clique} on $kn$ vertices: Copy the vertex set so that there are $k$ copies of each vertex, and color each copy a different color. Between vertices of the same color, there is no edge, and between vertices of different colors there is an edge if and only if there was an edge between these two vertices in the original graph.

\section{Technical Overview}
\label{sec:overview}

\subsection{2-DSP Algorithms}

We first outline a linear time algorithm solving \textsf{2-DP} in DAGs.
We then discuss the changes needed to solve the \textsf{2-DSP} problem in weighted DAGs, and then the additional ideas used to solve \textsf{2-DSP} in weighted undirected graphs.

Let $G$ be the input DAG.
For each edge $(u,v)$ in $G$, we introduce an indeterminate $x_{uv}$.
We assign each pair of paths in $G$ a certain monomial over the $x_{uv}$ variables, which records the pairs of consecutive vertices traversed by the paths.
These monomials are constructed so that any pair of disjoint paths has a unique monomial.

Let $F$ be the sum of monomials corresponding to all pairs of paths $\pair{P_1,P_2}$ such that $P_i$ is an $(s_i,t_i)$-path in $G$.
Let $F_{\disj}$ and $F_{\cap}$ be the sums of monomials corresponding to all such pairs of paths which are disjoint and intersecting respectively.
Since each disjoint pair of paths produces a distinct monomial, we can solve \textsf{2-DP} by testing whether $F_{\disj}$ is a nonzero polynomial. 
We can perform this test by evaluating $F_{\disj}$ at a random point, by the Schwartz-Zippel lemma (\Cref{obvious}).

Since every pair of paths is either disjoint or intersecting, we have 
    \[F = F_{\disj} + F_{\cap}\]
which implies that 
    \[F_{\disj} = F - F_{\cap}.\]
This relationship is pictured in \Cref{fig:disj-to-intersect}.

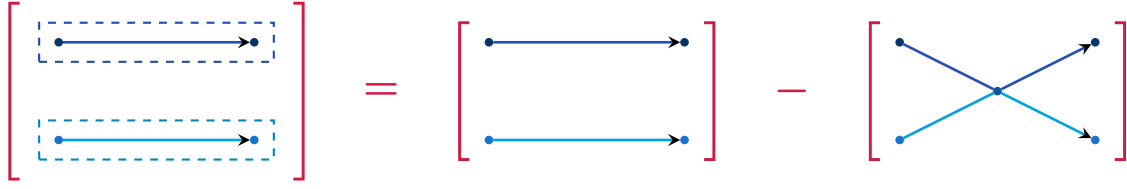
\begin{figure}[t]
    \centering
    \begin{tikzpicture}[scale = 1.3,
lightvtx/.style=
{draw=dblue!70!midnight, fill=dblue!70!midnight,
circle, inner sep=1pt, minimum width=1pt},
darkvtx/.style=
{draw=cdblue!5!midnight, fill=cdblue!5!midnight, circle, inner sep=1pt, minimum width=1pt},
intervtx/.style=
{draw=dblue!35!midnight, fill=dblue!35!midnight, circle, inner sep = 1pt, minimum width = 1pt},
lightedge/.style=
{draw=dsblue!80!midnight, line width = 1pt, -stealth},
lightedgenoarrow/.style=
{draw=dsblue!80!midnight, line width = 1pt},
darkedge/.style=
{draw=cdblue!90!midnight, line width = 1pt, -stealth},
darkedgenoarrow/.style=
{draw=cdblue!90!midnight, line width = 1pt},
lightborder/.style=
{draw=dsblue!60!midnight, thick, dashed},
darkborder/.style=
{draw=cdblue!80!midnight, thick, dashed}]

    %% Vertical Displacement
\def\yspace{1cm};
\def\yborder{0.2cm};
\def\ybigparen{0.2cm};
\def\yparen{0.2cm};

%% Horizontal Displacement
\def\xstep{2cm};
\def\xborder{0.2cm};
\def\xparen{0.2cm};

\def\extr{0.8cm};

% Horizontal Transition
\def\extr{0.9cm};

\node[darkvtx] (s1) at (0,\yspace) {};
\node[lightvtx] (s2) at (0,0) {};

\node[darkvtx] (t1) at (\xstep, \yspace) {};
\node[lightvtx] (t2) at (\xstep, 0) {};

\draw[darkedge] (s1) -- (t1);
\draw[lightedge] (s2) -- (t2);

\draw[darkborder] (-\xborder, \yspace - \yborder) rectangle (\xborder + \xstep, \yspace + \yborder);

\draw[lightborder] (-\xborder, - \yborder) rectangle (\xborder + \xstep, \yborder);

    \draw [alizarin, line width = 1.1pt] (-\xborder-\xparen,-\yborder-\ybigparen) to [square left brace ] (-\xborder-\xparen,\yspace + \yborder + \ybigparen);

    \draw [alizarin, line width = 1.1pt] (\xborder+\xstep+\xparen,-\yborder-\ybigparen) to [square right brace ] (\xborder+\xstep+\xparen,\yspace + \yborder + \ybigparen);

%% Equation Transition
\node at (\xstep + + \xborder + \xparen + \extr, 0.5*\yspace) 
{\scalebox{1.4}{\textcolor{alizarin}{$\boldsymbol{=}$}}};

\begin{scope}[xshift = \xstep + \xborder + \xparen + 2*\extr + \xparen]
    \node[darkvtx] (s1) at (0,\yspace) {};
\node[lightvtx] (s2) at (0,0) {};

\node[darkvtx] (t1) at (\xstep, \yspace) {};
\node[lightvtx] (t2) at (\xstep, 0) {};

\draw[darkedge] (s1) -- (t1);
\draw[lightedge] (s2) -- (t2);

    \draw [alizarin, line width = 1.1pt] (-\xparen,-\yparen) to [square left brace ] (-\xparen,\yspace + \yparen);

    \draw [alizarin, line width = 1.1pt] (\xstep+\xparen,-\yparen) to [square right brace ] (\xstep+\xparen,\yspace + \yparen);
\end{scope}

%% Subtraction 
\begin{scope}[xshift = \xstep + \xborder + \xparen + 2*\extr + \xparen + \xstep + \xparen]
\node at (\extr, 0.5*\yspace) 
{\scalebox{1.4}{\textcolor{alizarin}{$\boldsymbol{-}$}}};
\end{scope}

\begin{scope}[xshift = \xstep + \xborder + \xparen + 2*\extr + \xparen + \xstep + \xparen + 2*\extr + \xparen ]

    \node[darkvtx] (s1) at (0,\yspace) {};
\node[lightvtx] (s2) at (0,0) {};

\node[intervtx] (a) at (0.5*\xstep, 0.5*\yspace) {};

\node[darkvtx] (t1) at (\xstep, \yspace) {};
\node[lightvtx] (t2) at (\xstep, 0) {};

\draw[darkedge] (s1) -- (a) -- (t1);

\draw[lightedge] (s2) -- (a) -- (t2);

    \draw [alizarin, line width = 1.1pt] (-\xparen,-\yparen) to [square left brace ] (-\xparen,\yspace + \yparen);

    \draw [alizarin, line width = 1.1pt] (\xstep+\xparen,-\yparen) to [square right brace ] (\xstep+\xparen,\yspace + \yparen);
\end{scope}
    
\end{tikzpicture}
    \caption{To enumerate the family of disjoint pairs of paths on the left (the dashed borders around the paths indicate that the paths do not intersect), it suffices to enumerate all pairs of paths and subtract out those pairs in the family which intersect at some point.}
    \label{fig:disj-to-intersect}
\end{figure}

Thus, in order to evaluate $F_{\disj}$, it suffices to evaluate $F$ and $F_{\cap}$.
Since $F$ enumerates  pairs of paths from the sources to their corresponding targets with no constraints, it turns out that $F$ is easy to evaluate.
So solving \textsf{2-DP} amounts to evaluating $F_{\cap}$ efficiently.

To evaluate $F_{\cap}$, we need a way of enumerating over all pairs of intersecting paths. 
Each pair of intersecting paths overlaps at a unique earliest vertex $v$ (with respect to the topological order of $G$).
Consequently, if we let $F_v$ be the sum of monomials of pairs of intersecting paths with first intersection at $v$, we have 
    \begin{equation}
    \label{overview-disj}
    F_{\cap} = \sum_{v\in V} F_v
    \end{equation}
as depicted in \Cref{fig:intersection-casework}.

We evaluate each $F_v$ by relating it to a seemingly simpler polynomial.
Let $\tilde{F}_v$ be the polynomial enumerating pairs of paths $\pair{P_1,P_2}$ where $P_i$ is an $(s_i,t_i)$-path in $G$ such that
\begin{enumerate}
    \item $P_1$ and $P_2$ intersect at vertex $v$, and 
    \item the vertices appearing immediately before $v$ on $P_1$ and $P_2$ are distinct.
\end{enumerate}

We can think of property 2
% \nicole{should this say property 2?} 
as a relaxation of the condition that $P_1$ and $P_2$ have $v$ as their earliest intersection point: instead of requiring that $P_1[s_1,v]$ and $P_2[s_2,v]$ never overlap before $v$, we merely require that these subpaths do not overlap at the  position \emph{immediately} before $v$.
It turns out evaluating $\tilde{F}_v$ is easy, because we can enforce property 2 above by enumerating over all pairs of paths which intersect at $v$, and then subtracting out all such pairs  which overlap at some edge ending at $v$.
Simultaneously evaluating all $\tilde{F}_v$ can then be done in $O(m)$ time, roughly because we perform one subtraction for each possible edge the paths could overlap at.

So far, we have explained how to compute all $\tilde{F}_v$ values in linear time.
Now comes the key idea behind our algorithm: over fields of characteristic two,  the polynomials $\tilde{F}_v$ and $F_v$ are actually identical! 
Indeed, consider a pair of paths $\pair{P_1, P_2}$ enumerated by $\tilde{F}_v$, which intersects before $v$.
Let the first intersection point of these paths be some vertex $u$.
Then by condition 2 above, the subpaths $P_1[u,v]$ and $P_2[u,v]$ are distinct, because their penultimate vertices are distinct.
So if we define new paths 
    \[Q_1 = P_1[s_1,u]\dia P_2[u,v]\dia P_1[v,t_1]
    \quad\text{and}\quad 
    Q_2 = P_2[s_2,u]\dia P_1[u,v]\dia P_2[v,t_2]\]
obtained by swapping the $u$ to $v$ subpaths in $P_1$ and $P_2$, we get a new pair of paths $\pair{Q_1,Q_2}$ satisfying conditions 1 and 2 from before, such that each $Q_i$ is an $(s_i,t_i)$-path in $G$, which produces the same monomial as $\pair{P_1,P_2}$.
This subpath swapping operation is depicted in \Cref{fig:subpath-swap}, for $u=a$ and $v=b$.
Then modulo two, the contributions of the pairs $\pair{P_1,P_2}$ and $\pair{Q_1,Q_2}$ to $\tilde{F}_v$ will cancel out.
It follows that all pairs of paths which intersect before $v$ have net zero contribution to  $\tilde{F}_v$, and so $\tilde{F}_v = F_v$ as claimed.
This congruence is depicted in \Cref{fig:intersetion-relax},

Given this observation, we can use our evaluations of $\tilde{F}_v$ in \cref{overview-disj} to evaluate $F_{\cap}$ and thus, by the previous discussion, solve the \textsf{2-DP} problem.

\subsubsection*{From Disjoint Paths to Disjoint Shortest Paths}

To solve \textsf{2-DSP} in weighted DAGS, we can modify the \textsf{2-DP} algorithm sketched above as follows.
First, for $i\in [2]$, we compute $G_i$, the $s_i$-shortest paths DAG of $G$.
We then construct polynomials as above, but with the additional constraint that they only enumerate pairs of paths $\langle P_1,P_2\rangle$ with the property that every edge in path $P_i$ lies in $G_i$.
This ensures that we only enumerate pairs of paths which are shortest paths between their terminals.

With this change, the above algorithm for \textsf{2-DP} generalizes to solving \textsf{2-DSP}.

\begin{figure}[t]
    \centering
    \begin{tikzpicture}[scale = 1.6,
lightvtx/.style=
{draw=dblue!70!midnight, fill=dblue!70!midnight,
circle, inner sep=1pt, minimum width=1pt},
darkvtx/.style=
{draw=cdblue!5!midnight, fill=cdblue!5!midnight, circle, inner sep=1pt, minimum width=1pt},
intervtx/.style=
{draw=dblue!35!midnight, fill=dblue!35!midnight, circle, inner sep = 1pt, minimum width = 1pt},
lightedge/.style=
{draw=dsblue!80!midnight, line width = 1pt, -stealth},
lightedgenoarrow/.style=
{draw=dsblue!80!midnight, line width = 1pt},
darkedge/.style=
{draw=cdblue!90!midnight, line width = 1pt, -stealth},
darkedgenoarrow/.style=
{draw=cdblue!90!midnight, line width = 1pt},
lightborder/.style=
{draw=dsblue!60!midnight, thick, dashed},
darkborder/.style=
{draw=cdblue!80!midnight, thick, dashed}]

    %% Vertical Displacement
\def\yspace{1cm};
\def\yborder{0.2cm};
\def\ybigparen{0.2cm};
\def\yparen{0.2cm};

\def\yparenextra{0.06cm};

%% Horizontal Displacement
\def\xstep{2*\yspace}; 
\def\xborder{0.2cm};
\def\xparen{0.2cm};

\def\extr{0.8cm};
\def\sumspace{0.7cm};

% Horizontal Transition
\def\extr{0.9cm};

\def\boxoffset{0.19cm};

    \node[darkvtx] (s1) at (0,\yspace) {};
\node[lightvtx] (s2) at (0,0) {};

\node[intervtx] (a) at (0.5*\xstep, 0.5*\yspace) {};

\node[darkvtx] (t1) at (\xstep, \yspace) {};
\node[lightvtx] (t2) at (\xstep, 0) {};

\draw[darkedge] (s1) -- (a) -- (t1);

\draw[lightedge] (s2) -- (a) -- (t2);

    \draw [alizarin, line width = 1.1pt] (-\xparen,-\yparen) to [square left brace ] (-\xparen,\yspace + \yparen);

    \draw [alizarin, line width = 1.1pt] (\xstep+\xparen,-\yparen) to [square right brace ] (\xstep+\xparen,\yspace + \yparen);

\begin{scope}[xshift = \xstep + \xparen]
    \node at (\extr, 0.5*\yspace) 
{\scalebox{1.5}{\textcolor{alizarin}{$\boldsymbol{=}$}}};
\end{scope}

\begin{scope}[xshift = \xstep + \xparen + 2*\extr]
    \node at (0, 0.4*\yspace) 
{\scalebox{1.5}{\textcolor{alizarin}{$\displaystyle\sum_{\textcolor{dblue!30!midnight}{v\in V}}$}}};
\end{scope}

\begin{scope}[xshift = \xstep + \xparen + 2*\extr + \sumspace + \xparen]

    \node[darkvtx] (s1) at (0,\yspace) {};
\node[lightvtx] (s2) at (0,0) {};

\node[intervtx] (a) at (0.5*\xstep, 0.5*\yspace) {};

\node[below=3pt] at (a) {\textcolor{dblue!30!midnight}{$v$}};

\node[darkvtx] (t1) at (\xstep, \yspace) {};
\node[lightvtx] (t2) at (\xstep, 0) {};

\draw[darkedge] (s1) -- (a) -- (t1);

\draw[lightedge] (s2) -- (a) -- (t2);

\node[darkvtx] (predarka) at (0.3*\xstep, \yspace - 0.3*\yspace) {};
\node[lightvtx] (prelighta) at (0.3*\xstep, 0.3*\yspace) {};

\draw[darkborder, rotate around={333.434948:(0.15*\xstep, \yspace - 0.15*\yspace)}] (-\boxoffset, \yspace - 0.15*\yspace - 0.5*\boxoffset) rectangle (\boxoffset + 0.3*\xstep, \yspace - 0.15*\yspace + 0.5*\boxoffset);

\draw[lightborder, rotate around={26.5650512:(0.15*\xstep, 0.15*\yspace)}] (-\boxoffset, 0.15*\yspace - 0.5*\boxoffset) rectangle (\boxoffset + 0.3*\xstep, 0.15*\yspace + 0.5*\boxoffset);

    \draw [alizarin, line width = 1.1pt] (-\xparen,-\yparen-\yparenextra) to [square left brace ] (-\xparen,\yspace + \yparen + \yparenextra);

    \draw [alizarin, line width = 1.1pt] (\xstep+\xparen,-\yparen-\yparenextra) to [square right brace ] (\xstep+\xparen,\yspace + \yparen+\yparenextra);
    
\end{scope}

\end{tikzpicture}
    \caption{To enumerate the family of intersecting pairs of paths on the left, we can perform casework on the earliest intersection point $v$ for the paths (the dashed border on the subpaths on the right indicates that the paths do not intersect before $v$). }
    \label{fig:intersection-casework}
\end{figure}
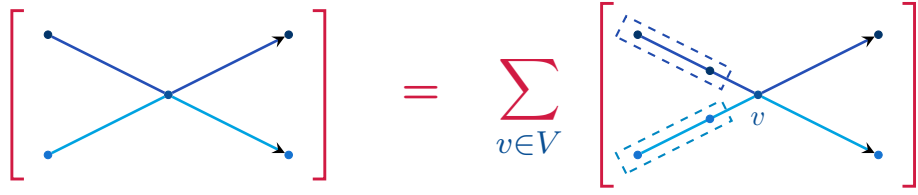

\begin{remark}[Enumeration Makes Generalization Easy]
Previous near-linear time algorithms for \textsf{2-DP} in DAGs and undirected graphs do not easily generalize to solving \textsf{2-DSP}.
In contrast, as outlined above, in our approach moving from \textsf{2-DP} to \textsf{2-DSP} is simple.
Why is this?

Intuitively, this happens because our algorithms take an \emph{enumerative} perspective on \textsf{2-DSP}, rather than the detection-based strategy of previous algorithms.
Older algorithms iteratively build up solutions to \textsf{2-DP} or \textsf{2-DSP}. 
Depending on the problem, this involves enforcing different sorts of constraints, since partial solutions to these problems may look quite different. 
In our approach, we just need to enumerate paths to solve \textsf{2-DP} and enumerate shortest paths to solve \textsf{2-DSP}.
Enumerating paths and shortest paths are both easy in DAGs by dynamic programming.
Hence algorithms for these two problems end up being essentially the same in our framework.
\end{remark}

\subsubsection*{From DAGs to Undirected Graphs}

When solving \textsf{2-DSP} in DAGs, we used the fact that DAGs have a topological order, so that any pair of paths intersects at a unique earliest vertex $v$ in this order.
This simple decomposition does not apply to solving \textsf{2-DSP} in undirected graphs, since we cannot rely on a fixed topological order.

Instead, we perform casework on the first vertex $v$ in $P_1$ lying in $P_1\cap P_2$.
We observe that in undirected graphs, there are two possibilities: $v$ is either the first vertex in $P_2$ lying in $P_1\cap P_2$, or it is the final vertex in $P_2$ lying in $P_1\cap P_2$.
Intuitively, the paths either ``agree'' and go in the same direction, or ``disagree'' and go in opposite directions. 

We then argue that over a field of characteristic two, we can efficiently enumerate over pairs of paths in each of these cases.
As with DAGs, 
we make this enumeration efficient by arguing that modulo 2 we can relax the (a priori difficult to check) condition of $v$ being the first intersection point on $P_1$ to some simpler ``local'' condition.
When the paths agree, this argument is similar to the reasoning used for solving \textsf{2-DSP} in DAGs.

For the case where the paths disagree, this enumeration is  more complicated, because there is no consistent linear ordering of the vertices neighboring $v$ on the two shortest paths, but can still be implemented in linear time using a more sophisticated local condition.
Specifically, if we let $a_i$ and $b_i$ denote the nodes appearing immediately before and after $v$ on path $P_i$, then to enumerate the `disagreeing paths'' modulo two, we prove that it suffices to enumerate paths $P_1$ and $P_2$ which intersect at $v$ and have the properties that $a_1\neq a_2$, $b_1\neq b_2$, and $a_1\neq b_2$.
Intuitively, using the subpath swapping idea depicted in \Cref{fig:intersetion-relax}, the conditions that $a_1\neq a_2$ and $a_1\neq b_2$ ensure that $v$ is the first vertex of $P_1$ lying in $P_1\cap P_2$, and the condition that $b_1\neq b_2$  ensures that the paths disagree.
To implement this idea, we need a slightly more complicated subpath swapping argument, which can also handle the case where two paths $P_1$ and $P_2$ intersect at vertices $u$ and $v$, with $u$ appearing before $v$ on $P_1$ but $u$ appearing after $v$ on $P_2$ (this situation does not occur in DAGs, but can occur in undirected graphs).
We do this by combining the previous subpath swapping idea with the observation that in undirected graphs we can also traverse subpaths in the reverse direction (so it is possible to swap the subpaths $P_1[u,v]$ and $P_2[v,u]$ in $P_1$ and $P_2$, even though $u$ and $v$ appear in different orders on $P_1$ and $P_2$).

By combining the enumerations for both cases, we can evaluate $F_{\disj}$, and thus solve \textsf{2-DSP} over undirected graphs.

\begin{figure}[t]
    \centering
    \begin{tikzpicture}[scale = 1.6,
lightvtx/.style=
{draw=dblue!70!midnight, fill=dblue!70!midnight,
circle, inner sep=1pt, minimum width=1pt},
darkvtx/.style=
{draw=cdblue!5!midnight, fill=cdblue!5!midnight, circle, inner sep=1pt, minimum width=1pt},
intervtx/.style=
{draw=dblue!35!midnight, fill=dblue!35!midnight, circle, inner sep = 1pt, minimum width = 1pt},
lightedge/.style=
{draw=dsblue!80!midnight, line width = 1pt, -stealth},
lightedgenoarrow/.style=
{draw=dsblue!80!midnight, line width = 1pt},
darkedge/.style=
{draw=cdblue!90!midnight, line width = 1pt, -stealth},
darkedgenoarrow/.style=
{draw=cdblue!90!midnight, line width = 1pt},
lightborder/.style=
{draw=dsblue!60!midnight, thick, dashed},
darkborder/.style=
{draw=cdblue!80!midnight, thick, dashed}]

    %% Vertical Displacement
\def\yspace{1cm};
\def\yborder{0.2cm};
\def\ybigparen{0.2cm};
\def\yparen{0.2cm};

\def\yparenextra{0.06cm};

%% Horizontal Displacement
\def\xstep{2*\yspace}; 
\def\xborder{0.2cm};
\def\xparen{0.2cm};

\def\extr{0.8cm};
\def\sumspace{0.7cm};

% Horizontal Transition
\def\extr{0.9cm};
\def\bigextr{1.3cm};

\def\boxoffset{0.19cm};

    \node[darkvtx] (s1) at (0,\yspace) {};
\node[lightvtx] (s2) at (0,0) {};

\node[intervtx] (a) at (0.5*\xstep, 0.5*\yspace) {};

\node[below=3pt] at (a) {\textcolor{dblue!30!midnight}{$v$}};

\node[darkvtx] (t1) at (\xstep, \yspace) {};
\node[lightvtx] (t2) at (\xstep, 0) {};

\draw[darkedge] (s1) -- (a) -- (t1);

\draw[lightedge] (s2) -- (a) -- (t2);

\node[darkvtx] (predarka) at (0.3*\xstep, \yspace - 0.3*\yspace) {};
\node[lightvtx] (prelighta) at (0.3*\xstep, 0.3*\yspace) {};

\draw[darkborder, rotate around={333.434948:(0.15*\xstep, \yspace - 0.15*\yspace)}] (-\boxoffset, \yspace - 0.15*\yspace - 0.5*\boxoffset) rectangle (\boxoffset + 0.3*\xstep, \yspace - 0.15*\yspace + 0.5*\boxoffset);

\draw[lightborder, rotate around={26.5650512:(0.15*\xstep, 0.15*\yspace)}] (-\boxoffset, 0.15*\yspace - 0.5*\boxoffset) rectangle (\boxoffset + 0.3*\xstep, 0.15*\yspace + 0.5*\boxoffset);

    \draw [alizarin, line width = 1.1pt] (-\xparen,-\yparen-\yparenextra) to [square left brace ] (-\xparen,\yspace + \yparen + \yparenextra);

    \draw [alizarin, line width = 1.1pt] (\xstep+\xparen,-\yparen-\yparenextra) to [square right brace ] (\xstep+\xparen,\yspace + \yparen+\yparenextra);

\begin{scope}[xshift = \xstep + \xparen + \extr]
        \node at (0, 0.5*\yspace) 
{\scalebox{1.5}{\textcolor{alizarin}{$\boldsymbol{\equiv}$}}};
\end{scope}

\begin{scope}[xshift = \xstep + \xparen + 2*\extr + \xparen]
        \node[darkvtx] (s1) at (0,\yspace) {};
\node[lightvtx] (s2) at (0,0) {};

\node[intervtx] (a) at (0.5*\xstep, 0.5*\yspace) {};

\node[below=3pt] at (a) {\textcolor{dblue!30!midnight}{$v$}};

\node[darkvtx] (t1) at (\xstep, \yspace) {};
\node[lightvtx] (t2) at (\xstep, 0) {};

\draw[darkedge] (s1) -- (a) -- (t1);

\draw[lightedge] (s2) -- (a) -- (t2);

\node[darkvtx] (predarka) at (0.3*\xstep, \yspace - 0.3*\yspace) {};
\node[lightvtx] (prelighta) at (0.3*\xstep, 0.3*\yspace) {};

\draw[darkborder, rotate around={333.434948:(predarka)}] (-\boxoffset + 0.3*\xstep, \yspace - 0.3*\yspace - 0.5*\boxoffset) rectangle (\boxoffset + 0.3*\xstep, \yspace - 0.3*\yspace + 0.5*\boxoffset);

\draw[lightborder, rotate around={26.5650512:(prelighta)}] (-\boxoffset + 0.3*\xstep, 0.3*\yspace - 0.5*\boxoffset) rectangle (\boxoffset + 0.3*\xstep, 0.3*\yspace + 0.5*\boxoffset);

    \draw [alizarin, line width = 1.1pt] (-\xparen,-\yparen) to [square left brace ] (-\xparen,\yspace + \yparen);

    \draw [alizarin, line width = 1.1pt] (\xstep+\xparen,-\yparen) to [square right brace ] (\xstep+\xparen,\yspace + \yparen);
\end{scope}

\begin{scope}[xshift = \xstep + \xparen + 2*\extr + \xstep + \xparen + \bigextr]
            \node at (0, 0.5*\yspace) 
{\scalebox{1.5}{\textcolor{alizarin}{$\pmod 2$}}};
\end{scope}

\end{tikzpicture}
    \caption{
    If we work modulo two, then we can enumerate pairs of paths which have common first intersection at node $v$ by enumerating pairs of paths which intersect at $v$ and have the property that the vertices appearing immediately before $v$ on each path are distinct. }
    \label{fig:intersetion-relax}
\end{figure}
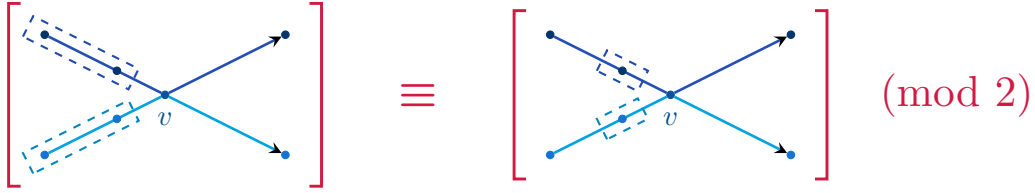

\subsection{\texorpdfstring{$\boldsymbol{k}$}{k}-EDSP Algorithm}

The previous algorithm of \cite[Proposition 10]{BercziKobayashi2017} for \textsf{$k$-EDSP} works by constructing a graph $G'$ encoding information about $k$-tuples of edge-disjoint shortest paths in the original graph $G$.
This new graph $G'$ has special nodes $\vec{s}$ and $\vec{t}$, such that there is a path from $\vec{s}$ to $\vec{t}$ in $G'$ if and only if $G$ contains $k$ edge-disjoint shortest paths connecting its terminals.
The nodes of the new graph $G'$ are $k$-tuples of edges $(e_1, \dots, e_k)$ where each $e_i$ is an edge in $G$.
So constructing $G'$ already takes $\Omega(m^k)$ time.

Our algorithm for \textsf{$k$-EDSP} uses the same general idea.
We construct an alternate graph $G'$ which still has the property that finding a single path between two specified vertices of $G'$ solves the \textsf{$k$-EDSP} problem in $G$.
However, we design $G'$ to have nodes of the form $(v_1, \dots, v_k)$, where each $v_i$ is a \emph{vertex} in $G$.
Our construction produces a graph with $n^k$ nodes and $O(mn^{k-1})$ edges, which yields the speed-up. 
We avoid the $\Omega(m^k)$ bottleneck of the previous algorithm by showing how to encode edge-disjointness information simply through the $k$-tuples of vertices, rather than edges, that the $k$ potential solution paths in $G$ traverse.

\subsection{Lower Bounds}

\subsubsection*{Disjoint Shortest Paths}

Our proof of \cref{lbdspboth} is based on the reduction of 
 % Bentert, Nichterlein, Renken, and Zschoche
\cite[Proposition 1]{Geometric-Lens} from \textsf{$k$-Clique} to $2k$-\textsf{DSP} on undirected graphs, which also easily extends to DAGs. %A similar reduction to $2k$-\textsf{EDSP} on planar DAGs is also known~\cite[Theorem 4]{AmiriWargallav4}. 
Our contribution is a transformation that reduces the number of paths in their reduction from $2k$ to $k$ by exploiting the symmetry of the construction. 

The reduction of \cite{Geometric-Lens} maps each vertex $v$ in the \textsf{$k$-Clique} instance to a horizontal path $P_v$ and a vertical path $Q_v$, each of length $n$. These paths are arranged so that for each pair of vertices $(v,w)$ in the input graph, the paths $P_v$ and $Q_w$ intersect if and only if $(v,w)$ is not an edge in the input graph. To achieve this, the paths are placed along a grid, and at the intersection point in the grid between paths $P_v$ and $Q_w$, these two paths are modified to bypass each other to avoid intersection if $(v,w)$ is an edge in the input graph.

The main idea of our transformation is the following. Since the known reduction is symmetric along the diagonal of the grid, it contains some redundancy. 
We remove this redundancy by only keeping the portion of the grid below the diagonal. 
To do this, we only have one path $P_v$ for each vertex $v$ in the input graph, and each such path has both a horizontal component and a vertical component. Each path turns from horizontal to vertical when it hits the diagonal. 
As a result, each pair of paths $(P_v, P_w)$ has exactly one intersection point in the grid (which we bypass if $(v,w)$ is an edge in the input graph). 
Since we produce only a single path $P_v$ for each vertex $v$, we obtain a reduction to $k$-\textsf{DSP} instead of $2k$-\textsf{DSP}.

\subsubsection*{Disjoint Paths}

The starting point for \cref{lbdp} is the work of \cite{Slivkins2010}, which reduces from \textsf{$k$-Clique} to \textsf{$p$-EDP} in a DAG with $O(kn)$ nodes, for $p=k + \binom{k}{2}$. 
The parameter blows up from $k$ to $p$ in this way because the reduction uses $k$ solution paths to pick $k$ vertices in the original graph, and then for each of the $\binom{k}{2}$ pairs of vertices chosen, uses an additional solution path to verify that the vertices in that pair are adjacent in the original graph. 

We improve upon this by modifying the reduction graph to allow some solution paths to check multiple edges simultaneously.
This lets us avoid using $\binom{k}{2}$ solution paths to separately check for edges between each pair of nodes in a candidate $k$-clique. 
Instead, we employ just $\lfloor k^2/4\rfloor$ solution paths in the reduction, roughly halving the number of paths needed. 

To do this, we need to precisely identify which paths can check for multiple edges without compromising the correctness of the reduction. 
To this end, we examine the structure of the reduction and define a notion of a \emph{covering family} which characterizes which paths can safely check for multiple edges at once. 
Formally, a $k$-covering family is a collection $\mathcal{L}$ of increasing lists of positive integers, with the property that for all integers $i,j$ with $1\le i < j\le k$, some list in $\mathcal{L}$ contains $i$ and $j$ as consecutive members.

We show that for any $k$, the smallest number of lists in a $k$-covering family is $\lambda(k)  = \lfloor k^2/4\rfloor$ (note that merely obtaining asymptotically tight bounds would not suffice for designing interesting conditional lower bounds). 
We then insert this construction of a minimum size covering family into the framework of the reduction and prove that the reduction remains correct. 
Intuitively, given lists in a covering family, we can map each list $L$ to a path  which checks edges between vertex parts $V_i$ and $V_j$ for each $(i,j)$ pair appearing as consecutive members of $L$.

The original reduction of \cite{Slivkins2010} corresponds to implementing this strategy with the trivial $k$-covering family using $\binom{k}{2}$ lists, achieved by taking a single increasing list of two elements for each unordered pair of integers from $[k]$.
Our improved reduction comes from implementing this framework with the optimal bound of $\lfloor k^2/4\rfloor$ lists.

This yields reductions from \textsf{$k$-Clique} to \textsf{$p$-DP} and \textsf{$p$-DSP} for $p= k + \lambda(k) = k + \lfloor k^2/4\rfloor$.

 %As a result, our construction is slightly different from Slivkins's and we will present it in a self-contained way instead of referencing his work as a black box.

%to improve its encoding efficiency. First we will define the puzzle and present our solution to it. 

% Previous work reduces \textsf{$k$-Clique} on $kn$ vertices to \textsf{$\grp{k + \binom{k}{2}}$-EDP} in a DAG with $O(kn)$ nodes \cite{Slivkins2010}.
% The parameter blows up in this way because the reduction uses $k$ solution paths to pick $k$ vertices in the original graph, and then for each of the $\binom{k}{2}$ pairs of vertices chosen, uses an additional solution path to verify that those vertices are adjacent in the original graph.

% We improve on this earlier argument by modifying the reduction graph to allow some solution paths to check multiple edges at once.
% This lets us avoid using $\binom{k}{2}$ solution paths to separately check for edges between each pair of nodes in a candidate $k$-clique, and instead employ just $\lfloor k^2/4\rfloor$ solution paths in the reduction.

% We relate the 

\section{2-DSP}
\label{sec:2dsp}

In this section, we present our new algorithms for \textsf{2-DSP}.

We begin in \Cref{subsec:alg-prelim} by introducing notation and terminology used in our algorithms.
We also present a general ``subpath swapping lemma,'' which we repeatedly invoke to efficiently evaluate polynomials modulo 2.
Then in \Cref{subsec:helpers} we define and prove results about various  polynomials used in our algorithms.
All the definitions and results in \Cref{subsec:alg-prelim} and \Cref{subsec:helpers} apply to both DAGs and undirected graphs.

Following this setup, in \Cref{subsec:2dsp-DAG-alg} we present an algorithm for \textsf{2-DSP} in DAGs proving \Cref{2dsp-DAG}, and in \Cref{subsec:2dsp-undir-alg} we present an algorithm for \textsf{2-DSP} in undirected graphs proving \Cref{2dsp-undir}.
We note that
\Cref{subsec:2dsp-undir-alg} can be read independently of \Cref{subsec:2dsp-DAG-alg} (but both \Cref{subsec:2dsp-DAG-alg,subsec:2dsp-undir-alg}
rely on the definitions and results from \Cref{subsec:alg-prelim,subsec:helpers}).

\subsection{Graph and Polynomial Preliminaries}
\label{subsec:alg-prelim}

\subsubsection*{Graph Notation}
Let $G$ denote the input graph on $n$ vertices and $m$ edges.
Let $V$ be the vertex set of $G$.

For each $i\in [2]$, we define $G_i$ to be the $s_i$-shortest paths DAG of $G$.

Given $i\in [2]$ and vertex $v$, we let $V_{\text{in}}^{i}(v)$ denote the in-neighbors of $v$ in $G_i$, and $\vout^i(v)$ denote the out-neighbors of $v$ in $G_i$.
We further let $\vin(v) = \vin^1(v)\cap \vin^2(v)$ and $\vout(v) = \vout^1(v)\cap \vout^2(v)$ be the sets of in-neighbors and out-neighbors respectively of node $v$ common to both $G_1$ and $G_2$.
Furthermore, we define $\vmix(v) = \vin^1(v)\cap \vout^2(2)$ to be the ``mixed neighborhood'' of $v$. 

We call a pair of paths $\langle P_1,P_2\rangle$ \emph{standard} if each $P_i$ is an $(s_i,t_i)$-path in $G_i$.

\subsubsection*{Algebraic Preliminaries}

For each edge $e = (u,v)$, we introduce an indeterminate $x_{uv}$.

If $G$ is undirected, we set 
$x_{uv} = x_{vu} = x_{\set{u,v}}$
to reflect that edges are unordered pairs of vertices.

\noindent Given a path $P = \langle v_0, \dots, v_r\rangle $, we assign it the monomial
    \[f(P) = \prod_{j=0}^{r-1}x_{v_jv_{j+1}}.\]    
Given a pair of paths $\ca{P} = \langle P_1, P_2\rangle$, we assign it monomial 
    \[f(\ca{P}) = f(P_1,P_2) = f(P_1)f(P_2).\]
Given a collection $\ca{S}$ of paths or pairs of paths, we say that a ``polynomial $F$ \emph{enumerates} $\ca{S}$,'' or equivalently ``$F$ is the \emph{enumerating polynomial} for $\ca{S}$,'' if 
    \[F = \sum_{S\in \ca{S}} f(S).\]

\subsubsection*{Subpath Swapping}

\begin{lemma}[Shortest Path Swapping]
\label{preserve-shortest-paths}
    Let $P_1$ and $P_2$ be shortest paths passing through vertices $a$ and $b$ in that order.
    Then the walks obtained by swapping the $a$ to $b$ subpaths in $P_1$ and $P_2$ are also shortest paths.
\end{lemma}
\begin{proof}
Since $P_1$ and $P_2$ are shortest paths, each of their $a$ to $b$ subpaths have length $\dist(a,b)$.
Since these subpaths have the same length, swapping the $a$ to $b$ subpaths of $P_1$ and $P_2$ produce walks with the same endpoints and lengths as $P_1$ and $P_2$ respectively.
Since all edge weights are positive, these walks have no repeat vertices (since if one of the walks did have repeat vertices, then it would contain a cycle which we could remove to produce a walk of shorter total length between its endpoints, which would then contradict the assumption that $P_1$ and $P_2$ are shortest paths).
Thus these walks must be shortest paths as claimed.
\end{proof}

We will repeatedly make use of the following result, which gives a way of simplifying enumerating polynomials for certain families of pairs of paths.
Although the statement may appear technical, the lemma simply formalizes the idea that for certain enumerating polynomials, we can pair up terms of equal value and get cancellation modulo two.

\begin{lemma}[Vanishing Modulo 2]
\label{lem:subpath-swapping-cancels}
    Let $\ca{F}$ be a family of pairs of paths in $G$, and let $\ca{S}\subseteq\ca{F}$.
    Suppose there exist maps $\alpha, \beta: \ca{S}\to V$ and $\Phi: \ca{S}\to \ca{S}$ such that for all $\ca{P} = \langle P_1,P_2\rangle\in\ca{S}$, 
    \begin{enumerate}
        \item the vertices $a = \alpha(\ca{P})$ and $b=\beta(\ca{P})$ lie in $P_1\cap P_2$, $a$ appears before $b$ in $P_1$ and $P_2$, and the subpaths $P_1[a,b]$ and $P_2[a,b]$ are distinct;
        
        \item we have 
        $\Phi(\ca{P}) = \langle Q_1, Q_2\rangle$, where $Q_1$ is obtained by replacing the $a$ to $b$ subpath in $P_1$ with $P_2[a,b]$, and $Q_2$ is obtained by replacing the $a$ to $b$ subpath in $P_2$ with $P_1[a,b]$; and

        \item we have $\Phi(\Phi(\ca{P})) = \ca{P}$.
    \end{enumerate}
    Then the enumerating polynomial for $\ca{F}$ is the same as the enumerating polynomial for $\ca{F}\setminus\ca{S}$.
\end{lemma}
\begin{proof}
    Let $F$ be the enumerating polynomial for $\ca{F}$.
    By definition, we have 
        \begin{equation}
        \label{enum:def}
        F = \sum_{\ca{P}\in\ca{F}} f(\ca{P}) = \sum_{\ca{P}\in\ca{F}\setminus \ca{S}} f(\ca{P}) + \sum_{\ca{P}\in\ca{S}} f(\ca{P}).
        \end{equation}
Take any $\ca{P} = \langle P_1,P_2\rangle\in \ca{S}$.
By property 1 from the lemma statement, the subpaths from $\alpha(\ca{P})$ to $\beta(\ca{P})$ in $P_1$ and $P_2$ are distinct.
Then by property 2, $\Phi(\ca{P})\neq\ca{P}$.
Consequently, by property 3, we can partition $\ca{S} = \ca{S}_1\cup \ca{S}_2$ into two equally sized pieces such that $\Phi$ is a bijection from $\ca{S}_1$ to $\ca{S}_2$.
So we can write 
    \begin{equation}
    \label{summing-in-half}
    \sum_{\ca{P}\in\ca{S}} f(\ca{P}) = \sum_{\ca{P}\in\ca{S}_1} f(\ca{P}) + \sum_{\ca{P}\in\ca{S}_2} f(\ca{P}) = \sum_{\ca{P}\in\ca{S}_1} \grp{f(\ca{P}) + f(\Phi(\ca{P}))}.
    \end{equation}
By property 2, the multiset of edges traversed by the pair $\ca{P}$ is the same as the multiset of edges traversed by $\Phi(\ca{P})$, for all $\ca{P}\in\ca{S}$.
Consequently, $f(\ca{P}) = f(\Phi(\ca{P}))$ for all $\ca{P}\in \ca{S}$.

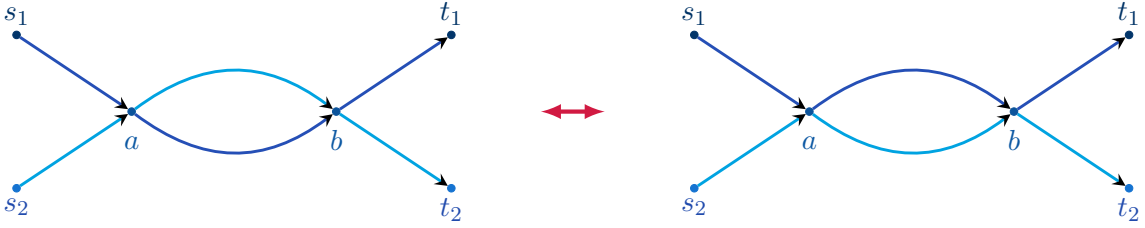
\begin{figure}[t]
    \centering
    %% General Subpath Swapping Setup
\begin{tikzpicture}[scale = 1.7,
lightvtx/.style=
{draw=dblue!70!midnight, fill=dblue!70!midnight,
circle, inner sep=1pt, minimum width=1pt},
darkvtx/.style=
{draw=cdblue!5!midnight, fill=cdblue!5!midnight, circle, inner sep=1pt, minimum width=1pt},
intervtx/.style=
{draw=dblue!35!midnight, fill=dblue!35!midnight, circle, inner sep = 1pt, minimum width = 1.1pt},
lightedge/.style=
{draw=dsblue!80!midnight, line width = 1.1pt, -stealth},
darkedge/.style=
{draw=cdblue!90!midnight, line width = 1.1pt, -stealth}]

%% Vertical Displacement
\def\ystep{0.6cm};
\def\controller{0.7*\ystep};

%% Horizontal Displacement
\def\xinit{0.9cm};
\def\xstep{1.6cm};

% Horizontal Transition
\def\extrw{0.5cm};
\def\extr{0.7cm};

\node[darkvtx] (s1) at (0,\ystep) {};
\node[lightvtx] (s2) at (0,-\ystep) {};

\node[midnight, above] at (s1) {$s_1$};
\node[cdblue!80!midnight, below] at (s2) {$s_2$};

\node[intervtx] (a) at (\xinit, 0) {};

\node[dblue!40!midnight, below=5.5pt] at (a) {$a$};

\node[intervtx] (b) at (\xinit + \xstep, 0) {};

\node[dblue!40!midnight, below=3pt] at (b) {$b$};

\node[darkvtx] (t1) at (2*\xinit + \xstep, \ystep) {};
\node[lightvtx] (t2) at (2*\xinit + \xstep, -\ystep) {};

\node[midnight, above] at (t1) {$t_1$};
\node[cdblue!80!midnight, below] at (t2) {$t_2$};

\node (lowerleftcontrol) at (\xinit + 0.35*\xstep, -\controller) {};
\node (lowerrightcontrol) at (\xinit + 0.65*\xstep, -\controller) {};

\node (upperleftcontrol) at (\xinit + 0.35*\xstep, \controller) {};
\node (upperrightcontrol) at (\xinit + 0.65*\xstep, \controller) {};

\draw[darkedge] (s1) -- (a);
\draw[lightedge] (s2) -- (a);

\draw[darkedge] (a) .. controls (lowerleftcontrol.center) and (lowerrightcontrol.center) ..   (b);
\draw[lightedge] (a) .. controls (upperleftcontrol.center) and (upperrightcontrol.center) ..   (b);

\draw[darkedge] (b) -- (t1);
\draw[lightedge] (b) -- (t2);

%% Edge Transition
\draw[ultra thick, alizarin, latex-latex] (2*\xinit + \xstep + \extr, 0) -- (2*\xinit + \xstep + \extr + \extrw, 0);

%% The swapped copy
\begin{scope}[xshift = 2*\xinit + \xstep + 2*\extr + \extrw]
    \node[darkvtx] (s1) at (0,\ystep) {};
\node[lightvtx] (s2) at (0,-\ystep) {};

\node[midnight, above] at (s1) {$s_1$};
\node[cdblue!80!midnight, below] at (s2) {$s_2$};

\node[intervtx] (a) at (\xinit, 0) {};

\node[dblue!40!midnight, below=5.5pt] at (a) {$a$};

\node[intervtx] (b) at (\xinit + \xstep, 0) {};

\node[dblue!40!midnight, below=3pt] at (b) {$b$};

\node[darkvtx] (t1) at (2*\xinit + \xstep, \ystep) {};
\node[lightvtx] (t2) at (2*\xinit + \xstep, -\ystep) {};

\node[midnight, above] at (t1) {$t_1$};
\node[cdblue!80!midnight, below] at (t2) {$t_2$};

\node (lowerleftcontrol) at (\xinit + 0.35*\xstep, -\controller) {};
\node (lowerrightcontrol) at (\xinit + 0.65*\xstep, -\controller) {};

\node (upperleftcontrol) at (\xinit + 0.35*\xstep, \controller) {};
\node (upperrightcontrol) at (\xinit + 0.65*\xstep, \controller) {};

\draw[darkedge] (s1) -- (a);
\draw[lightedge] (s2) -- (a);

\draw[lightedge] (a) .. controls (lowerleftcontrol.center) and (lowerrightcontrol.center) ..   (b);
\draw[darkedge] (a) .. controls (upperleftcontrol.center) and (upperrightcontrol.center) ..   (b);

\draw[darkedge] (b) -- (t1);
\draw[lightedge] (b) -- (t2);
\end{scope}
    
\end{tikzpicture}
    \caption{Given paths $P_1$ and $P_2$ which intersect at nodes $a = \alpha(P_1, P_2)$ and $b=\beta(P_1, P_2)$, such that $a$ appears before $b$ on both paths, if we swap the $a$ to $b$ subpaths of of $P_1$ and $P_2$ to produce new paths $Q_1$ and $Q_2$ respectively, then these pairs $f(P_1, P_2) = f(Q_1, Q_2)$ have the same monomials. 
    Moreover, swapping the $a$ to $b$ subpaths of $Q_1$ and $Q_2$ recovers $P_1$ and $P_2$.}
    \label{fig:subpath-swap}
\end{figure}

The subpath swapping procedure determined by $\Phi$ is depicted in \Cref{fig:subpath-swap}.

Since we work over a field of characteristic two, this implies that 
\[\sum_{\ca{P}\in\ca{S}_1} \grp{f(\ca{P}) + f(\Phi(\ca{P}))} = 0.\]

Substituting the above equation into \cref{summing-in-half} implies that 
    \[\sum_{\ca{P}\in\ca{S}} f(\ca{P}) = 0.\]

Then substituting the above equation into \cref{enum:def} yields 
    \[F = \sum_{\ca{P}\in\ca{F}\setminus \ca{S}} f(\ca{P}).\]

This proves that $F$ is the enumerating polynomial for $\ca{F}\setminus\ca{S}$ as desired.
\end{proof}

% Let $F_{\cap}$ be the enumerating polynomial for the collection of pairs of paths $\langle P_1, P_2\rangle$ where 
% \begin{enumerate}
%     \item $P_i$ is an $(s_i,t_i)$-path in $G_i$, and 
%     \item
%         $P_1$ and $P_2$ intersect at some vertex.
% \end{enumerate}

% The following lemma shows that to evaluate $F_{\disj}$, it suffices to evaluate $F_{\cap}$.

% \begin{lemma}
%     We have 
%         \[F_{\disj} = L_1(t_1)L_2(t_2) - F_{\cap}.\]
% \end{lemma}
% \begin{proof}
    
% \end{proof}

\subsection{Helper Polynomials}
\label{subsec:helpers}

In this section we define various enumerating polynomials which will help us solve \textsf{2-DSP}.
All of the results in this section hold if the input graph $G$ is a DAG or is undirected.

\subsubsection*{Paths to and from Terminals}

For each $i\in [2]$ and vertex $v$, 
let $L_i(v)$ be the enumerating polynomial for the set of $(s_i,v)$-paths in $G_i$, and let $R_i(v)$ be the enumerating polynomial for the set of $(v,t_i)$-paths in $G_i$.
\begin{lemma}[Path Polynomials]
    \label{lem:dp}
    For each $i\in [2]$ and vertex $v$, we have 
        \begin{equation}
        \label{left:eq}
        L_i(v) = \sum_{u\in \vin^{i}(v)} L_i(u)x_{uv} 
        \end{equation}
    and
        \begin{equation}
        \label{right:eq}
        R_i(v) = \sum_{w\in\vout^i(v)} x_{vw}R_i(w).
        \end{equation}
\end{lemma}
\begin{proof}
    Since $L_i(u)$ enumerates $(s_i,u)$-paths in $G_i$, the polynomial
        \[L_i(u)x_{uv}\]
    enumerates $(s_i,v)$-paths in $G_i$, whose penultimate vertex is $u$.
    Every $(s_i,v)$-path in $G_i$ has some unique penultimate vertex $u\in \vin^i(v)$.
    Consequently
    \[\sum_{u\in \vin^{i}(v)} L_i(u)x_{uv} \]
    enumerates all $(s_i,v)$-paths in $G_i$, which proves \cref{left:eq} as desired.
    
    Symmetric reasoning proves \cref{right:eq}.
\end{proof}

\subsubsection*{Intersecting Paths}

\noindent Let $F_{\disj}$ be the enumerating polynomial for the collection of vertex-disjoint, standard pairs of paths.
Let $F_\cap$ be the enumerating polynomial for the collection of intersecting, standard pairs of paths.
% Observe that $F_{\disj}$ and $F_\cap$ are essentially complementary, enumerating disjoint pairs of paths and intersecting pairs of paths respectively.
Our algorithms for \textsf{2-DSP} work by evaluating $F_{\disj}$ at a random point.
The following observation shows that we can compute an evaluation of $F_{\disj}$ by evaluating $F_\cap$ instead.

\begin{lemma}[Disjoint Paths $\le$ Intersecting Paths]
    \label{disj-to-intersect}
    We have
    \[F_{\disj} = L_1(t_1)L_2(t_2) - F_{\cap}.\]
\end{lemma}
\begin{proof}
            By expanding out the product, we see that $L_1(t_1)L_2(t_2)$ enumerates all standard pairs of paths $\langle P_1,P_2\rangle$.
            Each such pair is either vertex-disjoint or consists of paths intersecting at a common node, so we have 
            \[L_1(t_1)L_2(t_2) = F_{\disj} + F_{\cap}\]
            which implies the desired result.
\end{proof}

\subsubsection*{Linkages from Two Sources to a Common Vertex}

Our algorithms will use the following ``linkage'' polynomials, which enumerate pairs of internally vertex-disjoint paths beginning at the source nodes $s_1$ and $s_2$, and ending at a common vertex.

\begin{definition}[Source Linkage]
\label{def:disjoint-v}
Given a vertex $v$, let $\ca{D}(v)$ be the set of pairs of paths $\langle P_1,P_2\rangle$ where each $P_i$ is an $(s_i,v)$-path in $G_i$ and the paths intersect only at $v$.
Let $D(v)$ be the enumerating polynomial for $\ca{D}(v)$.
\end{definition}

We now observe that for the purpose of enumeration modulo two, one can replace the ``global constraint'' that a pairs of paths in $\ca{D}(v)$ has no intersection anywhere before $v$ with the easier to check ``local constraint'' that a pair of paths has no intersection immediately before $v$.

\begin{definition}[Linkage Relaxation]
    Given a vertex $v$, let $\tilde{\ca{D}}(v)$ be the set of pairs of paths $\langle P_1,P_2\rangle$, where each $P_i$ is an $(s_i,v)$-path in $G_i$, such that the penultimate vertices of $P_1$ and $P_2$ are distinct.
\end{definition}

\begin{lemma}[Relaxing Source Linkages]
    \label{linkage-relaxation}
    For each vertex $v$, the polynomial $D(v)$ enumerates $\tilde{\ca{D}}(v)$.
\end{lemma}
        \begin{proof}
            For convenience, let $\ca{F} = \tilde{\ca{D}}(v)$.
            Let $\ca{S} = \ca{F}\setminus\ca{D}(v)$ be the family of pairs of paths $\langle P_1,P_2\rangle$ where each $P_i$ is an $(s_i,v)$-path in $G_i$, such that 
            \begin{enumerate}
                \item the paths $P_1$ and $P_2$ intersect at some node other than $v$, and 
                \item the nodes immediately before $v$ on $P_1$ and $P_2$ are distinct.
            \end{enumerate}
            Take arbitrary $\langle P_1,P_2\rangle\in\ca{S}$.
            Let $u$ be the vertex in $P_1\cap P_2$ maximizing the value of $\dist(u,v)$.
            By condition 1 above, $u\neq v$.
            By condition 2 above, $P_1[u,v]$ and $P_2[u,v]$ are distinct.
        
                Now define paths 
                \[Q_1 = P_1[s_1,u]\dia P_2[u,v] \quad
                \text{and}
                \quad
                Q_2 = P_2[s_2,u]\dia P_1[u,v].\]
            By \Cref{preserve-shortest-paths} each $Q_i$ is a shortest path, and thus an $(s_i,v)$-path in $G_i$.
        
            The pair $\langle Q_1,Q_2\rangle$ satisfies condition 1 above, since $Q_1$ and $Q_2$ intersect at $u$.
            This pair also satisfies condition 2 above, since the penultimate vertices of $Q_1$ and $Q_2$ are the same as the penultimate vertices of $P_2$ and $P_1$ respectively.
            Also, $u$ is the node in $Q_1\cap Q_2$ maximizing $\dist(u,v)$, so we can perform the same subpath swapping operation as above to go from $\langle Q_1, Q_2\rangle$ to $\langle P_1, P_2\rangle$.
        
            Then by the discussion in the previous paragraph, \Cref{lem:subpath-swapping-cancels} implies that the enumerating polynomial for $\ca{F}$ is the same as the enumerating polynomial for $\ca{F}\setminus\ca{S} = \ca{D}(v)$.
           The enumerating polynomial for $\ca{D}(v)$ is $D(v)$, so this proves the desired result.
        \end{proof}

\begin{lemma}[Enumerating Relaxed Linkages]
\label{enum:relax-linkage}
    For each vertex $v$, the polynomial
            \[L_1(v)L_2(v) - \sum_{u\in \vin(v)} L_1(u)L_2(u)x_{uv}^2\]
        enumerates $\tilde{D}(v)$.
\end{lemma}
\begin{proof}
        By expanding out the product, we see that the polynomial  $L_1(v)L_2(v)$  enumerates all pairs of paths $\langle P_1,P_2\rangle$ where each $P_i$ is an $(s_i,v)$-path in $G_i$.
    We claim that 
    \begin{equation}
    \label{DAG:error}
    \sum_{u\in \vin(v)} L_1(u)L_2(u)x_{uv}^2
    \end{equation}
    enumerates all pairs of paths $\langle P_1,P_2\rangle$ such that each $P_i$ is a path in $G_i$ beginning at $s_i$ and ending at an edge $(u,v)$ for some node $u\in \vin(v)$.

    Indeed, paths $P_1$ and $P_2$ from any such pair can be split along their final edges into  
        \[P_1 = P_1[s_1,u]\dia (u,v) \quad\text{and}\quad
        P_2 = P_2[s_2,u] \dia (u,v).\]
    The paths $P_i[s_i,u]$ are enumerated by the $L_i(u)$ factors in \cref{DAG:error}, and the two copies of  $(u,v)$ are encoded by the $x_{uv}^2$ factor in \cref{DAG:error}.
    Conversely, any monomial in the expansion of
    \begin{equation}
        \label{eq:quick}
        L_1(u)L_2(u)x_{uv}^2
    \end{equation}
    is the product of monomials for some $(s_i,u)$-paths $Q_i$ in $G_i$ and two occurrences of the edge $(u,v)$, so that if we define
        \[P_1 = Q_1 \dia (u,v)\quad\text{and}\quad
        P_2 = Q_2 \dia (u,v)\]
    then the monomial we were considering is precisely the monomial for the pair $\langle P_1, P_2\rangle$.
    Summing over all vertices $u\in\vin$ proves that \cref{eq:quick} enumerates the pairs of paths we described above.

    From the discussion above, it follows that 
    \[
    L_1(v)L_2(v) - \sum_{u\in \vin(v)} L_1(u)L_2(u)x_{uv}^2
    \]
    enumerates all pairs of paths $\langle P_1,P_2\rangle$, where  $P_i$ is an $(s_i,v)$-path in $G_i$, such that $P_1$ and $P_2$ have distinct penultimate vertices.
    Thus this polynomial enumerates $\tilde{D}(v)$ as desired.
\end{proof}

\begin{lemma}[Enumerating Source Linkages]
\label{lem:DAG-left-disjoint}
    For each vertex $v$, we have 
\[
        D(v) = L_1(v)L_2(v) - \sum_{u\in \vin(v)} L_1(u)L_2(u)x_{uv}^2.
    \]
\end{lemma}
\begin{proof}
    This result follows by combining \Cref{linkage-relaxation,enum:relax-linkage}.
\end{proof}

    %% If we do not want to use the general lemma. 
    % It follows that this map we have defined, from $\langle P_1, P_2\rangle$ to $\langle Q_1, Q_2\rangle$, matches pairs of paths not in $\ca{D}(v)$ which are enumerated by \cref{eq:real-expr} into groups of two.
    % Moreover, this map preserves the monomial corresponding to a pair, because the multiset of edges used in the $P_i$ paths is the same as the multiset of edges used by the $Q_i$ paths.
    % Since we are working over a field of characteristic two, this implies that all such pairs of paths have net zero contribution to the polynomial from \cref{eq:real-expr}.

\subsection{2-DSP in DAGs}
\label{subsec:2dsp-DAG-alg}

In this subsection, assume that $G$ is a weighted DAG, and recall the definitions from \Cref{subsec:alg-prelim}.
We fix an arbitrary topological order on the vertices of $G$.

Our main goal is to efficiently evaluate a polynomial encoding pairs of disjoint shortest paths in $G$.
By \Cref{disj-to-intersect}, it suffices to evaluate a polynomial enumerating pairs of intersecting shortest paths in $G$.
We do this by casework on the first intersection point of these paths. 

\begin{lemma}[Enumerating Intersecting Paths]
    \label{lem:DAG-disj-poly}
    We have 
        \[F_{\cap} =  \sum_{v\in V} D(v)R_1(v)R_2(v).\]
\end{lemma}
\begin{proof}        
For any fixed vertex $v$, we claim that 
            \begin{equation}
            \label{leftright}
            D(v)R_1(v)R_2(v)
            \end{equation}
        enumerates all standard pairs of paths $\langle P_1,P_2\rangle$ such that $P_1$ and $P_2$ have first  intersection at $v$.
        Indeed, given any such pair of paths $\langle P_1,P_2\rangle$, we can decompose
            \[P_i = P_i[s_i,v]\dia P_i[v,t_i]\]
        and observe that the pair $\langle P_1[s_1,v],P_2[s_2,v]\rangle$ is enumerated by the $D(v)$ factor in \cref{leftright}, and the $P_i[v,t_i]$ paths are enumerated by the respective $R_i(v)$ factors from \cref{leftright}.
        Conversely, any monomial in the expansion of \cref{leftright} is the product of monomials encoding the edges of a pair of paths $\langle A_1,A_2\rangle$, where $A_i$ is an $(s_i,v)$-path in $G_i$ such that $A_1$ and $A_2$ only intersect at $v$, and paths $B_i$ from $v$ to $t_i$ in $G_i$.
        Then if we define 
            \[P_i = A_i\dia B_i\]
        we see that the $P_i$ are $(s_i,t_i)$-paths in $G_i$ with the property that $P_1$ and $P_2$ first intersect at $v$.
        Here, we are using the fact that $G$ is a DAG---this ensures that every node in $A_1$ or $A_2$ appears at or before $v$ in the topological order, and that every node in $B_1$ or $B_2$ appears at or after $v$ in the topological order, so that 
        $A_1\cap B_2 = A_2\cap B_1 = \set{v}$.

        Since every pair of intersecting paths intersects at a unique earliest vertex, it follows that 
        \[\sum_{v\in V} D(v)R_1(v)R_2(v)\]
        enumerates all intersecting, standard pairs of paths $\langle P_1,P_2\rangle$.
        This proves the desired result.
\end{proof}
% Our algorithm is as follows.
% \begin{enumerate}
%     \item Set each $x_{uv}$ equal to an independent, uniform random element of $\mathbb{F}_{2^d}$.
%     \item Compute $L_i(v)$ and $R_i(v)$ at this assignment, for each $i\in\set{1,2}$ and vertex $v$.
%     \item Compute $D(v)$ at this assignment for each vertex $v$.
%     \item 
%     For our assignment of values to the $x_{uv}$ variables, compute
%         \begin{equation}
%         \label{DAG-poly-expr}
%         L_1(t_1)L_2(t_2) + \sum_{v\in V} D(v)R_1(v)R_2(v)
%         \end{equation}
%     Return YES if the evaluation is nonzero, and NO otherwise.    
% \end{enumerate}
% \begin{lemma}
%     \label{lem:2dsp-DAG-correct}
%     The algorithm above correctly solves \textsf{2-DSP} with high probability.
% \end{lemma}
% \begin{proof}
% Let $F_{\disj}$ be the enumerating polynomial 
% Each pair of paths in the above set gives rise to a different monomial (since the edge sets of different pairs is different), so two disjoint shortest paths exist in $G$ if and only if $F_{\disj}$ is nonzero.
% We claim that 
%     \[F_{\disj} = L_1(t_1)L_2(t_2) + \sum_{v\in V} D(v)R_1(v)R_2(v).\]
% If we prove this result, then by the above discussion and the Schwartz-Zippel lemma 
% \end{proof}

\twodspDAG*
\begin{proof}
    Each pair of internally vertex-disjoint paths produces a distinct monomial in $F_{\disj}$.
    It follows that disjoint $(s_i,t_i)$-shortest paths exist in $G$ if and only if $F_{\disj}$ is nonzero as a polynomial.

    So we can solve \textsf{2-DSP} as follows.
    We assign each $x_{uv}$ variable an independent, uniform random element of $\mathbb{F}_{2^q}$, and then evaluate $F_{\disj}$ on this assignment.
    If the evaluation is nonzero we return YES (two disjoint shortest paths exist), and otherwise we return NO.
    
    If two disjoint shortest paths do not exist, then $F_{\disj}$ is the zero polynomial, and our algorithm correctly returns NO.
    If two disjoint shortest paths do exist, then $F_{\disj}$ is a nonzero polynomial of degree strictly less than $2n$, so by Schwartz-Zippel (\Cref{obvious}) our algorithm correctly returns YES with high probability for large enough $q = O(\log n)$.

    It remains to show that we can compute $F_{\disj}$ in linear time.

    First, we can compute $G_1$ and $G_2$ in linear time by \Cref{prop:sp-DAG}.
    
    Then, by dynamic programming forwards and backwards over the topological order of $G$, we can evaluate the polynomials  $L_i(v)$   and $R_i(v)$ for each $i\in [2]$ and vertex $v$ at our given assignment in linear time, using the recurrences from \Cref{lem:dp}.
    
    Having computed these values, \Cref{lem:DAG-left-disjoint} shows that for any vertex $v$ we can compute $D(v)$ at the given assignment in $O(\deg_{\text{in}}(v))$ time.
    So we can evaluate $D(v)$ for all $v$ in $O(m)$ time.
    
    Given the above evaluations, we can compute $F_{\cap}$ at the given point in $O(n)$  time by \Cref{lem:DAG-disj-poly}.
    
     Finally, given the value of $F_{\cap}$, we can evaluate $F_{\disj}$ in $O(1)$ additional time by \Cref{disj-to-intersect}.

     Thus we can solve \textsf{2-DSP} in linear time.
\end{proof}

\subsection{2-DSP in Undirected Graphs}
\label{subsec:2dsp-undir-alg}

In this subsection, we assume that $G$ is a weighted undirected graph.

To solve \textsf{2-DSP} in $G$, we will show how to efficiently evaluate a polynomial enumerating pairs of disjoint shortest paths in $G$.
To help construct this polynomial, it will first be helpful to characterize the ways in which two shortest paths can intersect in an undirected graph.

\subsubsection*{Structure of Shortest Paths}

We begin with the following simple observation about shortest paths in undirected graphs.

\begin{proposition}[Shortest Path Orderings]
\label{sp-orderings}
Let $G$ be an undirected graph.
    Suppose vertices $a,b,c$ appear in that order on some shortest path of $G$.
    Then on any shortest path in $G$ passing through these three nodes, $b$ appears between $a$ and $c$.    
\end{proposition}
\begin{proof}
    Since some shortest path in $G$ passes through vertices $a,b,c$ in that order, we know that $\dist(a,b)$ and $\dist(b,c)$ are both less than $\dist(a,c).$

    Now, consider any shortest path in $G$ which passes through these three vertices.
    If $a$ appears between $b$ and $c$ in this shortest path, then 
        \[\dist(a,c) < \dist(b,c)\]
    which contradicts the observation from the first sentence.
    
    Similarly, if $c$ appears between $a$ and $b$ in this shortest path, then
        \[\dist(a,c) < \dist(a,b)\]
    which again contradicts the observation  from the first sentence of this proof.

    Thus $b$ must appear between $a$ and $c$ on the shortest path as claimed.
\end{proof}
% \Cref{sp-orderings} lets us characterize the ways two shortest paths can intersect in an undirected graph.
% The following definitions will help us describe this characterization.

% Note that the proof of \Cref{sp-orderings} requires that $G$ has positive edge weights.

\begin{definition}[Intersection Types]
    Let $P_1$ and $P_2$ be intersecting shortest paths in an undirected graph.
    Let $v$ be the first vertex in $P_1$ appearing in $P_1\cap P_2$.
\begin{itemize}
    \item If $P_1\cap P_2 = \set{v}$, we say $P_1$ and $P_2$ have \emph{single 
    intersection}.
    \item If $|P_1\cap P_2|\ge 2$ and $v$ is also the first vertex of $P_2$ in $P_1\cap P_2$, we say $P_1$ and $P_2$ \emph{agree}.
    \item 
If $|P_1\cap P_2|\ge 2$ and $v$ is the last vertex of $P_2$ in $P_1\cap P_2$, we say $P_1$ and $P_2$ are \emph{reversing}.
\end{itemize}
If $P_1$ and $P_2$ do not agree, we say they \emph{disagree}.
\end{definition}

In other words, $P_1$ and $P_2$ disagree if the first vertex $v$ of $P_1$ appearing in $P_1\cap P_2$ is also the last vertex of $P_2$ appearing in $P_1\cap P_2$.
Equivalently, paths $P_1$ and $P_2$ disagree if they have single intersection or are reversing.

\begin{lemma}[Intersection Types are Exhaustive]
\label{cor:types}
      Let $P_1$ and $P_2$ be intersecting shortest paths in an undirected graph.
    Then either $P_1$ and $P_2$ agree, are reversing, or have single intersection.  
\end{lemma}
\begin{proof}
    If $|P_1\cap P_2| = 1$, then the paths have single intersection.

    Otherwise, $|P_1\cap P_2| \ge 2$.
Let $v$ be the first vertex in $P_1$ appearing in $P_1\cap P_2$.
It suffices to show $v$ is either the first or last vertex in $P_2$ appearing in $P_1\cap P_2$.

Suppose to the contrary this is not the case.
Then let $u$ and $w$ be the first and last vertices respectively in path $P_2$ appearing in $P_1\cap P_2$.
By assumption, $u,v,w$ are all distinct.
By definition, $u,v,w$ appear in that order on $P_2$.
Since $P_2$ is a shortest path, by \Cref{sp-orderings}, $v$ must appear between $u$ and $w$ on $P_1$ as well.
This contradicts the definition of $v$ as the first vertex on $P_1$ lying in the intersection $P_1\cap P_2$.
    
Thus our original supposition was false, and the paths $P_1$ and $P_2$ must agree or be reversing whenever $|P_1\cap P_2|\ge 2$, which proves the desired result.
\end{proof}
% Note that \Cref{cor:types} implies that any pair of intersecting shortest paths which disagree must have single intersection or reverse.
% This characterization will be useful in our algorithm for \textsf{2-DSP}.

\subsubsection*{Agreeing and Disagreeing Polynomials}

Let $\Fagree$ be the enumerating polynomial for the collection of standard pairs of paths $\langle P_1, P_2\rangle$ which agree.
Let $\Frev$ be the enumerating polynomial for the collection of standard pairs of paths $\langle P_1, P_2\rangle$ which disagree.
Any intersecting paths either agree or disagree, so
\begin{equation}
    \label{agree-taut}
    F_{\cap} = \Fagree + \Frev.
\end{equation}
Thus, to compute $F_{\cap}$, it suffices to compute $\Fagree$ and $\Frev$ separately.

\subsubsection*{Agreeing Paths}

We first show how to evaluate $\Fagree$ efficiently.

\begin{lemma}[Common Final Intersection] 
\label{last-intersection}
If paths $P_1$ and $P_2$ agree, then they have a common last intersection point, distinct from their common first intersection point.
\end{lemma}
\begin{proof}
    Since $P_1$ and $P_2$ agree, they have a common first intersection point at some vertex $v$.
    Then the vertex $w\in P_1\cap P_2$ which maximizes $\dist(v,w)$ must be the last node on both paths in $P_1\cap P_2$.
    Since the paths agree, we have $|P_1\cap P_2|\ge 2$, so $w\neq v$.
\end{proof}

\noindent We say a pair of paths $\langle P_1,P_2\rangle$ is \emph{edge-agreeing} if $P_i$ is an $(s_i,t_i)$-path in $G_i$, and $P_1$ and $P_2$ traverse a common edge in the same direction.

    \begin{lemma}[Edge-Agreeing $\subseteq$ Agreeing]
        \label{edge-agree}
        Any edge-agreeing pair is agreeing.
    \end{lemma}
    \begin{proof}
        Suppose pair $\langle P_1,P_2\rangle$ is edge-agreeing.
        
        Let $e = (a,b)$ be a common edge traversed  by both $P_1$ and $P_2$.

        Since $\set{a,b}\subseteq P_1\cap P_2$, we have $|P_1\cap P_2|\ge 2$.
        
    So by \Cref{cor:types}, $P_1$ and $P_2$ are either agreeing or reversing.

        Suppose to the contrary that $P_1$ and $P_2$ are reversing.
        Let $v$ be the first vertex on $P_1$ in $P_1\cap P_2$.
        Then $v\neq b$, since $a$ appears before $b$ on $P_1$.
        
        Since the paths are reversing, $v$ is the final vertex on $P_2$ in $P_1\cap P_2$.
        Then $v\neq a$, since $b$ appears after $a$ on $P_2$.
        Path $P_1$ passes through $v, a, b$ in that order.
        Then by \Cref{sp-orderings}, $a$ must appear between $v$ and $b$ on any shortest path containing these three vertices.
        However, this contradicts the fact that $b$ is between $a$ and $v$ on $P_2$.

        Thus $P_1$ and $P_2$ are not reversing, and so must agree as claimed.
    \end{proof}

\begin{lemma}
\label{lem:edge-agree}
    The polynomial $\Fagree$ enumerates the collection of edge-agreeing paths in $G$.
\end{lemma}
\begin{proof}
By \Cref{edge-agree}, the set of pairs $\Fagree$ enumerates includes all edge-agreeing pairs.

    Let $\ca{F}$ be the family of standard pairs of paths $\langle P_1,P_2\rangle$ such that $P_1$ and $P_2$ agree.
    
    By definition, $\Fagree$ enumerates $\ca{F}$.
    
    By \Cref{edge-agree}, every edge-agreeing standard pair of paths is in $\ca{F}$.

    Let $\ca{S}$ be the collection of pairs which are agreeing but not edge-agreeing.
    To prove the desired result, it suffices to show that the monomials corresponding to pairs in $\ca{S}$ have net zero contribution to $\Fagree$ over a field of characteristic two.

    To that end, take any pair $\langle P_1,P_2\rangle\in \ca{S}$.
    Since $\langle P_1,P_2\rangle$ is agreeing, $P_1$ and $P_2$ have a unique first intersection point $v$.
    By \Cref{last-intersection}, these paths also have a unique last intersection point $w\neq v$.

    Hence, we can decompose the paths into
    \[P_1 = P_1[s_1,v] \dia P_1[v,w] \dia P_1[w,t_1] \quad \text{and} \quad 
    P_2 = P_2[s_2,v] \dia P_2[v,w] \dia P_2[w,t_2].\]
    
    Define the paths 
        \[Q_1 = P_1[s_1,v] \dia P_2[v,w] \dia P_1[w,t_1] \quad \text{and} \quad 
    Q_2 = P_2[s_2,v] \dia P_1[v,w] \dia P_2[w,t_2]\]
    by swapping the $v$ to $w$ subpaths of $P_1$ and $P_2$.

    By \Cref{preserve-shortest-paths} the $Q_i$ are $(s_i,t_i)$-shortest paths.

        Since $P_1$ and $P_2$ are not edge-agreeing, the subpaths $P_1[v,w]$ and $P_2[v,w]$ are distinct.
        % so the pairs $\langle Q_1, Q_2\rangle$ and $\langle P_1, P_2\rangle$ are distinct.

    % Both $P_1[v,w]$ and $P_2[v,w]$ have length $\dist(v,w)$, so the $Q_i$ are  still $(s_i,t_i)$-shortest paths.

    The paths $Q_1$ and $Q_2$ have a common first intersection point of $v$, and so the pair $\langle Q_1, Q_2\rangle$ is agreeing.
    Since $\langle P_1,P_2\rangle$ is not edge-agreeing, neither is $\langle Q_1, Q_2\rangle$.
    Thus $\langle Q_1,Q_2\rangle\in \ca{S}$.
    % Also
        % \[f(P_1,P_2) = f(Q_1,Q_2)\]
    % since both pairs traverse the same multisets of edges.
    Finally, if we swap the $v$ to $w$ subpaths of $Q_1$ and $Q_2$, we recover $P_1$ and $P_2$.

    The above discussion implies that the map $\Phi$ from $\langle P_1, P_2\rangle$ to $\langle Q_1, Q_2\rangle$ described above
    % matches pairs from $\ca{S}$ into groups of two which produce the same monomials.
    % It follows that over a field of characteristic two, the sum of the monomials of all pairs in $\ca{S}$ is zero.
    % So only edge-agreeing pairs of paths have nonzero contribution to $\Fagree$, which proves the desired result.
    satisfies all conditions of \Cref{lem:subpath-swapping-cancels}, and so $\Fagree$ is the enumerating polynomial for $\ca{F}\setminus\ca{S}$, which is precisely the set of edge-agreeing pairs.
\end{proof}

% Actually, you only some in edge occurring in both G1 and G2, need to change that. 
\begin{lemma}[Enumerating Agreeing Pairs]
\label{enum:agreeing-pairs}
    We have 
        \[\Fagree = 
        \sum_{v\in V}
        \sum_{w\in \vout(v)} \grp{D(v)x_{vw}^2} R_1(w)R_2(w).\]
\end{lemma}
\begin{proof}
Let $\ca{F}$ be the family of edge-agreeing pairs.
    By \Cref{lem:edge-agree}, it suffices to show that 
    \[\sum_{v\in V}
        \sum_{w\in \vout(v)} \grp{D(v)x_{vw}^2} R_1(w)R_2(w)\]
    is the enumerating polynomial for $\ca{F}$.
    To that end, the following claim about the individual terms of the above sum will be useful.

    \begin{claim}
        \label{first-edge-overlap}
        For any choice of vertices $v$ and $w$ with $w\in\vout(v)$,
        the polynomial
            \begin{equation}
            \label{eq:single-addend-agreeing}
            \grp{D(v)x_{vw}^2} R_1(w)R_2(w)
            \end{equation}
            enumerates all standard pairs of paths $\langle P_1, P_2\rangle$ such that $P_1$ and $P_2$ overlap at edge $e = (v,w)$, and have common first intersection at vertex $v$.
    \end{claim}
    \begin{claimproof}
       Take any pair $\langle P_1,P_2\rangle$ satisfying the conditions from the statement of the claim.
       Then we can decompose these paths in the form
        \[P_i = P_i[s_i,v]\dia (v,w) \dia P_i[w,t_i]\]
        such that the $P_i[s_i,v]$ subpaths  intersect only at $v$.
        
            This means that the pair $\langle P_1[s_1,v], P_2[s_2,v]\rangle$ is enumerated by $D(v)$, the two edges $(v,w)$ are enumerated by $x_{vw}^2$, and each path $P_i[w,t_i]$ is enumerated by $R_i(w)$, so that the expansion of the polynomial \cref{eq:single-addend-agreeing}
            includes some monomial corresponding to the given pair.

            Conversely, any monomial in the expansion of \cref{eq:single-addend-agreeing} is the product of monomials recording the edges traversed by a pair of paths $\langle A_1, A_2\rangle$ only intersecting at node $v$, where $A_i$ is an $(s_i,v)$-path in $G_i$, two copies of the edge $(v,w)$, and some $(w,t_i)$-paths $B_i$ in $G_i$.
            This product is equal to the monomial given by  
                \[f(A_1\dia (v,w)\dia B_1, A_2\dia (v,w)\dia B_2).\]
            Define the paths $P_i = A_i\dia (v,w)\dia B_i$ for each $i\in [2]$.
            Since $A_i$ and $B_i$ are paths in $G_i$, and $(v,w)$ is an edge in both $G_1$ and $G_2$, we know that each $P_i$ is an $(s_i,t_i)$-shortest path.
                
            We claim that $A_1$ does not intersect $B_2$.
            Indeed, suppose to the contrary that $A_1$ and $B_2$ intersect at some node $u$.
            Then $P_1$ is a shortest path which passes through nodes $u,v,w$ in that order, yet $P_2$ is a shortest path which passes through $v,w,u$ in that order, which contradicts \Cref{sp-orderings}.

            Hence $A_1$ does not intersect $B_2$.
            Symmetric reasoning shows that $A_2$ does not intersect $B_1$.
            
            Thus the paths $P_1$ and $P_2$ have common first intersection at $v$.
            Then the pair $\pair{P_1,P_2}$ satisfies the conditions from the claim statement, so the polynomial from \cref{eq:single-addend-agreeing} enumerates all pairs of paths described in the claim.
    \end{claimproof}
    By \Cref{first-edge-overlap}, the sum 
   \begin{equation}
   \label{eq:edge-agree-rhs}
    \sum_{v\in V}
        \sum_{w\in \vout(v)} \grp{D(v)x_{vw}^2} R_1(w)R_2(w)
   \end{equation}
   enumerates all edge-agreeing pairs  whose first intersection point $v$ is the beginning of an edge traversed by both paths in the pair.
   Let $\ca{H}$ be the set of such pairs, and let $\ca{S}=\ca{F}\setminus\ca{H}$.

   We claim that the set of pairs in $\ca{S}$ has net zero contribution to $\Fagree$.
   
   Indeed, take any pair $\langle P_1, P_2\rangle\in \ca{S}$.
    Since the pair is edge-agreeing, by \Cref{edge-agree} the pair is agreeing.
    Hence 
    $P_1$ and $P_2$ have a common first intersection at some node $v$.
    By \Cref{last-intersection}, these paths also have a common last intersection point at some node $w\neq v$.
    Since the pair is not in $\ca{H}$, the subpaths $P_1[v,w]$ and $P_2[v,w]$ are distinct.
    If we swap these subpaths 
    to produce paths 
        \[Q_1 = P_1[s_1,v] \dia P_2[v,w]\dia P_1[w,t_1]\quad\text{and}\quad 
        Q_2 = P_2[s_2,v]\dia P_1[v,w]\dia P_2[w,t_2]\]
    then by \Cref{preserve-shortest-paths} each $Q_i$ is still an $(s_i,t_i)$-path in $G_i$.
    % because $P_1[v,w]$ and $P_2[v,w]$ both have the same length.
    Since the $P_i$ are edge-agreeing, and this edge-overlap must occur in their $v$ to $w$ subpaths, the  $Q_i$ are also edge-agreeing.
    By assumption, $v$ is not the beginning of a common edge traversed by the $P_i$, so it is also not such a beginning for $Q_i$.
    These observations combined show that $\langle Q_1,Q_2\rangle\in\ca{S}$.
    Also, swapping $v$ to $w$ subpaths in the $Q_i$ recover the $P_i$ paths.
    It follows that the map $\Phi$ sending $\langle P_1,P_2\rangle$ to $\langle Q_1,Q_2\rangle$ satisfies all the conditions of \Cref{lem:subpath-swapping-cancels}, so in fact $\Fagree$ enumerates $\ca{F}\setminus\ca{S} = \ca{H}$.
    
    Since \cref{eq:edge-agree-rhs} also enumerates $\ca{H}$, this proves the desired result.
\end{proof}

\subsubsection*{Disagreeing Paths}

Let $\Fdis$ be the family of disagreeing, standard pairs of paths.
Recall that $\Frev$ enumerates $\Fdis$.
Our goal is to show how to efficiently evaluate $\Frev$.
To do this, we describe a subfamily of pairs in $\Fdis$, and argue that $\Frev$ enumerates this subfamily.

\begin{definition}
\label{def:local-relaxation}
For each node $v$, let $\tilde{\ca{B}}(v)$ be the set of standard pairs of paths $\langle P_1,P_2\rangle$ intersecting at $v$, such that 
    if we let $a_i$ and $b_i$ denote the nodes appearing immediately before and after $v$ on $P_i$ respectively, then
    \begin{enumerate}
    \item $a_1\neq a_2$, 
    \item $b_1\neq b_2$, and 
    \item $a_1\neq b_2$.
\end{enumerate}
Let $\ca{B}(v)\subseteq\tilde{\ca{B}}(v)$ be the subfamily of such pairs such that $v$ is the first vertex in $P_1$ lying in $P_1\cap P_2$.
Then define the collection
    \[\ca{B} = \bigcup_{v\in V} \ca{B}(v).\]
% Finally, let $\Bdis = \ca{B}\cap \Fdis$ be the family of disagreeing pairs in $\ca{B}$.
\end{definition}

Next, we prove some lemmas, which will help us prove that to enumerate $\Fdis$, it suffices to design enumerating polynomials for $\tilde{B}(v)$ for each vertex $v$. 

\begin{lemma}
\label{lem:relax-B}
    For each vertex $v$, the  enumerating polynomial for $\tilde{\ca{B}}(v)$ enumerates $\ca{B}(v)$.
\end{lemma}
\begin{proof}
    Fix a vertex $v$.
    Let $\ca{S} = \tilde{\ca{B}}(v)\setminus\ca{B}(v)$.
    % We will show that monomials from pairs in $\ca{S}$ have net zero contribution the enumerating polynomial for $\ca{B}(v)$.

    Take any pair $\langle P_1,P_2\rangle\in\ca{S}$.
    Let $u$ be the first vertex in $P_1$ lying in $P_1\cap P_2$.
    Since the pair is not in $\ca{B}(v)$, we know that $u\neq v$.
    
    There are two cases to consider, based off the relative positions of $u$ and $v$ on $P_2$.

    \paragraph{Case 1: $\boldsymbol{u}$ before $\boldsymbol{v}$}  Suppose that $u$ appears before $v$ in $P_2$.
    In this case, we  form paths 
        \[Q_1 = P_1[s_1,u]\dia P_2[u,v]\dia P_1[v,t_1]\quad\text{and}\quad Q_2 = P_2[s_2,u]\dia P_1[u,v]\dia P_2[v,t_2]\]
    by swapping the $u$ to $v$ subpaths in $P_1$ and $P_2$.
    By \Cref{preserve-shortest-paths}, each $Q_i$ is still an $(s_i,t_i)$-path in $G_i$.
    After this subpath swap, the vertices immediately before  $v$ in the $Q_i$ are still distinct, and the vertices immediately after $v$ in the $Q_i$ are still distinct.
    Moreover, the vertex immediately before $v$ in $Q_1$ is not equal to the vertex immediately after $v$ in $Q_2$, since these are distinct nodes on  $P_2$.
    Then by \Cref{def:local-relaxation}, $\langle Q_1,Q_2\rangle\in\tilde{\ca{B}}(v)$.
    Also, since $Q_1$ and $Q_2$ intersect at $u$, $\langle Q_1,Q_2\rangle\not\in\ca{B}(v)$.
    Finally, since $u$ is the first vertex in $Q_1$ lying in $Q_1\cap Q_2$, applying the same subpath swap procedure as above to $\langle Q_1,Q_2\rangle$ recovers $\langle P_1,P_2\rangle$.
    
    So the map $\Phi$ sending $\langle P_1,P_2\rangle$ to $\langle Q_1,Q_2\rangle$ as above satisfies all conditions of \Cref{lem:subpath-swapping-cancels}.

    \paragraph{Case 2: $\boldsymbol{u}$ after $\boldsymbol{v}$} Suppose that $u$ appears after $v$ on $P_2$.
        In this case, we form paths 
        \[Q_1 = P_1[s_1,u]\dia \rev{P_2}[u,v]\dia P_1[v,t_i]\quad\text{and}\quad 
        Q_2 = P_2[s_1,v]\dia \rev{P_1}[v,u]\dia P_2[u,t_i]\]
    by swapping the $u$ to $v$ subpaths in $P_1$ and $P_2$ (and using the fact that $G$ is undirected, so we can traverse the edges in these subpaths backwards).
    Since the $u$ to $v$ subpaths of $P_1$  and $P_2$ both have length $\dist(u,v)$, each path $Q_i$ has the same length as $P_i$, and is thus still an $(s_i,t_i)$-path in $G_i$.

    Let $a_i$ and $b_i$ denote the vertices on $P_i$ immediately before and after $v$ respectively.

    % Since $\langle P_1,P_2\rangle\in\ca{S}$, by \Cref{def:local-relaxation} we know that $b_1\neq b_2$ and $a_1\not\in\set{a_2,b_2}$.
    Then the nodes immediately before and after $v$ on $Q_1$ are $b_2$ and $b_1$ respectively, and the nodes immediately before and after $v$ on $Q_2$ are $a_2$ and $a_1$.
    
    We know that $b_i\neq a_i$, since $a_i$ and $b_i$ are distinct vertices on the path $P_i$.
    Thus $\langle Q_1,Q_2\rangle$ satisfies the first two conditions from \Cref{def:local-relaxation}.
    
    Since $\langle P_1,P_2\rangle\in\ca{S}$, by condition three of \Cref{def:local-relaxation}, we know that $b_2\neq a_1$.
    Since $b_2$ is the node immediately before $v$ in $Q_1$ and $a_1$ is the node immediately after $v$ in $Q_2$, we see that $\langle Q_1,Q_2\rangle$ satisfies condition three of \Cref{def:local-relaxation} as well.
    
Thus  $\langle Q_1,Q_2\rangle\in\tilde{\ca{B}}(v)$.

    Furthermore, since $u\in Q_1\cap Q_2$ appears before $v$ in $Q_1$, we have $\pair{Q_1,Q_2}\not\in\ca{B}(v)$.

    Thus $\pair{Q_1,Q_2}\in\ca{S}$.

    This shows that the swapping procedure described above, sending $\pair{P_1,P_2}$ to $\pair{Q_1,Q_2}$, is a map from $\ca{S}$ to itself.
    Moreover, because $u$ is the first vertex on $P_1$ lying in $P_1\cap P_2$ and the first vertex on $Q_1$ lying in $Q_1\cap Q_2$, performing the above swapping procedure on $\pair{Q_1,Q_2}$ recovers $\pair{P_1,P_2}$.
    Finally, since the multisets of edges traversed by $\pair{P_1,P_2}$ and $\pair{Q_1,Q_2}$ are the same, we have 
        \[f(P_1,P_2) = f(Q_1,Q_2).\]
    The above discussion shows that $\ca{S}$ can be partitioned into groups of size two, with each group consisting of two pairs with the same monomial.
    Then the monomials from pairs in $\ca{S}$ have net zero contribution to the enumerating polynomial for $\tilde{\ca{B}}(v)$ modulo two.
    
    Hence the enumerating polynomial for $\tilde{\ca{B}}(v)$ actually enumerates $\tilde{\ca{B}}(v)\setminus\ca{S} = \ca{B}(v)$, as claimed.
\end{proof}

\begin{lemma}
\label{lem:relax-disagreeing-paths}
    The enumerating polynomial for $\ca{B}$ is $\Frev$.
\end{lemma}
\begin{proof}
    First, observe that $\Fdis\sub\ca{B}$.
    Indeed, suppose $\langle P_1,P_2\rangle\in\Fdis$ is a standard pair of paths which disagree.
    Let $v$ be the first vertex of $P_1$ in $P_1\cap P_2$.
    Let $a_i$ and $b_i$ be the vertices immediately before and after $v$ in $P_i$.
    Then $a_1\not\in\set{a_2,b_2}$ by the definition of $v$, and $b_1\neq b_2$ because $P_1$ and $P_2$ disagree.
    Thus $\langle P_1,P_2\rangle$ satisfies all the conditions from \Cref{def:local-relaxation}, so this pair is in $\ca{B}$.
    Since we chose an arbitrary pair $\langle P_1,P_2\rangle$ from $\Fdis$, we have $\Fdis\sub\ca{B}$ as claimed.

    Let $\ca{S} = \ca{B}\setminus \Fdis$ be the set of agreeing pairs in $\ca{B}$.

    Take any pair $\langle P_1,P_2\rangle\in\ca{S}$.
    Let $v$ be the first common intersection point of the pair.
    Let $w$ be the last common intersection point of the pair.
    By \Cref{last-intersection}, $w\neq v$.
    Let
        \[Q_1 = P_1[s_1,v]\dia P_2[v,w]\dia P_1[w,t_1]
        \quad\text{and}
        \quad
        Q_2 = P_2[s_1,v]\dia P_1[v,w]\dia P_2[w,t_1]
        \]
    be the paths formed by swapping the $v$ to $w$ subpaths in $P_1$ and $P_2$.

    We claim that $\langle Q_1,Q_2\rangle\in\ca{S}$ as well.

    % Indeed, since each $P_i[v,w]$ has length $\dist(v,w)$, each $Q_i$ path is still an $(s_i,t_i)$-path in $G_i$.
    By \Cref{preserve-shortest-paths} each $Q_i$ is still an $(s_i,t_i)$-path in $G_i$.
    The first and last intersection points of these paths are still $v$ and $w$, so the new pair is still agreeing.
    
    Let vertices $a_i$ and $b_i$ be vertices immediately before and after $v$ on each $P_i$ respectively.
    Since the $a_i$ occur before $v$, the nodes before $v$ on $Q_1$ and $Q_2$ are distinct.
    Since the nodes after $v$ on $Q_1$ and $Q_2$ are $\set{b_1, b_2}$, these nodes are also distinct.
    Finally, the node before $v$ on $Q_1$ is $a_1$  and the node after $v$ on $Q_2$ is $b_1$, and $a_1\neq b_1$ since $a_1$ and $b_1$ are distinct vertices on path $P_1$, the node before $v$ on $Q_1$ and after $v$ on $Q_2$ are distinct.
    
    The discussion in the previous paragraph shows that $\langle Q_1,Q_2\rangle\in\ca{S}$.

    Finally, if we apply the subpath swapping procedure above to $\langle Q_1,Q_2\rangle$, we recover $\langle P_1,P_2\rangle$.
    Then the map $\Phi$ on $\ca{S}$ sending $\langle P_1,P_2\rangle$ to $\langle Q_1,Q_2\rangle$ satisfies the conditions from the statement of \Cref{lem:subpath-swapping-cancels}, so the enumerating polynomial for $\ca{B}$ in fact enumerates $\ca{B}\setminus\ca{S} = \Fdis$ as claimed.
\end{proof}

By \Cref{lem:relax-B,lem:relax-disagreeing-paths}, enumerating $\Fdis$ reduces to enumerating $\tilde{\mathcal{B}}(v)$ for each vertex $v$.
Our next goal is to  perform this enumeration efficiently.
To do this, we start by defining and establishing formulas for  some additional helper polynomials.

\begin{definition}[Relaxed Target Linkages]
\label{def:relaxed-target}
For each vertex $v$, let $T(v)$ be the enumerating polynomial for the set of pairs of paths $\langle P_1,P_2\rangle$ where 
\begin{enumerate}
    \item each $P_i$ is a $(v,t_i)$-path in $G_i$, and
    \item the second nodes of $P_1$ and $P_2$ are distinct.
\end{enumerate} 
\end{definition}

\begin{lemma}
\label{target-relax}
    For each vertex $v$, we have 
       \[T(v) = R_1(v)R_2(v) - \sum_{w\in\vout(v)}x_{vw}^2R_1(w)R_2(w).\]
\end{lemma}
\begin{proof}
    This follows from symmetric reasoning to the proof of \Cref{enum:relax-linkage}.
\end{proof}

\begin{definition}
\label{def:error-term}
    For each vertex $v$, let $H(v)$ be the enumerating polynomial for the set of standard pairs of paths $\langle P_1,P_2\rangle$ intersecting at $v$, such that the vertex immediately before $v$ on $P_1$ is the same as the vertex immediately after $v$ on $P_2$.
\end{definition}

Intuitively, $H(v)$ is a polynomial we use to enforce the $a_1\neq b_2$ condition from \Cref{def:local-relaxation}.

\begin{lemma}
\label{lem:error-computation}
    For each vertex $v$, we have 
    % \[H(v) = \sum_{u\in\vmix(v)} L_1(u)x_{uv}\grp{L_2(v) - L_2(u)x_{uv}}x_{vu}R_1(v)R_2(v).\]    
    \[H(v) = \sum_{u\in\vmix(v)} L_1(u)x_{uv}R_1(v)L_2(v)x_{vu}R_2(u).\]
\end{lemma}
\begin{proof}
For any pair of paths $\langle P_1,P_2\rangle$ satisfying the conditions from \Cref{def:error-term}, there exists a unique vertex $u$ such that $u$ appears immediately before $v$ on $P_1$ and immediately after $v$ on $P_2$.
Any such $u$ must lie in $\vin^1(v)\cap \vout^2(v) = \vmix(v)$ by definition.

Given $v\in V$ and $u\in\vmix(v)$,
let $\ca{F}(u,v)$ be the set of all standard pairs of paths $\langle P_1,P_2\rangle$ such that $P_1$ traverses edge $(u,v)$ and $P_2$ traverses edge $(v,u)$.
The discussion in the previous paragraph implies that to prove the lemma, it suffices to show that the polynomial 
\begin{equation}
    \label{eq:error:addend}
    L_1(u)x_{uv}R_1(v)\cdot L_2(v)x_{vu}R_2(u)
\end{equation}
enumerates $\ca{F}(u,v)$ for all $v$ and $u$.

To that end, let $\pair{P_1,P_2}\in\ca{F}(u,v)$.
We claim the monomial $f(P_1,P_2)$ appears in the expansion of \cref{eq:error:addend}.
Indeed, by definition, we can split 
    \[P_1 = P_1[s_1,u]\dia (u,v)\dia P_1[v,t_1]\quad\text{and}\quad
    P_2 = P_2[s_2,v]\dia (v,u)\dia P_2[u,t_2].\]
Then path $P_1[s,u]$ is enumerated by $L_1(u)$, edge $(u,v)$ is enumerated by $x_{uv}$, path $P_1[v,t_1]$ is enumerated by $R_1(v)$, path $P_2[s,v]$ is enumerated by $L_2(v)$, edge $(v,u)$ is enumerated by $x_{vu}$, and path $P_2[u,t_2]$ is enumerated by $R_2(u)$, so the expansion of \cref{eq:error:addend} has the term $f(P_1,P_2)$.

Conversely, we claim that any monomial in the expansion of \cref{eq:error:addend} corresponds to a pair of paths in $\ca{F}(u,v)$.
Indeed, any such monomial is the product of monomials from each of the factors in \cref{eq:error:addend}.
By definition, $L_1(u)$ enumerates $(s_1,u)$-paths $A_1$ in $G_1$, $R_1(v)$ enumerates $(v,t_1)$-paths $B_1$ in $G_1$, $L_2(v)$ enumerates $(s_2,v)$-paths $A_2$ in $G_2$, and $R_2(u)$ enumerates $(u,t_2)$-paths $B_2$ in $G_2$.
So an arbitrary monomial in \cref{eq:error:addend} is of the form 
    \[f(A_1)x_{uv}f(B_1)\cdot f(A_2)x_{vu}f(B_2)\]
for $A_i$ and $B_i$ satisfying the conditions above.

Since $u\in\vmix(v)$, if we define paths 
    \[P_1 = A_1\dia (u,v)\dia B_1\quad\text{and}\quad P_2=A_2\dia (v,u)\dia B_2\]
then we see that $P_i$ is an $(s_i,t_i)$-path in $G_i$.
Moreover, $P_1$ traverses $(u,v)$ and $P_2$ traverses $(v,u)$.
So $\pair{P_1,P_2}\in\ca{F}(u,v)$, and since 
    \[f(P_1,P_2) = f(P_1)\cdot f(P_2) = f(A_1)x_{uv}f(B_1)\cdot f(A_2)x_{vu}f(B_2)\]
we get that monomials of \cref{eq:error:addend} correspond to pairs in $\ca{F}(u,v)$.
Thus \cref{eq:error:addend} enumerates $\ca{F}(u,v)$.

By the discussion at the beginning of the proof, this proves the claim.
\end{proof}

\begin{lemma}[Enumerating Disagreeing Pairs]
\label{enum:disgreeing-paths}
    We have 
    \[\Frev = \sum_{v\in V} \grp{D(v)T(v) - H(v) }.\]
\end{lemma}
\begin{proof}
    By \Cref{lem:relax-disagreeing-paths}, $\Fdis$ is the enumerating polynomial for $\ca{B}$.
    By \Cref{def:local-relaxation}, the enumerating polynomial for $\ca{B}$ is the sum of the enumerating polynomials for $\ca{B}(v)$ over all vertices $v$
    (here, we are using the fact that the $\ca{B}(v)$ are disjoint for each $v$ by definition).
    Applying \Cref{lem:relax-B}, we see this sum is equal to the sum of the enumerating polynomials for $\tilde{\ca{B}}(v)$ over all vertices $v$.
    So, it suffices to show that for each vertex $v$, the polynomial
        \begin{equation}
\label{remove:condition}
        D(v)T(v) - H(v)
        \end{equation}
    enumerates $\tilde{\mathcal{B}}(v)$.

    By \Cref{linkage-relaxation,enum:relax-linkage}, the polynomial $D(v)$ enumerates all pairs $\langle A_1,A_2\rangle$ where each $A_i$ is an $(s_i,v)$-path in $G_i$, and the vertices $a_i$ immediately before $v$ in $A_i$ are distinct.
    By \Cref{def:relaxed-target}, the polynomial $T(v)$ enumerates all pairs $\langle B_1,B_2\rangle$ where each $B_i$ is a $(v,t_i)$-path in $G_i$, and the  vertices $b_i$ after $v$ in $B_i$ are distinct.
    
    Then the product $D(v)T(v)$ enumerates all pairs of paths \[\langle A_1\dia B_1,A_2\dia B_2\rangle\] for $A_i$ and $B_i$ satisfying the conditions from the previous paragraph.
    Note that this pair satisfies conditions 1 and 2 from \Cref{def:local-relaxation}.
    Moreover, any standard pair of paths $\langle P_1,P_2\rangle$ satisfying conditions 1 and 2 from \Cref{def:local-relaxation} can be decomposed into 
        \[P_i = A_i\dia B_i\]
    where $A_i$ is an $(s_i,v)$-path in $G_i$, $B_i$ is a $(v,t_i)$-path in $G_i$, the penultimate vertices of the $A_i$  are distinct, and the second vertices of the $B_i$ are distinct.

    Thus $D(v)T(v)$ enumerates precisely all standard pairs of paths satisfying the first two conditions from \Cref{def:local-relaxation}.

    By \Cref{def:error-term}, $H(v)$ enumerates the standard pairs of paths $\langle P_1,P_2\rangle$ intersecting at $v$, such that the node immediately before $v$ on $P_1$ is the same as the node immediately after $v$ on $P_2$.
    If we let $a_i$ and $b_i$ again denote the nodes immediately before and after $v$ on $P_i$ respectively, this means we are enumerating all standard pairs of paths intersecting at $v$ such that $a_1 = b_2$.
    Note that all such paths also have $a_1\neq a_2$ (since $a_2$ and $b_2$ are distinct nodes in $P_2$) and $b_1\neq b_2$ (since $a_1$ and $b_1$ are distinct nodes in $P_1$).

    Thus $H(v)$ enumerates precisely the standard pairs of paths which satisfy the first two conditions but fail the third condition of \Cref{def:local-relaxation}.

    Consequently, the difference \cref{remove:condition} enumerates the standard pairs of paths which satisfy all conditions from \Cref{def:local-relaxation}.
    Thus the polynomial from \cref{remove:condition}   enumerates $\tilde{B}(v)$ as claimed, which proves the desired result.
\end{proof}

% Given a vertex $v$, let $\ca{B}(v)$ be the set of intersecting pairs of paths $\langle P_1,P_2\rangle$ where $P_i$ is an $(s_i,t_i)$-path in $G_i$ such that if we let $a_i$ be the node immediately before $v$ in $P_i$ and $b_i$ be the node immediately after $v$ on $P_i$, then
% $a_1\neq a_2$, $b_1\neq b_2$, and $a_1\neq b_2$.
% Let $B(v)$ be the enumerating polynomial for $\ca{B}(v)$.

% \begin{lemma}
%     We have 
%         \[\Frev = \sum_{v\in V} B(v).\]
% \end{lemma}
% \begin{proof}
%     Throughout, let $P_i$ denote an $(s_i,t_i)$-path in $G_i$.

%     Any disagreeing pair of paths $\langle P_1, P_2\rangle$ has a unique vertex $v$.
    
% \end{proof}

We are finally ready to prove our main theorem.
Note that the first part of the proof is completely identical to the proof of \Cref{2dsp-DAG}, since our algorithms for \textsf{2-DSP} in DAGs and undirected graphs have the same overall structure.

\twodspundir*
\begin{proof}
% The beginning of this proof is an exact copy of the beginning of the proof for Thm 2 (2-DSP in a DAG in linear time).
    Each pair of internally vertex-disjoint paths produces a distinct monomial in $F_{\disj}$.
    It follows that disjoint $(s_i,t_i)$-shortest paths exist in $G$ if and only if $F_{\disj}$ is nonzero as a polynomial.

    So we can solve \textsf{2-DSP} as follows.
    We assign each $x_{uv}$ variable an independent, uniform random element of $\mathbb{F}_{2^q}$, and then evaluate $F_{\disj}$ on this assignment.
    If the evaluation is nonzero we return YES (two disjoint shortest paths exist), and otherwise we return NO.
    
    If two disjoint shortest paths do not exist, then $F_{\disj}$ is the zero polynomial, and our algorithm correctly returns NO.
    If two disjoint shortest paths do exist, then $F_{\disj}$ is a nonzero polynomial of degree strictly less than $2n$, so by Schwartz-Zippel (\Cref{obvious}) our algorithm correctly returns YES with high probability for large enough $q = O(\log n)$.

    It remains to show that we can compute $F_{\disj}$ in linear time.

        First, we can compute $G_1$ and $G_2$ in linear time by \Cref{prop:sp-DAG}.

            Then, by dynamic programming forwards and backwards over the topological order of $G$, we can evaluate the polynomials  $L_i(v)$   and $R_i(v)$ for each $i\in\set{1,2}$ and vertex $v$ at our given assignment in linear time, using the recurrences from \Cref{lem:dp}.

    %% Things change!
    
    Having computed these values, \Cref{lem:DAG-left-disjoint,target-relax} shows that for any vertex $v$ we can compute $D(v)$ and $T(v)$ at the given assignment in $O(\deg_{\text{in}}(v))$ and $O(\deg_{\text{out}}(v))$ time respectively.
    So we can evaluate $D(v)$ and $T(v)$ for all $v$ in $O(m)$ time.

    Then by \Cref{enum:agreeing-pairs}, we can evaluate $\Fagree$ in $O(m)$ time.

    By \Cref{lem:error-computation}, we can compute $H(v)$ at any given vertex $v$ in $O(\deg_{\text{in}}(v))$ time.
    So we can evaluate $H(v)$ for all vertices $v$ in $O(m)$ time.

    From the values for $D(v)$, $T(v)$, and $H(v)$, by \Cref{enum:disgreeing-paths} we can compute $\Frev$ in $O(n)$ time. 

    Having computed $\Fagree$ and $\Frev$, by \cref{agree-taut} we can compute $F_{\cap}$ as well in $O(1)$ time.

    Finally, given the value of $F_{\cap}$, we can evaluate $F_{\disj}$ in $O(1)$ additional time by \Cref{disj-to-intersect}.

     Thus we can solve \textsf{2-DSP} in linear time.
    \end{proof}

\subsection{Search to Decision Reduction}
\search*
% \begin{theorem}
%     We can solve \textsf{2-DSP} over weighted DAGs and undirected graphs with $n$ vertices and $m$ edges, and find a solution if it exists, in $O(mn)$ time.
% \end{theorem}
\begin{proof}
    The proofs of \Cref{2dsp-undir,2dsp-DAG} construct arithmetic circuits of size $O(m)$ for the two disjoint shortest paths polynomial $F_{\disj}$ for DAGs and undirected graphs.
    By the Baur-Strassen theorem (see \cite{BaurStrassen1983} for the original proof, and \cite[Theorem 2.5]{ShpilkaYehudayoff2009} for a more recent exposition), we can construct arithmetic circuits of size $O(m)$ which simultaneously compute all single-order partial derivatives of $F_{\disj}$.

    Given an edge $(u,v)$, the polynomial $(\partial /\partial x_{uv})\, F_{\disj}$ is nonzero if and only if  edge $(u,v)$ appears in some solution to the \textsf{2-DSP} problem.
    So by Schwartz-Zippel (\Cref{obvious}), with high probability,  $(\partial /\partial x_{uv})\, F_{\disj}$ has nonzero evaluation at a uniform random assignment over $\mathbb{F}_{2^q}$ if and only if $(u,v)$ is an edge occurring in some pair of disjoint shortest paths, for $q = O(\log n)$ sufficiently large.

    We can compute all partial derivatives at some random evaluation point in $O(m)$ time using the arithmetic circuit for these polynomials.
    We pick an edge $(s_1, v)$ such that $(\partial /\partial x_{s_1v})\, F_{\disj}$ has nonzero evaluation.
    Then we delete vertex $s_1$ from $G$, and consider a smaller instance of \textsf{2-DSP} on the graph, where source $s_1$ is replaced with $v$.
    We can repeat this process on the new instance, to find the first edge on a $(v,t_1)$-shortest path which is disjoint from some $(s_2,t_2)$-shortest path.
    Repeating this process at most $n$ times, we can recover an $(s_1,t_1)$-shortest path $P_1$, which is disjoint from some $(s_2,t_2)$-shortest path.
    
    At this point, we just delete all vertices of $P_1$ from the original graph $G$, find an $(s_2,t_2)$-shortest path $P_2$ in the resulting graph in linear time, and then return $\langle P_1,P_2\rangle$ as our answer.

    Overall, we compute at most $n$ evaluations of arithmetic circuits of size $O(m)$, so the algorithm runs in $O(mn)$ time.
\end{proof}

\section{Edge-Disjoint Paths Algorithm}\label{sec:edge}

In this section, we present our algorithm for \textsf{$k$-EDSP} in weighted DAGs and prove \Cref{thm:edsp}. 
The algorithm works by constructing a large graph on $n^k$ nodes, whose paths correspond to collections of edge-disjoint paths in the original graph.
As we previously mentioned in \Cref{sec:intro,sec:overview}, our algorithm uses the same framework as that of \cite{FortuneHopcroftWyllie1980}.

Throughout this section, we let $G$ be the input DAG on $n$ nodes and $m$ edges.
For each $i\in [k]$, we let $G_i$ denote the $s_i$-shortest paths DAG of $G$.

% \subsection{Disjoint Paths Graph}
% \label{subsec:k-EDSP-graph-construction}

    Fix a topological order $(\prec)$ of $G$.
    
    We construct a graph $G'$ on $n^k$ nodes, whose paths encode $k$-tuples of paths in $G$.

    The graph $G'$ has a node for each $k$-tuple of vertices from $G$.

    Given a node $\vec{v} = (v_1, \dots, v_k)$ in $G'$ (where each $v_i$ is a vertex in $G$), let $v = \early(\vec{v})$ be the unique vertex $v$ such that 
        \begin{itemize}
            \item $v = v_i$ for some index $i\in [k]$, and 
            \item $v\preceq v_j$ for all $j\in [k]$.
        \end{itemize}
    In other words, $\early(\vec{v})$ is the earliest coordinate of $\vec{v}$ with respect to the topological order of $G$.
    Let $I(\vec{v})$ denote the set of indices $i\in [k]$ such that $v_i = \early(\vec{v})$.
    Then we include an edge from node $\vec{v} = (v_1, \dots, v_k)$ to node $\vec{w} = (w_1,\dots, w_k)$ in $G'$ precisely when 

    \begin{enumerate}
        \item for all $i\in I(\vec{v})$, the pair $(v_i,w_i)$ is an edge in $G_i$,
        \item the vertices $w_i$ are pairwise distinct over all $i\in I(\vec{v})$, and 
        \item for all $j\not\in I(\vec{v})$, we have $w_j = v_j$.
    \end{enumerate}
        We refer to the rules above as the \emph{conditions for edges of $G'$}.
    Intuitively, the conditions say that we can move from a node $\vec{v}$ consisting of $k$ vertices $v_i$ by stepping from all of the earliest vertices to new distinct vertices.
    An example of these edge transitions is depicted in \Cref{fig:k-EDSP}.

    Let $\vec{s} = (s_1, \dots, s_k)$ and $\vec{t} = (t_1, \dots, t_k)$ be nodes in $G'$ containing all of the source and target vertices in $G$ respectively.
    As the following lemmas show, our conditions for including edges in $G'$ allow us to relate paths from $\vec{s}$ to $\vec{t}$ in $G'$ to solutions to the \textsf{$k$-EDSP} problem in $G$.

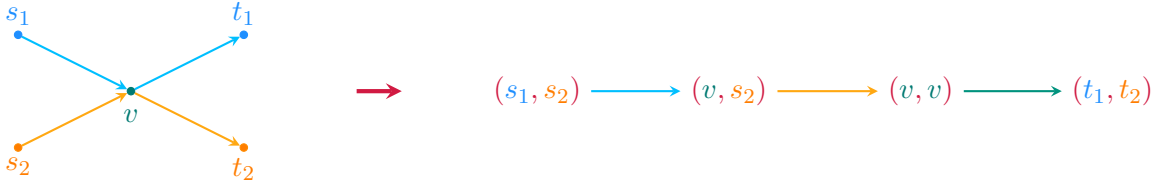
\begin{figure}[t]
    \centering
    \begin{tikzpicture}[scale=1.5,
bvtx/.style={draw=dblue, fill=dblue, circle, inner sep=1pt, minimum width=1pt},
gvtx/.style={inner sep=1pt, circle, draw=dorange, fill=amber, minimum width=1pt},
cvtx/.style={inner sep=1pt, circle, draw=skobeloff, fill=skobeloff, minimum width=1pt},
svtx/.style={inner sep=1pt, circle, draw=skobeloff, fill=skobeloff, minimum width=1pt},
bpath/.style={dsblue, thick,-stealth},
gpath/.style={dtang,thick,-stealth},
mpath/.style={skobeloff!70!cgreen,thick,-stealth}]

%% Vertical Displacement
\def\yspace{1cm};

%% Horizontal Displacement
\def\xstep{2cm};
\def\xsml{1.7cm};

\def\extrw{0.4cm};
\def\extr{1cm};
\def\extrlabel{0.2cm};

\node[bvtx] (s1) at (0,\yspace) {};
\node[bvtx] (t1) at (\xstep, \yspace) {};
\node[cvtx] (a) at (0.5*\xstep, 0.5*\yspace) {};
\node[gvtx] (s2) at (0,0) {};
\node[gvtx] (t2) at (\xstep, 0) {};

\draw[bpath] (s1) -- (a);
\draw[bpath] (a) -- (t1);
\draw[gpath] (s2) -- (a);
\draw[gpath] (a) -- (t2);

\node[above, dblue] at (s1.center) {$s_1$};
\node[below, amber] at (s2.center) {$s_2$};

\node[above, dblue] at (t1.center) {$t_1$};
\node[below, amber] at (t2.center) {$t_2$};

\node[below=3pt, skobeloff] at (a.center) {$v$};

\begin{scope}[xshift = \xstep + \extr]
        \draw[ultra thick, alizarin,-stealth] (0,0.5*\yspace) -- (\extrw, 0.5*\yspace);
\end{scope}

\begin{scope}[xshift = \xstep + \extr + \extrw + \extr+\extrlabel]

    \node (s1s2) at (0,0.5*\yspace) {$\textcolor{alizarin}{(\textcolor{dblue}{s_1},\textcolor{amber}{s_2})}$};

    \node (vs2) at (\xsml,0.5*\yspace) {$\textcolor{alizarin}{(\textcolor{skobeloff}{v},\textcolor{amber}{s_2})}$};

    \node (vv) at (2*\xsml,0.5*\yspace) {$\textcolor{alizarin}{(\textcolor{skobeloff}{v},\textcolor{skobeloff}{v})}$};

    \node (s2t2) at (3*\xsml,0.5*\yspace) {$\textcolor{alizarin}{(\textcolor{dblue}{t_1},\textcolor{amber}{t_2})}$};

    \draw[bpath] (s1s2) -- (vs2);
    \draw[gpath] (vs2) -- (vv);
    \draw[mpath] (vv) -- (s2t2);
\end{scope}
\end{tikzpicture}
    \caption{For $k=2$, the paths $\pair{s_1, v, t_1}$ and $\pair{s_2, v, t_2}$ in $G$ map to a single path in $G'$, pictured on the right, with respect to the topological order $s_1 \prec s_2 \prec v \prec t_1 \prec t_2$.
    Since both coordinates of the node $(v,v)$ are equal, we can step from $(v,v)$ to $(t_1, t_2)$ with a single edge in $G'$.} 
    \label{fig:k-EDSP}
\end{figure}

    \begin{lemma}
        \label{metagraph-EDP-to-path}
        If $G$ contains edge-disjoint $(s_i,t_i)$-shortest paths, then $G'$ contains a path from $\vec{s}$ to $\vec{t}$.
    \end{lemma}
    \begin{proof}
        For $i\in [k]$, let $P_i$  be $(s_i,t_i)$-shortest paths in $G$ which are all edge-disjoint.
        We show how to simultaneously traverse these paths $P_i$ to recover a path in $G'$ from $\vec{s}$ to $\vec{t}$.

                Initialize $P \leftarrow \langle \vec{s}\,\rangle$ and  $\vec{v} \leftarrow \vec{s}$.

        Suppose $\vec{v} = (v_1, \dots, v_k)$.
        Define the new node $\vec{w} = (w_1, \dots, w_k)$ in $G'$ by setting $w_i$ to be the node after $v_i$ on $P_i$ for each $i\in I(\vec{v})$, and setting $w_j = v_j$ for all $j\not\in I(\vec{v})$.
        Then append $\vec{w}$ to $P$, and update $\vec{v} \leftarrow \vec{w}$.

        Repeat the process in the above paragraph until $\vec{v} = \vec{t}$.
        We claim this procedure halts and produces a path $P$ from $\vec{s}$ to $\vec{t}$ in $G'$.

        Indeed, an easy induction argument shows that the value of $\vec{v} = (v_1, \dots, v_k)$ always has the property that $v_i$ is a vertex in $P_i$, and each step strictly decreases the value of 
            \[\sum_{i=1}^k \dist(v_i,t_i).\]
        The above quantity is a nonnegative integer, and so stops decreasing once it reaches zero, at which point we must have $\vec{v} = \vec{t}$.
        Moreover, our procedure is designed so that each step from a node $\vec{v}$ to a node $\vec{w}$ on $P$ satisfies conditions 1 and 3 for edges in $G'$.
        Condition 2 for edges in $G'$ is also satisfied, because the $P_i$ are edge-disjoint, so $P$ is a valid path in $G'$ as desired.
    \end{proof}

    \begin{lemma}
        \label{metagraph-path-to-EDP}
        If $G'$ contains a path form $\vec{s}$ to $\vec{t}$, then $G$ contains edge-disjoint $(s_i,t_i)$-shortest paths.
    \end{lemma}
    \begin{proof}
        Let $P$ be a path from $\vec{s}$ to $\vec{t}$ in $G'$.
        We show how to recover edge-disjoint paths in $G$ by reading off the edges traversed in $P$.

        Initialize $\vec{v}\leftarrow \vec{s}$. 
        For each $i\in [k]$, initialize $P_i \leftarrow \langle s_i\rangle$.

        Suppose $\vec{v} = (v_1, \dots, v_k)$.
        Let $\vec{w} = (w_1, \dots, w_k)$ be the node after $\vec{v}$ on $P$.
        For each index $i$ with $w_i\neq v_i$, append $w_i$ to $P_i$.
        Then update $\vec{v}\leftarrow \vec{w}$.

        Repeat the process in the above paragraph until $\vec{v} = \vec{t}$.
        We claim this procedure produces edge-disjoint shortest paths $P_i$ from $s_i$ to $t_i$.

        Indeed, an easy induction argument shows that at any time in the procedure, the $i^{\text{th}}$ coordinate of $\vec{v}$ is the last vertex in path $P_i$.
        Since $\vec{v}$ begins at $\vec{s}$ and ends at $\vec{t}$, and $G$ is a DAG, each $P_i$ is an $(s_i,t_i)$-path.
        By condition 1 of edges in $G'$, any edge $(\vec{v},\vec{w})$ in $G'$ which changes the value of the $i^{\text{th}}$ coordinate must change that coordinate by stepping along an edge lying in a shortest path from $s_i$.
        It follows that each $P_i$ is an $(s_i,t_i)$-shortest path.

        It remains to prove that the $P_i$ are edge-disjoint.
        Suppose to the contrary that some of these paths overlap at some edge. 
        Without loss of generality, suppose paths $P_1$ and $P_2$ overlap at some edge $(a,b)$.

        Let $\vec{v} = (v_1, \dots, v_k)$ be the last node in $P$ which has $v_1 = a$ as its first coordinate. 
        Such a node exists since $P_1$ passes through $a$.
        By condition 3 for edges in $G'$, we must have $a = \early(\vec{v})$.
        We claim that $v_2 = a$ as well.

        Suppose $v_2\neq a$.
        Then since $a = \early(\vec{v})$, we must have $v_2 \succ a$.
        
        In this case, define $\vec{u}$ to be the last node in $P$ which has $u_2 = a$ as its second coordinate.
        Such a node exists since $P_2$ passes through $a$.
        By condition 3 for edges in $G'$, we have $a = \early(\vec{u})$.
        
        By the assumption that $v_2\succ a$ and condition 1 for edges of $G'$, node $\vec{u}$ occurs before $\vec{v}$ in $P$.
        Then since $v_1 = a$, by condition 1 for edges of $G'$, the first coordinate $u_1$ of $\vec{u}$ must satisfy $u_1\preceq a$.
        Since $a = \early(\vec{u})$, this forces $u_1 = a = \early(\vec{u})$.

        Since $u_1 = \early(\vec{u})$,
        by conditions 1 and 2 for edges in $G'$, all nodes in $P$ after $\vec{u}$ cannot have first coordinate equal to $a$.
        This contradicts the fact that $\vec{u}$ occurs before $\vec{v}$ in $P$.
        So our initial assumption was false, and in fact $v_2 = a$ as claimed.

        Let $\vec{w} = (w_1, \dots, w_k)$ be the node after $\vec{v}$ on $P$.
        Since $v_1 = v_2 = \early(\vec{v})$, by condition 2 for edges of $G'$ we have $w_1\neq w_2$.
        Our procedure for constructing the paths $P_i$ from $P$ has $P_1$ traverse edge $(a,w_1)$ and $P_2$ traverse edge $(a,w_2)$.
        Since $w_1$ and $w_2$ distinct, this contradicts the assumption that paths $P_1$ and $P_2$ both traverse edge $(a,b)$.

        It follows that our original assumption was false, and the paths $P_i$ are edge-disjoint as claimed.
        This completes the proof.   
    \end{proof}

% \subsection{Algorithm}
% \label{subsec:k-EDSP-alg}

\edsp*
\begin{proof}

    First, we compute a topological order of $G$ in linear time.
    Then we compute the $s_i$-shortest paths DAGs $G_i$ of $G$ for all $i\in [k]$, which takes linear time by \Cref{prop:sp-DAG}.

    By \Cref{metagraph-EDP-to-path,metagraph-path-to-EDP}, we can solve \textsf{$k$-EDSP} on $G$ by constructing the graph $G'$ described in 
    % \Cref{subsec:k-EDSP-graph-construction}, 
    this section,
    and checking whether this graph contains a path from $\vec{s}$ to $\vec{t}$.
    Checking whether there is a path from $\vec{s}$ to $\vec{t}$ takes time linear in the size of $G'$, so to prove the theorem, it suffices to show that we can construct $G'$ in $O(mn^{k-1})$ time.

    We can construct all nodes of $G'$ in $O(n^k)$ time.
    We then go through each node $\vec{v}$ in $G$, and add in edges from $\vec{v}$ to $\vec{w}$ in the out-neighborhood of $\vec{v}$ according to the three conditions.
    Each addition of an edge takes $O(1)$ time, so it suffices to show that $G'$ has $O(mn^{k-1})$ edges.

    Consider a node $\vec{v}$ in $G$.
    Let $v = \early(\vec{v})$ be the earliest coordinate of $\vec{v}$, and let $\ell = |I(\vec{v})|$ denote the number of coordinates in $\vec{v}$ equal to this earliest vertex.
    There are at most $n^{k-\ell}$ choices for the coordinates of $\vec{v}$ not equal to $v$.
    By conditions 1 and 3 for edges in $G'$, the node $\vec{v}$  can have edges to at most $\grp{\outdeg(v)}^\ell$ nodes in $G'$.

    Summing over all possible values for $\ell = |I(\vec{v})|$ and $v=\early(\vec{v})$, we see that the number of edges in $G'$ is bounded above by 
        \[\sum_{\ell=1}^{k}\sum_{v\in V}  n^{k-\ell}\grp{\outdeg(v)}^\ell = \sum_{\ell=1}^{k} \sum_{v\in V} \grp{n^{k-\ell}\grp{\outdeg(v)}^{\ell-1}\cdot \outdeg(v)}.\]
    Substituting the inequality $\outdeg(v)\le n$ in the right hand side above, we see that the number of edges in $G'$ is at most 
    \[\sum_{\ell=1}^{k} \sum_{v\in V} \grp{n^{k-\ell}\cdot n^{\ell-1}\cdot \outdeg(v)} = kn^{k-1} \grp{\sum_{v\in V} \outdeg(v)} = kmn^{k-1}.\]
    For constant $k$, the above expression is $O(mn^{k-1})$, which proves the desired result.
\end{proof}

\section{Lower Bounds}\label{sec:lb}

\subsection{Disjoint Shortest Paths}

Our goal in this section is to prove the following theorem. 

\lbdspboth*

% Our reduction is based on \cite{Geometric-Lens} with the additional changes outlined in the technical overview.

\paragraph{Construction}

Let $G = V_1\sqcup \dots\sqcup V_k$ be the input instance of \textsf{$k$-Clique}.

Order each of the vertices in each $V_i$.
Order the set of all vertices $V$ in the graph, by putting all vertices of $V_i$ before vertices of $V_j$ for $i < j$ (and within a part $V_i$ using the order of $V_i$).
For each index $i$, let $\bar{V}_i = V\setminus V_i$.
The order on $V$ induces an order on each $\bar{V}_i$ as well.

We now produce an instance $G'$ of \textsf{$k$-DSP}, depicted in \Cref{fig:dsp-reduction}.

\begin{figure}[t]
    \centering
    \begin{tikzpicture}[scale=1.2,
term/.style={draw=midnight, fill=midnight, circle, inner sep=1pt, minimum width=1pt},
disj/.style={draw=dblue, fill=dblue, circle, inner sep=1pt, minimum width=1pt},
hit/.style={inner sep=1pt, circle, draw=dorange, fill=amber, minimum width=1pt},
dedge/.style={dsblue, thick,-stealth},
hedge/.style={dtang,thick,-stealth},
tedge/.style={cdblue!50!white,thick,-stealth},
bedge/.style={very thick, cdblue!70!midnight},
lbedge/.style={thick, cdblue!70!white}]

    %% Horizontal Spacing
    \def\xstep{1cm}; % main step size
    \def\xsml{0.35cm}; % tiny offset for disjoint edges

    %% Vertical Spacing
    \def\ystep{1cm}; % main step size
    \def\ysml{-0.35cm}; % tiny offset for disjoint edges

    %% Horizontal Transition
    \def\extr{1.2cm};
    \def\extrw{0.6cm};

    %% Circle Positioning
    \def\crad{2cm};
    \def\ycent{2.5cm};

    %% Label Offset
    \def\ldist{0.25cm};

    %%% The 3-Clique instance
    % V_1 set
    \node[term,yshift=\ycent] (A1) at (100:\crad) {};
    \node[term,yshift=\ycent] (A2) at (80:\crad) {};

    \node[midnight,above] at (A1) {$a_1$};
    \node[midnight,above] at (A2) {$a_2$};
    % V_2 set
    \node[term,yshift=\ycent] (B1) at (200:\crad) {};    \node[term,yshift=\ycent] (B2) at (220:\crad) {};

    \node[midnight,yshift=\ycent] at (200:1.12*\crad) {$b_1$};
    \node[midnight,yshift=\ycent] at (220:1.12*\crad) {$b_2$};
    % V_3 set
    \node[term,yshift=\ycent] (C1) at (340:\crad) {};    \node[term,yshift=\ycent] (C2) at (320:\crad) {};    

    \node[midnight,yshift=\ycent] at (340:1.12*\crad) {$c_1$};
    \node[midnight,yshift=\ycent] at (320:1.12*\crad) {$c_2$};

    %% Clique edges
    \draw[bedge] (A1) -- (B1);
        \draw[bedge] (A1) -- (C1);
        \draw[bedge] (B1) -- (C1);
    \draw[lbedge] (A2) -- (B1);
    \draw[lbedge] (A2) -- (C2);
    \draw[lbedge] (B2) -- (A1);
    \draw[lbedge] (B2) -- (C2);

    %% Highlight triangle
    \begin{scope}[on background layer]
        \draw[Yellow!40!white, line width = 10pt, rounded corners] (A1.center) -- (B1.center) -- (C1.center) -- cycle;
    \end{scope}

    %% Transition Arrow
    \draw[xshift=\crad+\extr, line width = 1.8pt, alizarin, yshift=2.5*\ystep,-stealth] (0,0) -- (\extrw,0);

\begin{scope}[xshift=2*\extr+\extrw+\crad]
    %%% The 3-DSP instance
    
    % s_A 
    \node[term] (sA) at (1.5*\xstep, 5*\ystep) {};
    % a_1(b_1)
    \node (a1b1pre) at (\xstep, 4*\ystep) {};
    \node[disj,yshift=\ysml] (a1b1) at (a1b1pre) {};
    % a_1(b_2)
    \node (a1b2pre) at (\xstep, 3*\ystep) {};
    \node[disj,yshift=\ysml] (a1b2) at (a1b2pre) {};
    % a_1(c_1)
    \node (a1c1pre) at (\xstep, 2*\ystep) {};
    \node[disj,yshift=\ysml] (a1c1) at (a1c1pre) {};
    % a_1(c_2) = c_2(a_1)
    \node[hit] (a1c2) at (\xstep, \ystep) {};
    
    % a_2(b_1)
    \node (a2b1pre) at (2*\xstep, 4*\ystep) {};
    \node[disj,yshift=\ysml] (a2b1) at (a2b1pre) {};
    % a_2(b_2) = b_2(a_2)
    \node[hit] (a2b2) at (2*\xstep,3*\ystep) {};
    % a_2(c_1) = c_1(a_2)
    \node[hit] (a2c1) at (2*\xstep,2*\ystep) {};
    % a_2(c_2)
    \node (a2c2pre) at (2*\xstep, \ystep) {};
    \node[disj,yshift=\ysml] (a2c2) at (a2c2pre) {};

    % t_A
    \node[term] (tA) at (1.5*\xstep, 0) {};

    % s_B
    \node[term] (sB) at (0, 3.5*\ystep) {};
    % b_1(a_1)
    \node[disj,xshift=\xsml] (b1a1) at (a1b1pre) {};
    % b_1(a_2)
    \node[disj,xshift=\xsml] (b1a2) at (a2b1pre) {};
    % b_1(c_1)
    \node (b1c1pre) at (3*\xstep, 2*\ystep) {};
    \node[disj,yshift=\ysml] (b1c1) at (b1c1pre) {};
    % b_1(c_2) = c_2(b_1)
    \node[hit] (b1c2) at (3*\xstep, \ystep) {};
    % b_2(a_1)
    \node[disj,xshift=\xsml] (b2a1) at (a1b2pre) {};
    % b_2(a_2) = a_2(b_2)
    \node[hit] (b2a2) at (a2b2) {};
    % b_2(c_1) = c_1(b_2)
    \node[hit] (b2c1) at (4*\xstep, 2*\ystep) {};
    % b_2(c_2)
    \node (b2c2pre) at (4*\xstep, \ystep) {};
    \node[disj,yshift=\ysml] (b2c2) at (b2c2pre) {};

    % t_B
    \node[term] (tB) at (3.5*\xstep, 0) {};

    % s_C
    \node[term] (sC) at (0, 1.5*\ystep) {};
    % c_1(a_1)
    \node[disj,xshift=\xsml] (c1a1) at (a1c1pre) {};
    % c_1(a_2) = a_2(c_1)
    \node[hit] (c1a2) at (a2c1) {};
    % c_1(b_1)
    \node[disj,xshift=\xsml] (c1b1) at (b1c1pre) {};
    % c_1(b_2) = b_2(c_1)
    \node[hit] (c1b2) at (b2c1) {};
    % c_2(a_1) = a_1(c_2)
    \node[hit] (c2a1) at (a1c2) {};
    % c_2(a_2)
    \node[disj,xshift=\xsml] (c2a2) at (a2c2pre) {};
    % c_2(b_1)
    \node[hit] (c2b1) at (b1c2) {};
    % c_2(b_2)
    \node[disj,xshift=\xsml] (c2b2) at (b2c2pre) {};

    % t_C
    \node[term] (tC) at (5*\xstep, 1.5*\ystep) {};

    % Edges exiting s_A
    \node (sAa1pre) at (\xstep, 4.5*\ystep) {};
    \draw[dedge] (sA) -- (sAa1pre.center) -- (a1b1);

    \node (sAa2pre) at (2*\xstep,4.5*\ystep) {};
    \draw[dedge] (sA) -- (sAa2pre.center) -- (a2b1);

    % a_1 path
    \draw[dedge] (a1b1) -- (a1b2);
    \draw[dedge] (a1b2) -- (a1c1);
    \draw[hedge] (a1c1) -- (a1c2);

    % a_2 path
    \draw[hedge] (a2b1) -- (a2b2);
    \draw[hedge] (a2b2) -- (a2c1);
    \draw[dedge] (a2c1) -- (a2c2);

    % Edges entering t_A
    \node (a1tApre) at (\xstep, 0.5*\ystep) {};
    \node (a2tApre) at (2*\xstep, 0.5*\ystep) {};

    \draw[tedge] (a1c2) -- (a1tApre.center) -- (tA);
    \draw[tedge] (a2c2) -- (a2tApre.center) -- (tA);

    % Edges exiting s_B
    \node (sBb1pre) at (0.5*\xstep, 4*\ystep) {};
    \draw[dedge] (sB) -- (sBb1pre.center) -- (b1a1);

    \node (sBb2pre) at (0.5*\xstep,3*\ystep) {};
    \draw[dedge] (sB) -- (sBb2pre.center) -- (b2a1);

    % b_1 path
    \draw[dedge] (b1a1) -- (b1a2);
    \node (b1c1anchor) at (3*\xstep, 3*\ystep - \ysml) {};
    \draw[dedge] (b1a2) -- (b1c1anchor.center) -- (b1c1);
    \draw[hedge] (b1c1) -- (b1c2);

    % b_2 path
    \draw[hedge] (b2a1) -- (b2a2);
    \node (b2c1anchor1) [xshift=\xstep+\xsml] at  (b2a2) {};
    \node (b2c1anchor2) [yshift=-\ysml] at  (b2c1) {};
    \draw[hedge] (b2a2) -- (b2c1anchor1.center) -- (b2c1anchor2.center) -- (b2c1);
    \draw[dedge] (b2c1) -- (b2c2);

    % Edges entering t_B
    \node (b1tBpre) at (3*\xstep, 0.5*\ystep) {};
    \node (b2tBpre) at (4*\xstep, 0.5*\ystep) {};

    \draw[tedge] (b1c2) -- (b1tBpre.center) -- (tB);
    \draw[tedge] (b2c2) -- (b2tBpre.center) -- (tB);

    % Edges exiting s_C
    \node (sCc1pre) at (0.5*\xstep, 2*\ystep) {};
    \draw[dedge] (sC) -- (sCc1pre.center) -- (c1a1);

    \node (sCc2pre) at (0.5*\xstep, \ystep) {};
    \draw[hedge] (sC) -- (sCc2pre.center) -- (c2a1);

    % c_1 path
    \draw[hedge] (c1a1) -- (c1a2);
    \draw[dedge] (c1a2) -- (c1b1);
    \draw[hedge] (c1b1) -- (c1b2);
    % c_2 path
    \draw[dedge] (c2a1) -- (c2a2);
    \draw[hedge] (c2a2) -- (c2b1);
    \draw[dedge] (c2b1) -- (c2b2);

    % Edges entering t_C
    \node (c1tCpre) at (4.5*\xstep, 2*\ystep) {};
    \node (c2tCpre) at (4.5*\xstep, \ystep) {};

    \draw[tedge] (c1b2) -- (c1tCpre.center) -- (tC);
    \draw[tedge] (c2b2) -- (c2tCpre.center) -- (tC);

    %% Terminal node labels
    \node[above,midnight] at (sA) {$s_1$};
        \node[left,midnight] at (sB) {$s_2$};
        \node[left,midnight] at (sC) {$s_3$};
        \node[below,midnight] at (tA) {$t_1$};
        \node[below,midnight] at (tB) {$t_2$};
        \node[right,midnight] at (tC) {$t_3$};

    \node[cdblue, xshift=-\ldist, yshift=\ldist] at (sAa1pre) {$a_1$};
    \node[cdblue, xshift=\ldist, yshift=\ldist] at (sAa2pre) {$a_2$};

    \node[cdblue, xshift=-\ldist,yshift=\ldist] at (sBb1pre) {$b_1$};
    \node[cdblue, xshift=-\ldist,yshift=-\ldist] at (sBb2pre) {$b_2$};

    \node[cdblue, xshift=-\ldist, yshift=\ldist] at (sCc1pre) {$c_1$};
    \node[cdblue, xshift=-\ldist, yshift=-\ldist] at (sCc2pre) {$c_2$};

    % Highlight a1 path
    \begin{scope}[on background layer]
        \draw[Yellow!40!white, line width = 8pt, rounded corners] (sA.center) -- (sAa1pre.center) -- (a1tApre.center) -- (tA.center);
    \end{scope}
    % Highlight b1 path
    \begin{scope}[on background layer]
        \draw[Yellow!40!white, line width = 8pt, rounded corners] (sB.center) -- (sBb1pre.center) -- (b1a2.center) -- (b1c1anchor.center) -- (b1tBpre.center) -- (tB.center);
    \end{scope}
    % Highlight c1 path
    \begin{scope}[on background layer]
        \draw[Yellow!40!white, line width = 8pt, rounded corners] (sC.center) -- (sCc1pre.center) -- (c1tCpre.center) -- (tC.center);
    \end{scope}
\end{scope}
    
\end{tikzpicture}
    \caption{An example of the reduction from \textsf{$k$-Clique} to \textsf{$k$-DSP} for $k=3$ and $n=2$.
    The \textsf{$k$-Clique} instance $G$ has vertex parts $V_1 = \set{a_1, a_2}$, $V_2 = \set{b_1, b_2}$, and $V_3 = \set{c_1, c_2}$. 
    The unique triangle $(a_1,b_1,c_1)$ in $G$ on the left has  highlighted bold edges. The vertices in this triangle are mapped to the disjoint highlighted paths in $G'$ on the right. For vertices $v,w$ in different parts of $G$, if $(v,w)$ is an edge, the paths for $v$ and $w$ in $G'$ pass by each other with blue edges (since $v(w)\neq w(v)$ in $G'$), and if $(v,w)$ is not an edge, the paths for $v$ and $w$ intersect at the orange node $v(w) = w(v)$. }
    \label{fig:dsp-reduction}
\end{figure}
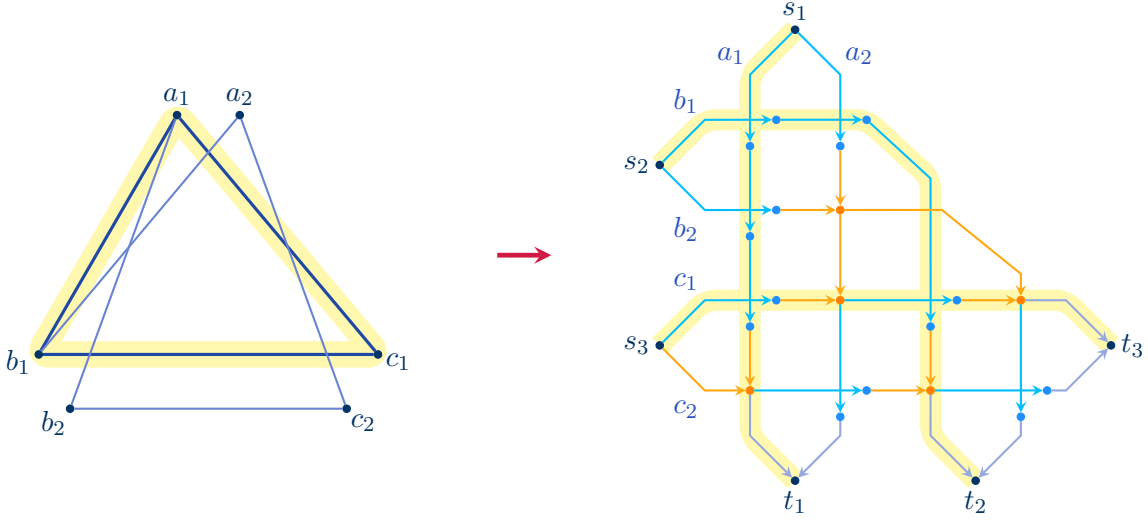

For each $i\in[k]$, graph $G'$ has a source node $s_i$ and a target node $t_i$.

For each choice of vertices $v,w$ from different parts, we introduce node $v(w)$ in $G'$.
% Intuitively, the nodes $v(w)$ and $w(v)$ will form a gadget checking if $(v,w)$ is an edge in $G$.

For every vertex $v\in V_i$, the graph $G'$ includes a path $P(v)$ from $s_i$ to $t_i$, whose internal nodes are all the nodes  of the form $v(w)$ where $w\in\bar{V}_i$, traversed according to the order of $\bar{V}_i$.

Finally, for each pair $(v,w)$ which is not an edge in $G$, we identify the nodes $v(w)=w(v)$ in $G'$.

This completes the construction of $G'$.

\paragraph{Correctness}

To prove our reduction is correct, we characterize shortest paths in $G'$.

\begin{lemma}
\label{canonical-sp}
    The $(s_i,t_i)$-shortest paths in $G'$ are precisely the $P(v)$ paths for vertices $v\in V_i$.
\end{lemma}
\begin{proof}
Fix $i\in[k]$.

Given a vertex $w\in \bar{V}_i$, let $\pos(w)$ denote the position of $w$ in $\overline{V}$ from the end of this list.
For example, if $w$ is the final vertex in $\overline{V}_i$ we set $\pos(w) = 1$, and if $w$ is the first vertex in $\bar{V}_i$, we set $\pos(w) = |\bar{V}_i| = (k-1)n$.

Given nodes $a$ and $b$ in $G'$, in this proof we let $\dist(a,b)$ denote the distance from $a$ to $b$ in $G'$.

\begin{claim}
\label{claim:distance}
    For every $v\in V_i$ and $w\in \bar{V}_i$,
    we have 
    \begin{equation}
    \label{induction}
    \dist(v(w),t_i) \ge \pos(w)\quad\text{and}\quad \dist(w(v),t_i)\ge \pos(w).
    \end{equation}
\end{claim}
\begin{claimproof}
    We prove the inequalities by induction on the value of $\pos(w)$.

    For the base case, suppose $\pos(w) = 1$.
    Then $v(w)\neq t_i$ and $w(v)\neq t_i$, so the distance from these nodes to $t_i$ is at least $1$ as claimed.

    For the inductive step, suppose now that $\pos(w) \ge 2$, and we have already proven the claim for all vertices occurring after $w$ in $\bar{V}_i$.
    
    % The node $v(w)$ has at most two edges exiting it in $G'$.
    
    % The path $P(v)$ always has an edge exiting $v(w)$.
    % Second, if $v(w) = w(v)$ (because $(v,w)$ is not an edge in $G$), then the path $P(w)$ contains an edge exiting $v(w)$ as well.
    % No other edges exit $v(w)$.
    The only edges exiting $v(w)$ and $w(v)$ come from the paths $P(v)$ and $P(w)$.
    We consider these two cases separately below.

Suppose we take a path from $v(w)$ or $w(v)$ which begins with an edge from $P(v)$.
Let $w'$ be the node after $w$ in $\bar{V}_i$.
Then $v(w')$ is the second vertex in this path.
By the induction hypothesis, 
\[\dist(v(w'),t_i) \ge \pos(w') = \pos(w)-1.\]
By the above equation, any path from $v(w)$ or $w(v)$ to $t_i$ in this case has length at least 
    \[\dist(v(w'),t_i) + 1  \ge \pos(w).\]

Suppose now instead that we take a path from $v(w)$ or $w(v)$ to $t_i$ which begins with an edge from $P(w)$.
No such path exists if $v$ is the final vertex in $V_i$, so we may assume that $v$ is not the final vertex in $V_i$.
Let $v'$ be the vertex after $v$ in $V_i$.
Then $w(v')$ is the second vertex in our path.
By the induction hypothesis, 
\[\dist(w(v'),t_i)\ge \pos(w)\]
which implies that the path in this case has length at least $\pos(w)$ as well.

So in either case, \cref{induction} holds.
This completes the induction, and proves the desired result.
\end{claimproof}

By \Cref{claim:distance}, for any $v\in V_i$ and $w\in\bar{V}_i$ we have 
    \[\dist(v(w),t_i)\ge \pos(w).\]

The $v(w)$ to $t_i$ subpath of $P(v)$ shows that $\dist(v(w),t_i)\le \pos(w)$, so in fact we have 
    \[\dist(v(w),t_i)= \pos(w).\]

Now, every out-neighbor of $s_i$ is of the form $v(w^\star)$, where $v\in V_i$ and $w^\star$ is the first node in $\bar{V}_i$.
Then by the above discussion,
    \begin{equation}
    \label{eq:sp-length}
    \dist(s_i,t_i) = |\bar{V}_i| + 1.
    \end{equation}
This immediately implies that for all $v\in V_i$, each $P(v)$ is an $(s_i,t_i)$-shortest path.

It remains to show that these are the only $(s_i,t_i)$-shortest paths in $G'$.

To that end, let $Q$ be an $(s_i,t_i)$-shortest path in $G'$.
Then the the second vertex in $Q$ is of the form $v(w^\star)$, for $v\in V_i$ and $w^\star$ defined as above.
Let $w$ be the last vertex in $\bar{V}_i$ with the property that the $s_i$ to $v(w)$ subpaths of $Q$ and $P(v)$ agree.

Suppose $w$ is not the final vertex in $\bar{V}_i$.
In case, immediately after $v(w)$, path $Q$ must traverse an edge of $P(w)$.

If $v$ is the final vertex in $V_i$, then this edge brings $Q$ to a node of the form $w(v')$, where $v'\in V_j$ for some $j > i$.
However, there is no path from such a node $w(v')$ to $t_i$.
Since we assumed $Q$ is a path to $t_i$, this case cannot occur.

So $v$ is not the final vertex in $V_i$.
Let $v'$ be the vertex after $v$ in $V_i$.

Then using the edge of $P(w)$, $Q$ goes from node $v(w)$ to $w(v')$.
By \Cref{claim:distance}, we know that 
    \[\dist(w(v'),t_i)\ge \pos(w).\]

So the $w(v')$ to $t_i$ subpath of $Q$ has length at least $\pos(w)$.

However, we know that $P(v)$ and $Q$ agree up to $v(w)$, so the $s_i$ to $v(w)$ subpath of $Q$ has length 
    \[|\bar{V}_i| - \pos(w) + 1.\]
    
The path $Q$ consists of these two subpaths and the edge from $v(w)$ to $w(v')$.

Thus, $Q$ has length at least 
\[\grp{|\bar{V}_i| - \pos(w) + 1} + 1 + \pos(w) = |\bar{V}_i| + 2\]

which is greater than $\dist(s_i,t_i)$ by \cref{eq:sp-length}.

This violates the definition of $Q$ as a shortest path.
So $w$ must be the final vertex in $\bar{V}_i$, which forces $Q = P(v)$.
So all $(s_i,t_i)$-shortest paths in $G'$ are of the desired form.
\end{proof}

We are now ready to prove our lower bound for \textsf{$k$-DSP}.

% \begin{lemma}
%     \label{lbdsp}
%     There is a reduction from \textsf{$k$-Clique} to \textsf{$k$-DSP} on a DAG with $O((kn)^2)$ vertices, that runs in $O((kn)^2)$ time.
% \end{lemma}
\lbdspboth*
\begin{proof}
    Let $G = V_1\sqcup\dots\sqcup V_k$ instance of \textsf{$k$-Clique}.
    Construct the graph $G'$ defined in this section in $O((kn)^2)$ time.
    By definition, $G'$ has $O((kn)^2)$ nodes.
    
    We claim that $G$ contains a $k$-clique if and only if $G'$ contains disjoint $(s_i,t_i)$-shortest paths.

    Indeed, suppose $G$ contains a clique of the form $(v_1, \dots, v_k)\in V_1\times\dots\times V_k$.
    
    Take the $P(v_i)$ paths.
    By \Cref{canonical-sp} these are $(s_i,t_i)$-shortest paths.
    
    The internal nodes of $P(v_i)$ are of the form $v_i(w)$, for $w\in \bar{V}_i$.
    
    Thus the only way paths $P(v_i)$ and $P(v_j)$ could intersect for $i\neq j$ is if $v_i(v_j) = v_j(v_i)$ in $G'$.
    However, vertices $v_i$ and $v_j$ belong to a clique, which means $(v_i,v_j)$ is an edge in $G$, so $v_i(v_j)\neq v_j(v_i)$.

    Hence the $P(v_i)$ are vertex-disjoint paths.

    Conversely, suppose $G'$ contains disjoint $(s_i,t_i)$-shortest paths.
    By \Cref{canonical-sp}, these paths are of the form $P(v_i)$ for some vertices $v_i\in V_i$.
    Since these paths are vertex-disjoint, we must have $v_i(v_j)\neq v_j(v_i)$ for all $i\neq j$.
    This means that $(v_i,v_j)$ are edges in $G$ for all $i\neq j$.
    
    Thus $(v_1,\dots, v_k)$ forms a clique in $G$, which proves the desired result.    
\end{proof}

% Our lower bound for \textsf{$k$-DSP} in undirected graphs is very similar to our proof of \Cref{lbdsp} above, using a slight modification from \cite[Section 4]{Geometric-Lens}, so we only sketch its proof.

% \begin{lemma}
%     \label{lbdspundir}
%     There is a reduction from \textsf{$k$-Clique} on a $kn$-vertex graph to \textsf{$k$-DSP} on an undirected graph with $O((kn)^2)$ vertices, that runs in $O((kn)^2)$ time.
% \end{lemma}
% \begin{proof}[Proof sketch.]
%         Let $G = V_1\sqcup\dots\sqcup V_k$ instance of \textsf{$k$-Clique}.
%     Construct the graph $G'$ defined in this section in $O((kn)^2)$ time.
%     Modify $G'$ to be undirected, by ignoring the orientation of each its edges. 
%     This reduction does not work as written, because $G'$ can contain $(s_i,t_i)$-shortest paths which use edges incident to terminals $s_j$ and $t_j$ for $j\neq i$.
%     To prevent this from happening, modify $G'$ further by subdividing each edge incident to terminal $kn-1$ times (i.e., we replace each edge incident to a source $s_i$ or a target $t_i$ in $G'$ with a new path of length $kn$).

%     With these changes,  \Cref{canonical-sp} once again holds for the undirected graph $G'$.
    
%     Then similar reasoning to the proof of \Cref{lbdsp} proves the desired result.
% \end{proof} 

\subsection{Disjoint Paths}

Our goal in this section is to prove the following theorem.

\lbdp*

To establish these reductions, we first introduce the notion of a \emph{covering family} of increasing lists, and then apply the construction of such a family to the reduction framework of \cite{Slivkins2010}.

% \subsubsection*{The Increasing Lists Puzzle}

\subsubsection*{Covering Pairs with Lists}

% \begin{framed}
%     \noindent {\bf Increasing Lists Covering Family}\\
%     Given a positive integer $k$, we say a collection what is the smallest positive integer $\lambda = \lambda(k)$ such that there exists $\lambda$ lists of increasing integers with the property that for all integers $i,j$ with $1\le i < j\le k$, the pair $(i,j)$ appears as consecutive members of some list? 
% \end{framed}
% The increasing lists puzzle is defined as follows: Given a positive integer $k$, what is the minimum  $\ell$ such that there exist $\ell$ lists of increasing positive integers with the property that for all positive integers $i,j$ with $i < j\leq k$, the pair $(i,j)$ appear as consecutive members in some list?

\begin{definition}[Covering by Increasing Lists]
    Given a positive integer $k$, we say a collection of lists $\cal{L}$ of increasing integers is a \emph{$k$-covering family} if for all integers $i,j$ with $1\le i<j \le k$, the integers $i$ and $j$ appear as consecutive members of some list in $\mathcal{L}$.
    We let $\lambda(k)$ denote the minimum possible number of lists in a $k$-covering family.
\end{definition}

% \noindent Throughout this section, as above, we let $\lambda(k)$ denote the solution to the {\bf Increasing Lists Puzzle} for the input integer $k$.

The following lemma lower bounds the minimum possible size of a $k$-covering family. 

\begin{lemma}
    \label{lem:atleast-list}
        We have $\lambda(k)\ge \lfloor k^2/4\rfloor$.
\end{lemma}
\begin{proof}
Partition the set of the first $k$ positive integers into sets 
    \[[k] = A\sqcup B\]
of consecutive positive integers 
    \[A = [1, \lfloor k/2\rfloor] \quad \text{and}\quad B = [\lfloor k/2\rfloor + 1, k]\]
of size $\lfloor k/2 \rfloor$ and $\lceil k/2 \rceil$ respectively. 
We claim that it cannot be the case that for $a,a'\in A$ and $b,b'\in B$ both $(a,b)$ and $(a',b')$ appear consecutively in the same list $L$. 
This is because $L$ must be increasing, so both $a$ and $a'$ must appear before both $b$ and $b'$ in $L$. 
This means that whichever of $a$ or $a'$ appears first, cannot appear consecutively with $b$ or $b'$.

It follows that the number of lists needed to cover all pairs $(i,j)\in [k]^2$ with $i<j$ is at least 
    \[|A||B| = \lfloor k/2\rfloor\lceil k/2\rceil = \lfloor k^2/4\rfloor\]
    
    as claimed.
\end{proof}

In light of \Cref{lem:atleast-list}, the following lemma shows that we can efficiently construct a minimum size $k$-covering family (and that $\lambda(k) = \lfloor k^2/4\rfloor$).
This construction will be helpful in designing a reduction from \textsf{$k$-Clique}.
The proof of this lemma is due to an anonymous reviewer. 

\begin{lemma}
    \label{puzzle-solution}
    There is an algorithm which, given a positive integer $k$, runs in $O(k^3)$ time and constructs a 
     $k$-covering family $\mathcal{L}$ of $\lfloor k^2/4\rfloor$ lists with elements from $[k]$, with the property that each integer in $[k]$ appears in fewer than $k$ lists of $\mathcal{L}$.
\end{lemma}

\begin{proof} 

    For each pair $(a,d)$ of positive integers with $d\le k-1$ and $a\le \min(d,k-d)$, define $L_{a,d}$ to be the list of positive integers in $[k]$ whose $(j+1)^{\text{st}}$ term is $(a+jd)$ for each $j\ge 0$.
    Let $\mathcal{L}$ be the collection of all the $L_{a,d}$ lists, for the pairs $(a,d)$ satisfying the aforementioned conditions. 

    We claim that any integer $x\in [k]$ shows up in fewer than $k$ lists of $\mathcal{L}$.
    This is because for any $d\in [k-1]$, we have $x\in L_{a,d}$ only if $a\equiv x\pmod d$.
    Since $a\in [d]$, this means that for each $d$, there is at most one $a$ for which $x\in L_{a,d}$.
    Since $d\in [k-1]$ takes on fewer than $k$ values,  $x$ appears in fewer than $k$ lists of $\mathcal{L}$.

    We now show that $\mathcal{L}$ is a $k$-covering family.
    Take arbitrary integers $i,j\in [k]$ with $i < j$.
    
    Set $d = j-i$, and take $a$ to be the smallest positive integer with $a\equiv i\pmod d$.
    
    By construction, $a\le d$.
    Moreover, we have $a+d\le i+d = j\le k$, so $a\le k-d$ as well.
    So the pair $(a,d)$ corresponds to a list $L_{a,d}$.
    Then $i$ and $j$ are both in $L_{a,d}$ because $j\equiv i\equiv a\pmod d$.
    Finally, because $d = j-i$, we know that $i$ and $j$ are consecutive elements in $L_{a,d}$.
    So $\mathcal{L}$ is a  $k$-covering family as claimed.

    The number of lists in $\mathcal{L}$ is 
        \begin{equation}
        \label{eq:num-lists}
            \sum_{d=1}^{k-1} \min(d,k-d) = \sum_{d=1}^{\lfloor k/2\rfloor} d + \sum_{d=\lfloor k/2\rfloor + 1}^{k-1} (k-d).
        \end{equation}

    We now perform casework on the parity of $k$.

    \textbf{Case 1: $\boldsymbol{k}$ is even}

        If $k$ is even, we can write $k=2\ell$ for some positive integer $\ell$.
        Then \cref{eq:num-lists} simplifies to
            \[(1 + 2 + \dots + \ell) + (1 + 2 + \dots + \ell-1) = \frac{\ell(\ell+1)}{2} + \frac{\ell(\ell-1)}{2} = \ell^2 = \lfloor k^2/4\rfloor.\]

        \textbf{Case 2: $\boldsymbol{k}$ is odd}

    If $k$ is odd, we can write $k=2\ell+1$ for some positive integer $\ell$.
    Then \cref{eq:num-lists} simplifies to

        \[(1 + 2 + \dots + \ell) + (1 + 2 + \dots + \ell) = \ell^2 + \ell = \lfloor k^2/4\rfloor.\]

    In either case, we get that the number of lists in $\mathcal{L}$ is $\lfloor k^2/4\rfloor$ as claimed.

    We can easily construct $\mathcal{L}$ in $O(k^3)$ time by going through each of the $O(k^2)$ pairs $(a,d)$ defining a list $L_{a,d}$ and then listing its at most $O(k)$ elements by starting at $a$, and incrementing by $d$ until we reach an integer greater than $k$.
    This completes the proof of the lemma. 
\end{proof}

\subsubsection*{The Reduction}

% We first note that Slivkins's reduction is to the \textsf{EDP} problem, while our reduction is to the \textsf{DP} problem (and therefore our reduction is slightly stronger due to the known simple reduction from \textsf{DP} to \textsf{EDP}). As a result, our construction is slightly different from Slivkins's and we will present it in a self-contained way instead of referencing his work as a black box. Nevertheless, most of the following ideas are taken directly from his work. Our main novel contribution is the extraction of the increasing lists puzzle from the structure of his reduction, and the previously presented solution to the puzzle.

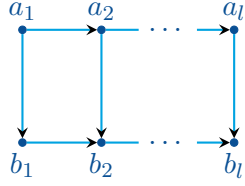
\begin{figure}[t]
    \centering
    %% A single gadget. 
\begin{tikzpicture}[scale=1.5,
gvtx/.style=
{draw=dblue!30!midnight, fill=dblue!30!midnight,
circle, inner sep=1pt, minimum width=1pt},
gedge/.style=
{draw=dsblue!80!midnight, thick, -stealth},
gnoarrow/.style=
{draw=dsblue!80!midnight, thick},
sedge/.style=
{draw=dtang!70!white, thick, -stealth, dashed}]
    %% Horizontal Spacing
    \def\xsml{0.7cm}; % step size within a gadget
    \def\xhalf{0.6*\xsml};
    \def\dotspace{0.4*\xhalf};

    %% Vertical Spacing
    \def\ygad{-1cm}; % vertical length

    \node[gvtx] (a1) at (0,0) {};
    \node[gvtx] (b1) at (0,\ygad) {};
    \node[gvtx] (a2) at (\xsml,0) {};
    \node[gvtx] (b2) at (\xsml,\ygad) {};

    \node (a2post) at (\xsml+\xhalf,0) {};
    \node (b2post) at (\xsml+\xhalf,\ygad) {};

    \node (aendpre) at (\xsml + \xhalf + 2*\dotspace,0) {};
    \node (bendpre) at (\xsml + \xhalf + 2*\dotspace,\ygad) {};

    \node[gvtx] (aend) at (\xsml + 2*\xhalf + 2*\dotspace, 0) {};
    \node[gvtx] (bend) at (\xsml + 2*\xhalf + 2*\dotspace, \ygad) {};

    \draw[gedge] (a1) -- (a2);
    \draw[gedge] (b1) -- (b2);
    
    \draw[gedge] (a1) -- (b1);
    \draw[gedge] (a2) -- (b2);

    \draw[gedge] (aend) -- (bend);

    % Ellipses position
    \node[dblue!30!midnight] at (\xsml + \xhalf + \dotspace, 0) {$\dots$};
    \node[dblue!30!midnight] at (\xsml + \xhalf + \dotspace, \ygad) {$\dots$};

    % Edges to and from ellpises
    \draw[gnoarrow] (a2) -- (a2post);
    \draw[gnoarrow] (b2) -- (b2post);
\draw[gedge] (aendpre) -- (aend);  
\draw[gedge] (bendpre) -- (bend);

% Various node labels
\node[dblue!30!midnight, above] at (a1) {$a_1$};
\node[dblue!30!midnight, below] at (b1) {$b_1$};
\node[dblue!30!midnight, above] at (a2) {$a_2$};
\node[dblue!30!midnight, below] at (b2) {$b_2$};   
\node[dblue!30!midnight, above] at (aend) {$a_l$};
\node[dblue!30!midnight, below] at (bend) {$b_l$};  
\end{tikzpicture}
    \caption{
    Given a vertex $v\in V_i$ from $G$, the gadget graph $G_v$ has top nodes $a_j = a_j(v)$ and bottom nodes $b_j = b_j(v)$ for each $j\in [l]$, where $l=l(i)$ is the number of lists in $\ca{L}$ containing  $i$.}
    \label{fig:k-DP-gadget}
\end{figure}

    Let $G = V_1\sqcup \dots \sqcup V_k$ be an instance of \textsf{$k$-Clique} on $kn$ vertices.
    
    We assume $n\ge 2$, since otherwise the problem is trivial.

    Set $\lambda = \lfloor k^2/4\rfloor$.
    By \Cref{puzzle-solution}, there exist a collection $\ca{L}$ of $\lambda$ increasing lists of integers from $[k]$, such that for every choice of $i,j\in [k]$ with $i < j$, the integers $i$ and $j$ appear as consecutive elements in some list of $\ca{L}$.
    Moreover, we can construct $\ca{L}$ in $O(1)$ time for constant $k$.
    We give this collection of lists some arbitrary order.
    
    Given $i\in [k]$, let $l(i)$ denote the number of lists in $\ca{L}$ containing $i$.
    % For each $r \in [s(i)]$, we let $L_r(i)$ denote the $r^{\text{th}}$ list containing $i$.

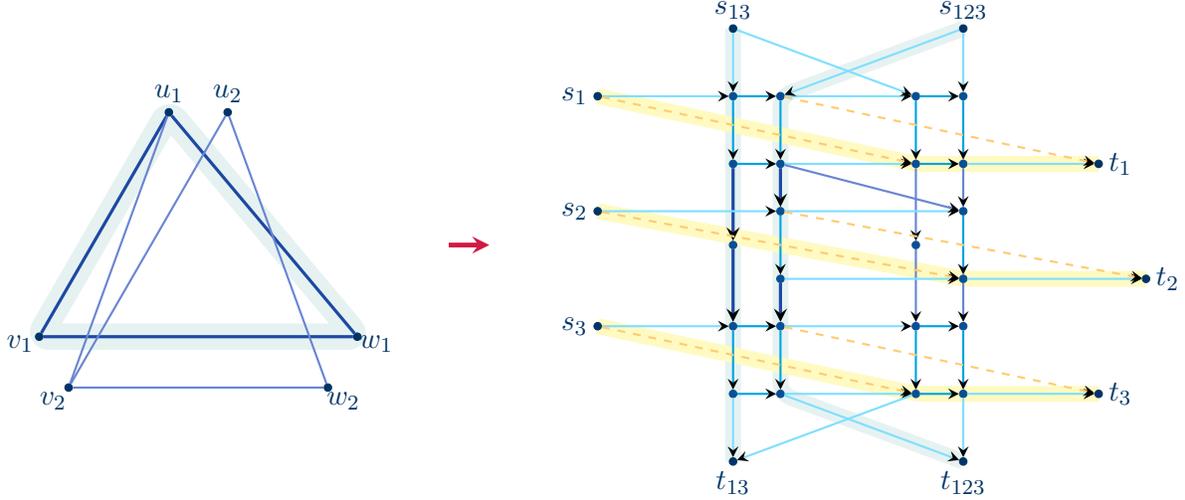
\begin{figure}[t]
    \centering
    \begin{tikzpicture}[scale=0.9,
term/.style=
{draw=midnight, fill=midnight, 
circle, inner sep=1pt, minimum width=1pt},
gvtx/.style=
{draw=dblue!30!midnight, fill=dblue!30!midnight,
circle, inner sep=1pt, minimum width=1pt},
redge/.style=
{draw=dsblue!50!white, thick, -stealth},
tedge/.style=
{draw=dsblue!50!white, thick, -stealth},
gedge/.style=
{draw=dsblue!80!midnight, thick, -stealth},
sedge/.style=
{draw=dtang!60!white, thick, -stealth, dashed},
oedge/.style=
{draw=cdblue!70!white, thick, -stealth},
cedge/.style=
{draw=cdblue!70!midnight, very thick, -stealth},
bedge/.style=
{very thick, cdblue!70!midnight},
lbedge/.style={thick, cdblue!70!white}]
    %% Horizontal Spacing
    \def\xbig{2cm}; % step size between gadgets in a row
    \def\xsml{0.7cm}; % step size within a gadget

    %% Vertical Spacing
    \def\ygad{-1cm}; % vertical length of a gadget
    \def\ybetrow{-0.7cm}; % vertical spacing between different rows (also  the y-coordinate for the top of row 1)
    \def\yterm{1cm};
    % vertical spacing between terminals and rows (remember to take into account the fact that first row has top y-coordinate at \ybetrow)

    %% Horizontal Transition
    \def\extr{-1.4cm};
    \def\extrw{-0.6cm};

    \def\extrlabel{-0.7cm};

    %% Circle Positioning
    \def\crad{2.5cm};
    \def\ycshift{0.5*\yterm};

    \begin{scope}[yshift = 0.5*((\yterm + \ybetrow) + (\ybetrow + 3*\ygad + 2*\ybetrow - \yterm)) - \ycshift, xshift=2*\extr + \extrw - \crad]
    %%% The 3-Clique instance
    % V_1 set
    \node[term] (A1) at (100:\crad) {};
    \node[term] (A2) at (80:\crad) {};

    \node[midnight,above] at (A1) {$u_1$};
    \node[midnight,above] at (A2) {$u_2$};
    % V_2 set
    \node[term] (B1) at (200:\crad) {};    \node[term] (B2) at (220:\crad) {};

    \node[midnight] at (200:1.12*\crad) {$v_1$};
    \node[midnight] at (220:1.12*\crad) {$v_2$};
    % V_3 set
    \node[term] (C1) at (340:\crad) {};    \node[term] (C2) at (320:\crad) {};    

    \node[midnight] at (340:1.12*\crad) {$w_1$};
    \node[midnight] at (320:1.12*\crad) {$w_2$};

    %% Clique edges
    \draw[bedge] (A1) -- (B1);
        \draw[bedge] (A1) -- (C1);
        \draw[bedge] (B1) -- (C1);
    \draw[lbedge] (A2) -- (B2);
    \draw[lbedge] (A2) -- (C2);
    \draw[lbedge] (B2) -- (A1);
    \draw[lbedge] (B2) -- (C2);
    \end{scope}

    \begin{scope}[xshift=2*\extr + \extrw + \extrlabel, yshift = 0.5*((\yterm + \ybetrow) + (\ybetrow + 3*\ygad + 2*\ybetrow - \yterm))]
    %% Transition Arrow
    \draw[line width = 1.8pt, alizarin,stealth-] (\crad,0) -- (\crad + \extrw,0);
    \end{scope}

    %% Row 1
    % s1
    \node[term] (s1) at (0,\ybetrow) {};
    \node[left,midnight] at (s1) {$s_1$};
    % G(a1)
    \node[gvtx] (a11top) at (\xbig, \ybetrow) {};
    \node[gvtx] (a12top) at (\xbig+\xsml, \ybetrow) {};
    \node[gvtx] (a11bot) at (\xbig, \ybetrow+\ygad) {};
    \node[gvtx] (a12bot) at (\xbig+\xsml, \ybetrow+\ygad) {};
    % G(a2)
    \node[gvtx] (a21top) at (2*\xbig+\xsml, \ybetrow) {};
    \node[gvtx] (a22top) at (2*\xbig+2*\xsml, \ybetrow) {};
    \node[gvtx] (a21bot) at (2*\xbig+\xsml, \ybetrow+\ygad) {};
    \node[gvtx] (a22bot) at (2*\xbig+2*\xsml, \ybetrow+\ygad) {};
    % t1
    \node[term] (t1) at (3*\xbig+2*\xsml,\ybetrow+\ygad) {};
    \node[right,midnight] at (t1) {$t_1$};
    % Basic row edges in row 1
    \draw[redge,-stealth] (s1) -- (a11top);
     \draw[redge,-stealth]   (a12top) -- (a21top);
     \draw[redge,-stealth]   (a12bot) -- (a21bot);
      \draw[redge,-stealth]  (a22bot) -- (t1);
    % Gadget edges in row 1
    \draw[gedge] (a11top) -- (a12top);
    \draw[gedge] (a11bot) -- (a12bot);
    \draw[gedge] (a11top) -- (a11bot);
    \draw[gedge] (a12top) -- (a12bot);

    \draw[gedge] (a21top) -- (a22top);
    \draw[gedge] (a21bot) -- (a22bot);
    \draw[gedge] (a21top) -- (a21bot);
    \draw[gedge] (a22top) -- (a22bot);

    % Skip edges in row 1
    \draw[sedge] (s1) -- (a21bot);
    \draw[sedge] (a12top) -- (t1);

    %% Row 3
    \begin{scope}[yshift={2*\ybetrow + 2*\ygad}]
    % s3
    \node[term] (s3) at (0,\ybetrow) {};
    \node[left,midnight] at (s3) {$s_3$};
    % G(c1)
    \node[gvtx] (c11top) at (\xbig, \ybetrow) {};

    \node[gvtx] (a11botc11topmid) at (\xbig,-0.5*\ygad) {};
    % \node[gvtx] (a11c11mid) at (\xbig, -0.5*\ygad)) {};
    
    \node[gvtx] (c12top) at (\xbig+\xsml, \ybetrow) {};
    \node[gvtx] (c11bot) at (\xbig, \ybetrow+\ygad) {};
    \node[gvtx] (c12bot) at (\xbig+\xsml, \ybetrow+\ygad) {};
    % G(c2)
    \node[gvtx] (c21top) at (2*\xbig+\xsml, \ybetrow) {};
    \node[gvtx] (c22top) at (2*\xbig+2*\xsml, \ybetrow) {};
    \node[gvtx] (c21bot) at (2*\xbig+\xsml, \ybetrow+\ygad) {};
    \node[gvtx] (c22bot) at (2*\xbig+2*\xsml, \ybetrow+\ygad) {};

    \node[gvtx] (a21botc21topmid) at (2*\xbig+\xsml,-0.5*\ygad) {};

    % t3
    \node[term] (t3) at (3*\xbig+2*\xsml,\ybetrow+\ygad) {};
    \node[right,midnight] at (t3) {$t_3$};
    % Basic row edges in row 3
    \draw[redge,-stealth] (s3) -- (c11top);
     \draw[redge,-stealth]   (c12top) -- (c21top);
     \draw[redge,-stealth]   (c12bot) -- (c21bot);
      \draw[redge,-stealth]  (c22bot) -- (t3);
    % Gadget edges in row 3
    \draw[gedge] (c11top) -- (c12top);
    \draw[gedge] (c11bot) -- (c12bot);
    \draw[gedge] (c11top) -- (c11bot);
    \draw[gedge] (c12top) -- (c12bot);

    \draw[gedge] (c21top) -- (c22top);
    \draw[gedge] (c21bot) -- (c22bot);
    \draw[gedge] (c21top) -- (c21bot);
    \draw[gedge] (c22top) -- (c22bot);

    % Skip edges in row 3
    \draw[sedge] (s3) -- (c21bot);
    \draw[sedge] (c12top) -- (t3);
    \end{scope}

    %% Terminals for lists
    % s_{13}, t_{13}
    \node[term] (s13) at (\xbig,\ybetrow+\yterm) {};
\node[midnight, above] at (s13) {$s_{13}$};
        \node[term] (t13) at (\xbig,\ybetrow+3*\ygad + 2*\ybetrow - \yterm) {};
\node[midnight, below] at (t13) {$t_{13}$};

    % s_{123}, t_{123}
    \node[term] (s123) at 
    (2*\xbig + 2*\xsml, \ybetrow+\yterm) {};
    \node[midnight, above] at (s123) {$s_{123}$};
    \node[term] (t123) at 
    (2*\xbig + 2*\xsml, \ybetrow + 3*\ygad + 2*\ybetrow -\yterm) {};
    \node[midnight, below] at (t123) {$t_{123}$};

    % Edges exiting s13
    \draw[tedge] (s13) -- (a11top);
    \draw[tedge] (s13) -- (a21top);
    % Edges exiting s123
    \draw[tedge] (s123) -- (a12top);
    \draw[tedge] (s123) -- (a22top);
    % Edges entering t13
    \draw[tedge] (c11bot) -- (t13);
    \draw[tedge] (c21bot) -- (t13);
    % Edges exiting s123
    \draw[tedge] (c12bot) -- (t123);
    \draw[tedge] (c22bot) -- (t123);

    % Edges between rows which we include here to make sure they don't cover elements of row 2
    \draw[cedge] (a11bot) -- (a11botc11topmid);
    \draw[cedge] (a11botc11topmid) -- (c11top);
    \draw[oedge] (a21bot) -- (a21botc21topmid);
    \draw[oedge] (a21botc21topmid) -- (c21top);

    %% Row 2
    \begin{scope}[yshift=\ygad + \ybetrow]
    % s2
    \node[term] (s2) at (0,\ybetrow) {};
    \node[left,midnight] at (s2) {$s_2$};
    % G(b1)
    \node[gvtx] (b1top) at (\xbig+\xsml, \ybetrow) {};
    \node[gvtx] (b1bot) at (\xbig+\xsml, \ybetrow+\ygad) {};
    % G(b2)
    \node[gvtx] (b2top) at (2*\xbig+2*\xsml, \ybetrow) {};
    \node[gvtx] (b2bot) at (2*\xbig+2*\xsml, \ybetrow+\ygad) {};
    % t2
    \node[term] (t2) at (3*\xbig+3*\xsml,\ybetrow+\ygad) {};
    \node[right,midnight] at (t2) {$t_2$};
    % Basic row edges in row 2
    \draw[redge] (s2) -- (b1top);
    \draw[redge] (b1top) -- (b2top);
    \draw[redge] (b1bot) -- (b2bot);
    \draw[redge] (b2bot) -- (t2);
    % Gadget edges in row 2
    \draw[gedge] (b1top) -- (b1bot) {};
    \draw[gedge] (b2top) -- (b2bot) {};
        % Skip edges in row 2
    \draw[sedge] (s2) -- (b2bot);
    \draw[sedge] (b1top) -- (t2);
    \end{scope}
    
    % Remaining edges between rows
    \draw[cedge] (a12bot) -- (b1top);
    \draw[cedge] (b1bot) -- (c12top);

    \draw[oedge] (a12bot) -- (b2top);
    \draw[oedge] (a22bot) -- (b2top);
    \draw[oedge] (b2bot) -- (c22top);

    % Highlighted paths 
    \begin{scope}[on background layer]
        \draw[skobeloff!10!white, line width = 6pt, rounded corners, opacity=0.7] (s13.center) -- (t13.center); 

            \draw[skobeloff!10!white, line width = 6pt, rounded corners, opacity=0.7] (s123.center) -- (a12top.center) -- (c12bot.center) --  (t123.center); 
    \end{scope}

    \begin{scope}[on background layer]
        \draw[Yellow!30!white, line width = 6pt, rounded corners, opacity=0.7] (s1.center) -- (a21bot.center) -- (t1.center);

    \draw[Yellow!30!white, line width = 6pt, rounded corners, opacity=0.7] (s2.center) -- (b2bot.center) -- (t2.center);    

    \draw[Yellow!30!white, line width = 6pt, rounded corners, opacity=0.7] (s3.center) -- (c21bot.center) -- (t3.center);    
    \end{scope}

    % Highlight Triangle
        \begin{scope}[on background layer]
        \draw[skobeloff!10!white, line width = 10pt, rounded corners, opacity=0.7] (A1.center) -- (B1.center) -- (C1.center) -- cycle;
    \end{scope}
\end{tikzpicture}
    \caption{Given an instance $G$ of \textsf{$k$-Clique} for $k=3$ and $n=2$ on the left, with parts $V_1 = \set{u_1, u_2}$, $V_2 = \set{v_1, v_2}$, and $V_3 = \set{w_1, w_2}$, and lists $\langle 1,3\rangle$ and $\langle 1,2,3\rangle$ 
    forming a $k$-covering family 
    for $k=3$, we produce the instance $G'$ of \textsf{$p$-DP} on the right for $p=5$.
    There are three rows of gadgets in $G'$, corresponding to the  parts of $G$.
    The gadgets in rows one and three have two columns because the  integers $1$ and $3$ are covered by two lists, while the gadgets in row two have 
a single column because the integer $2$ is covered by one list.
Note that the paths $\Pi_{13}(u_1,w_1)$ and $\Pi_{13}(u_2,w_2)$ from gadgets in row 1 to row 3 have length two. 
The triangle $(u_1,v_1,w_1)$ in $G$, highlighted on the left, is mapped to the \textsf{$p$-DP}  solution in $G'$ which takes the yellow highlighted paths $P(u_1), P(v_1), P(w_1)$ and the blue highlighted paths from $s_{13}$ to $t_{13}$ and $s_{123}$ to $t_{123}$.
These paths are disjoint because the first three paths skip the gadgets for vertices $u_1$, $v_1$, $w_1$ respectively, which leaves room for the $(s_{13},t_{13})$-path checking for edge $(u_1,w_1)$ in $G$, and the $(s_{123}, t_{123})$-path checking for edges $(u_1,v_1)$ and $(v_1,w_1)$ in $G$. 
} 
    \label{fig:dp-reduction}
\end{figure}

    \subparagraph{Vertex Gadgets}
    \hfill

    For each $i\in [k]$ and  vertex $v\in V_i$, we construct a gadget graph $G_{v}$ as depicted in \Cref{fig:k-DP-gadget}.
    
    For every $r\in [l(i)]$, we include nodes $a_r(v)$ and $b_r(v)$.
    
    We call the $a_r(v)$ and $b_r(v)$ the \emph{top} and \emph{bottom} nodes of $G_v$ respectively.
    
    For every $r\in [l(i)]$, we include edges $e_r(v) = (a_r(v), b_r(v))$.
    For each $r\in [l(i)-1]$, we include  edges from $a_r(v)$ to $a_{r+1}(v)$ and from $b_r(v)$ to $b_{r+1}(v)$.
    
    This completes the description of $G_{v}$.
    
    For convenience, if $L$ is the $r^{\text{th}}$ list containing $v$, we write $a_L(v) = a_r(v)$ and $b_L(v) = b_r(v)$.

    \subparagraph{Arranging Gadgets in Rows}\hfill

    % We now describe how these $G_v$ graphs are arranged and connected in rows.

    For each $i\in [k]$, we list the vertices of $V_i$ in some arbitrary order
        \[v_{i,1}, \dots, v_{i,n}.\]

    For every $j\in [n-1]$, we connect $G_{v_{i,j}}$ to $G_{v_{i,j+1}}$ by  including edges from 
    \[a_{l(i)}(v_{i,j})\text{ to }a_1(v_{i,j+1})\quad \text{and} \quad b_{l(i)}(v_{i,j})\text{ to }b_1(v_{i,j+1}).\]
    
    For every $j\in [n-2]$, we additionally include the edge  from $a_{l(i)}(v_{i,j})$ to $b_{1}(v_{i,j+2})$, which we call the  \emph{skip edge} for $v_{i,j+1}$, because traversing this edge corresponds to skipping the gadget for $v_{i,j+1}$.
    Skip edges are depicted with dashed lines in \Cref{fig:dp-reduction}.
    We view this construction as arranging the $G_v$ gadgets for all $v\in V_i$ in the ``$i^{\text{th}}$ row of $G'$.''

    \subparagraph{Terminal Vertices Checking Nodes}\hfill

    We now introduce source and target nodes and describe how they connect to the $G_v$ gadgets.

For each $i\in [k]$, we introduce new vertices $s_i$ and $t_i$.

    We include edges from $s_i$ to $a_1(v_{i,1})$ and from $b_{l(i)}(v_{i,n})$ to $t_i$.

    Additionally, we include skip edges from $s_i$ to $b_1(v_{i,2})$ and from $a_{l(i)}(v_{i,n-1})$ to $t_i$ (traversing these edges corresponds to skipping $G_{v_{i,1}}$ and $G_{v_{i,n}}$ respectively).

        \subparagraph{Paths Between Rows of Gadgets}\hfill
        
    For every choice of indices $x,y\in [k]$ with $x<y$, list $L\in\ca{L}$ covering $(x,y)$, and vertices $u\in V_x$  and $w\in V_y$ such that $(u,w)$ is an edge in $G$, we include a path 
        $\Pi_L(u,w)$ of length $2(y-x) - 1$ from $b_L(u)$ to $a_L(w)$ in $G'$.
    These paths encode adjacency information about $G$.
    We say the $\Pi_L(u,w)$ are the paths between rows corresponding to list $L$.

    \subparagraph{Terminal Vertices Checking Edges}\hfill

    For each list $L$ in our collection, we introduce source node $s_L$ and target node $t_L$.
    
    Suppose $i$ is the first element in $L$.
    Then for all $v\in V_i$, we have edges from $s_L$ to $a_L(v)$.

    Suppose $j$ is the final element in $L$.
    Then for all $v\in V_j$, we have edges from $a_L(v)$ to $t_L$.

    We observe that any $(s_L,t_L)$-path in $G'$ must  pass through nodes in row $i$ of $G'$ for all $i\in L$.

    \hfill
    
    This completes our construction of the graph $G'$, an example of which is depicted in \Cref{fig:dp-reduction}.
    
    We claim that $G'$ has disjoint paths from its sources to its targets if and only if $G$ has a $k$-clique.

    To prove this result, it will be helpful to first identify some special paths in $G'$.

    For any $i\in [k]$ and $v\in V_i$, we let $P(v)$ be the path which begins at $s_i$, passes through the top nodes of $G_u$ for all  $u\in V_i$ before $v$, then takes the skip edge skipping over $G_v$, passes through the bottom nodes of $G_w$ for all $w\in V_i$ after $v$, and then finally ends at $t_i$.

    \begin{lemma}
        \label{lem:row-sp}
        For every $v\in V$, $P(v)$ is a shortest path in $G'$.
    \end{lemma}
    \begin{proof}
        Fix $v\in V$.
        Let $i\in [k]$ be the index such that $v\in V_i$.

        Recall that for any $j\in [k]$, row $j$ of $G'$ consists of the all gadgets $G_u$ for vertices $u\in V_i$.
        The graph $G'$ is structured so that there is no edge from row $q$ to $p$ whenever $q > p$.
        Then since $s_i$ only has edges to nodes in row $i$ and $t_i$ only has edges from nodes in row $i$, any $(s_i,t_i)$-path in $G'$ must lie completely in row $i$.
        Such a path can use at most one skip edge, since a skip edge goes from a top node to a bottom node, and there are no edges from bottom nodes to top nodes within row $i$.
        If an $(s_i,t_i)$-path uses no skip edges, it has exactly $n\cdot l(i) + 1$ internal nodes.
        If instead the path uses one skip edge, it has exactly $(n-1)\cdot l(i)+1$ internal nodes, because it skips over one gadget.
        
        This implies that any $(s_i,t_i)$-path in $G'$ which uses a skip edge is a shortest path.
        
        Since $P(v)$ uses a skip edge, it is an $(s_i,t_i)$-shortest path as claimed.
    \end{proof}
    %     For any $u\in V_i$, let $\pos(u)$ denote the position of $u$ from the end of the ordered set $V_i$.
        
    %     For example, if $u$ is the final vertex in $V_i$, we set $\pos(u) = 1$, and if $u$ is the first vertex in $V_i$ we set $\pos(u) = |V_i| = n$.

    % \begin{claim}
    %     For every $u\in V_i$, we have 
    % \end{claim}
    % \begin{claimproof}
        
    % \end{claimproof}

    \begin{lemma}[Clique $\Rightarrow$ Disjoint Shortest Paths]
    \label{lem:clique-to-dsp}
        If $G$ has a $k$-clique, then $G'$ has node-disjoint shortest paths from its sources to its targets.
    \end{lemma}
    \begin{proof}
        Let $(v_1,\dots,v_k)\in V_1\times\dots\times V_k$ be a $k$-clique in $G$.

        For each list 
        \[L = \pair{i_1, \dots, i_\ell}\] in $\cal{L}$, 
        let $Q(L)$ be the path which begins at $s_L$, passes through nodes         
            \[a_L(v_{i_1}), b_L(v_{i_1}), \dots, a_L(v_{i_\ell}), b_L(v_{i_\ell})\]
        in that order using the paths corresponding to list $L$, and then ends at $t_L$.
        More precisely, $Q(L)$ goes from $a_L(v_{i_j})$ to $b_L(v_{i_j})$ for each $j\in [\ell]$ by taking the edge between these vertices in $G'$, and goes from $b_L(v_{i_j})$ to $a_L(v_{i_{j+1}})$ for each $j\in [\ell-1]$ by traversing the path $\Pi(v_{i_j}, v_{i_{j+1}})$.

        \begin{claim}
            \label{claim:Q-are-shortest}
            For every list $L\in \cal{L}$,  $Q(L)$ is a shortest path. 
        \end{claim}
        \begin{claimproof}
            The path $\Pi(v_{i_j}, v_{i_{j+1}})$ has length $2(i_{j+1} - i_j) - 1$.
            By telescoping, this means the sum of the lengths of the paths of $\Pi(v_{i_j}, v_{i_{j+1}})$ over all $j\in [\ell-1]$ is $2(i_{\ell} - i_1) - \ell$.
            The only edges of $Q(L)$  not in the $\Pi(v_{i_j}, v_{i_{j+1}})$ paths are the edge from $s_L$ to $a_L(v_{i_1})$, the edges from $a_L(v_{i_j})$ to $b_L(v_{i_j})$ for each $j\in [\ell]$, and the edge from $b_L(v_{i_\ell})$ to $t_L$.
            It follows that the length of $Q(L)$ is
            \[2 + \ell + \left(2(i_{\ell} - i_1) - \ell\right) = 2(i_{\ell} - i_1 + 1).\] 

            To prove the claim, it suffices to show that any $(s_L,t_L)$-path in $G'$ has length at least \[2(i_\ell - i_1 + 1).\]

            To that end, take an arbitrary $(s_L, t_L)$-path $Q$ in $G'$.
            Let $j_1, \dots, j_r\in [k]$ be the sequence of distinct rows $Q$ visits in order. 
            Since $G'$ only includes edges going from row $x$ to row $y$ for $x < y$, we must have $j_1 < j_2 < \dots < j_r$.

            By construction the only way to go from row $x$ to row $y$ in $G'$ is to traverse a path $\Pi_{L'}(u,w)$ for some nodes $u\in V_x$ and $w\in V_y$, and list $L'\in\mathcal{L}$.
            Such a path begins at the bottom of row $x$ at $b_L'(u)$, ends at the top of row $y$ at $a_L'(w)$, and has length $2(y-x) - 1$.
            So the sum of the lengths of all such paths $Q$ traverses to hit rows $j_1, \dots, j_r$ is at least $2(j_r - j_1) - r$ by telescoping. 
            Moreover, $Q$ must traverse at least one edge within each row it visits, to get from the top of that row to the bottom of that row. 
            This accounts for at least $r$ additional edges, so $Q$ has length at least
                \[\left(2(j_r-j_1) - r\right) + r = 2(j_r - j_1).\]
            Finally, $Q$ also must contain an edge leaving $s_L$, and an edge entering $t_L$.
            This accounts for two more edges, so $Q$ has length at least $2(j_r-j_1+1)$.

            By definition, $i_1$ and $i_\ell$ are the first and final elements in list $L$ respectively.
            Consequently, edges exiting $s_L$ must go to row $i_1$, and edges entering $t_L$ must depart from $i_\ell$.
            This forces $j_1 = i_1$ to be the first row $Q$ enters, and $j_r = i_\ell$ to be the last row $Q$ enters.
            Thus $Q$ has length at least 
                \[2(j_r-j_1+1) = 2(i_\ell - i_1 + 1).\]
            This lower bound for the length of an arbitrary $(s_L,t_L)$-path in $G'$ is equal to the length of $Q(L)$.
            Thus $Q(L)$ is a shortest path in $G'$, as claimed.
        \end{claimproof}

        We claim that the collection of paths obtained by taking $P(v_i)$ for each $i\in [k]$ and $Q(L)$ for each $L\in\ca{L}$ is a collection of node-disjoint paths in $G'$.

        Indeed, for each $i\in [k]$, $P(v_i)$ is in the $i^{\text{th}}$ row of $G'$, so $P(v_i)$ and $P(v_j)$ are disjoint for $i\neq j$.
        Similarly, for each choice of lists $L\neq K$ in $\ca{L}$, $Q(L)$ and $Q(K)$ pass through distinct vertices.

        Take $i\in [k]$ and $L\in \ca{L}$.
        The only way $Q(L)$ can intersect $P(v_i)$ is if $i\in L$.
        If $i\in L$, the only nodes at which $Q(L)$ can intersect $P(v_i)$ are $a_L(v_i)$ and $b_L(v_i)$.
        However, $P(v_i)$ skips $G_{v_i}$, so neither of these nodes occur in $P(v_i)$.
        Thus $P(v_i)$ and $Q(L)$ are disjoint.

        Since by \Cref{lem:row-sp} the $P(v_i)$ are shortest paths, and by \Cref{claim:Q-are-shortest} the $Q(L)$ are shortest paths, we have proven the desired result. 
    \end{proof}

    \begin{lemma}[Disjoint Paths $\Rightarrow$ Clique]
    \label{lem:dp-to-clique}
        If $G'$ has node-disjoint paths from its sources to its targets, then $G$ has a $k$-clique.
    \end{lemma}
    \begin{proof}
        Let $\ca{F}$ be a family of node-disjoint paths from the sources to the targets of $G'$.

        For each $i\in [k]$, let $P_i$ denote the $(s_i,t_i)$-path in $\ca{F}$.
        
        For each $L\in\ca{L}$, let $Q_L$ denote the $(s_L,t_L)$-path in $\ca{F}$.

        By construction of $G'$, each $P_i$ must stay in row $i$.
        Since skip edges take us from the top of a row to the bottom of a row, each $P_i$ uses at most one skip edge. 

        We say a path $Q$ ``uses the edges corresponding to a list $L\in\mathcal{L}$'' if $Q$ uses an edge of $\Pi_L(u,w)$, for some vertices $u$ and $w$.

        \begin{claim}
            \label{claim:lists-are-safe}
            For each list $L\in \cal{L}$, the path $Q_L$ does not use edges corresponding to list $L'\in\mathcal{L}$, for any $L'\neq L$.
        \end{claim}
        \begin{claimproof}
            Suppose to the contrary that there does exist a list $L\in\cal{L}$ such that $Q_L$ uses edges corresponding to lists distinct from $\cal{L}$.
            Pick such an $L$ which occurs latest in the ordering of $\cal{L}$.

            Consider the first edge $Q_L$ uses which corresponds to another list.
            
            Suppose this edge belongs to  $\Pi_{L'}(u,w)$ for some list $L'\neq L$, and vertices $u$ and $w$.
            Let $x\in [k]$ be the index such that $u\in V_x$.
            Since the edge we are considering is the first edge $Q_L$ uses which corresponds
            to a list other than $L$, before  this edge $Q_L$ can only have entered rows with indices contained in $L$.
            Thus $x\in L$, and $Q_L$ must have entered row $x$ at a vertex $a_L(v)$, for some $v\in V_x$.
            We also know that $x\in L'$, since otherwise $\Pi_{L'}(u,w)$ would not exist.
            Since $Q_L$ uses an edge of $\Pi_{L'}(u,w)$, we know that the starting vertex of this path $b_{L'}(u)$ must occur after $a_L(v)$ in the topological order of $G'$.
            
            Previously, we observed that each $P_i$ uses at most one skip edge.
            We claim that path $P_x$ uses the skip edge to skip gadget $G_v$ in row $x$.
            Suppose to the contrary that $P_x$ does not use this skip edge.
            Since $P_x$ and $Q_L$ are disjoint, $P_x$ must pass through $b_L(v)$ (otherwise it would hit $a_L(v)$).
            Since $b_{L'}(u)$ occurs after $a_L(v)$ in the topological order, this means that $P_x$ passes through $b_{L'}(u)$.
            This contradicts our assumption that $P_x$ and $Q_L$ are disjoint.

            Thus $P_x$ uses the skip edge corresponding to $G_v$ in row $x$, as claimed.

            Since $P_x$ uses this skip edge, we know that $P_x$ passes through $b_K(v')$ for all lists $K$ containing $x$ and all $v'\in V_x$ after $v$ in the order on $V_x$.
            Since $b_{L'}(u)$ is not in $P_x$ (because $P_x$ and $Q_L$ are disjoint), we must have $u=v$.

            Then for $a_L(v)$ to occur before $b_{L'}(v)$ in the topological order, we must have $L$ occur before $L'$ in the ordering for $\cal{L}$.

            Now, observe that $Q_{L'}$ cannot cross between rows while only using edges corresponding to $L'$.
            Indeed, if $Q_{L'}$ only uses edges corresponding to $L'$ to go between rows, it would pass through a column of row $x$, since $x\in L'$.
            To be disjoint from $P_x$, this is only possible if $Q_{L'}$ uses a column of $G_v$, since $P_x$ hits the columns of every other gadget in row $x$. 
            But $Q_{L}$ and $Q_{L'}$ are disjoint, and we already said that $Q_L$ hits the column of $G_v$ corresponding to $L'$.

            Thus, $Q_{L'}$ must use an edge corresponding to a list distinct from $L'$.
            We also observed earlier that $L'$ occurs after $L$ in the ordering on $\cal{L}$.

            This contradicts our choice of $L$ as the final list in $\cal{L}$ with the property that $Q_L$ uses an edge corresponding to a list distinct from $L$.

            Thus our initial assumption was false, and the claim holds.            
        \end{claimproof}

        \begin{claim}
            \label{captain-skipper}
            For every $i\in [k]$, $P_i$ uses a skip edge.
        \end{claim}
        \begin{claimproof}
        Suppose to the contrary that there exists an index $i\in [k]$ such that $P_i$ does not use a skip edge.
        Then for every $v\in V_i$ and $r\in [s(i)]$, the path $P_i$ passes through  at least one of $a_r(v)$ and $b_r(v)$.
        Let $L\in\ca{L}$ be a list containing $i$ (such a list exists because
        $\ca{L}$ is a $k$-covering family).
        By \Cref{claim:lists-are-safe} and the fact that each $P_i$ uses at most one skip edge, the path $Q_L$ must pass through vertices $a_L(v)$ and $b_L(v)$ for some $v\in V_i$.
        Consequently, $Q_L$ intersects $P_i$, which contradicts the assumption that paths in $\ca{F}$ are disjoint.
        So each $P_i$ traverses a skip edge as claimed.
        \end{claimproof}

        Since a skip edge moves a path from the top nodes to the bottom nodes of a row, by \Cref{captain-skipper} we know that each $P_i$ traverses exactly one skip edge.    
        
        For each $i\in [k]$, let $v_i\in V_i$ be the unique vertex such that $P_i$ traverses the edge skipping $G_{v_i}$.
        Then we claim that $(v_1, \dots, v_k)$ is a clique in $G$.

        Indeed, take any pair $(i,j)\in[k]^2$ with $i<j$.
        Let $L\in\ca{L}$ be a list covering $(i,j)$.

By \Cref{claim:lists-are-safe},  
        the path $Q_L$ traverses a path $\Pi_L(u,w)$ for some $u\in V_i$ and $w\in V_j$.
        
        Since the only gadgets skipped by $P_i$ and $P_j$ are $G_{v_i}$ and $G_{v_j}$ respectively, and $Q_L$ is disjoint from $P_i$ and $P_j$, we see that in fact $Q_L$ traverses the path $\Pi_L(v_i,v_j)$.
        
        But this path exists in $G'$ only if $(v_i,v_j)$ is an edge in $G$.
        
        So $(v_i,v_j)$ is an edge in $G$ for all $i<j$, and thus $G$ contains a $k$-clique as claimed.        
    \end{proof}

\lbdp*
\begin{proof}
    Construct a collection $\ca{L}$ of $\lambda=\lfloor k^2/4\rfloor$ lists which form a $k$-covering family using \Cref{puzzle-solution}.
    For constant $k$, this takes $O(1)$ time.
    This collection has the property that each element of $[k]$ shows up in fewer than $k$ list of $\mathcal{L}$.

    Let $G$ be the input instance of \textsf{$k$-Clique}.

    Using $\ca{L}$ and $G$, construct the graph $G'$ described previously in this section.
    
    By \Cref{lem:clique-to-dsp,lem:dp-to-clique}, solving \textsf{$p$-DP} or \textsf{$p$-DSP} on $G'$ solves \textsf{$k$-Clique} on $G$, where 
        \[p = k + \lambda = k + \lfloor k^2/4\rfloor.\]
    
    Since $G'$ consists of $2(k + \lambda)$ terminals and $kn$ gadget graphs $G_v$, each on fewer than $2k$ nodes, the graph $G'$ has at most $O(kn)$ nodes.
    Since we can construct $G'$ in time linear in its size, this proves the desired result.
\end{proof}

\section{Conclusion}
\label{sec:conclusion}

In this work, we obtained linear time algorithms for \textsf{2-DSP} in undirected graphs and DAGs.
These algorithms are based off algebraic methods, and as a consequence are \emph{randomized} and only solve the \emph{decision}, rather than search, version of \textsf{2-DSP}.
This motivates the following questions:

\begin{open}
    Is there a \emph{deterministic} linear time algorithm solving \textsf{2-DSP}?
\end{open}

\begin{open}
    Given a DAG or undirected graph $G$  with sources $s_1, s_2$ and targets $t_1, t_2$, is there a linear time algorithm \emph{finding} disjoint $(s_i,t_i)$-shortest paths in $G$ for $i\in \set{1,2}$?
\end{open}

It is also an interesting research direction to see if algebraic methods can help design faster algorithms for \textsf{$k$-DSP} in undirected graphs and DAGs when $k\ge 3$, or help tackle this problem in the case of general directed graphs.

In this work, we also established tighter reductions from finding cliques to disjoint path and shortest path problems.
There still remain large gaps however, between the current best conditional lower bounds and current fastest algorithms for these problems.

\begin{open}
    Is there a fixed integer $k\ge 3$ and constant $\delta > 0$ such that \textsf{$k$-DSP} in DAGs can be solved in $O(n^{k+1-\delta})$ time? 
    Or does some popular hypothesis rule out such an algorithm?
\end{open}

Since \textsf{$k$-Clique} admits nontrivial algorithms by reduction to matrix multiplication, it is possible that \textsf{$k$-DSP} can be solved faster using fast matrix multiplication algorithms.
On the other hand, if we want to rule out this possibility and obtain better conditional lower bounds for \textsf{$k$-DSP}, we should design reductions from problems which are harder than \textsf{$k$-Clique}.
In this context, a natural strategy would be to reduce from \textsf{Negative $k$-Clique} and \textsf{3-Uniform $k$-Hyperclique} instead, since these problems are conjectured to require $n^{k-o(k)}$ time to solve (and it is not known how to leverage matrix multiplication to solve these problems faster than exhaustive search).

For all $k\ge 3$, the current fastest algorithm for \textsf{$k$-DSP} in undirected graphs takes $n^{O(k\cdot k!)}$ time, much slower than the $O(mn^{k-1})$ time algorithm known for the problem in DAGs.
Despite this, every conditional lower bound that has been established for \textsf{$k$-DSP} in undirected graphs so far also extends to showing the same lower bound for the problem in DAGs.
This is bizarre behavior, and suggests we should try establishing a lower bound which separates the complexities of \textsf{$k$-DSP} in undirected graphs and DAGs.
If designing such a lower bound proves difficult, that would offer circumstantial evidence that far faster algorithms for \textsf{$k$-DSP} in undirected graphs exist.

\begin{open}
    Can we show a conditional lower bound for \textsf{$k$-DSP} in undirected graphs, which is stronger than any conditional lower bound known for \textsf{$k$-DSP} in DAGs? 
\end{open}

Finally, for large $k$, the best conditional time lower bounds we have for \textsf{$k$-DP} in DAGs are far weaker than the analogous lower bounds we have for \textsf{$k$-DSP} in DAGs.
This is despite the fact that the fastest algorithms we have for both problems run in the same time.
It would nice to resolve this discrepancy, either by designing faster algorithms for the latter problem, or showing better lower bounds for the former problem.

\begin{open}
Is there a fixed integer $k\ge 3$  such that we can solve \textsf{$k$-DP} in DAGs faster than we can solve \textsf{$k$-DSP} in weighted DAGs? 
\end{open}

\begin{open}
    Can we show a conditional lower bound for \textsf{$k$-DP} in DAGs matching the best known conditional lower bound for \textsf{$k$-DSP} in DAGs? 
\end{open}

\newpage

\bibliographystyle{alpha}
\bibliography{main}

\newpage

\appendix 

% \appendix 

\section{Standard Reductions}
\label{app:unnecessary}

In this section, we collect proofs of some standard reductions that allow us to move between different variants of disjoint path problems.

\begin{proposition}[Edge-Disjoint $\le$ Vertex-Disjoint]
    \label{app-prop:edge-to-vertex}
    There is a reduction from \textsf{$k$-EDSP} on $n$ vertices and $m$ edges to \textsf{$k$-DSP} on $m+k(n+2)$ nodes and $2k(m+1)$ edges.
\end{proposition}
\begin{proof}
    Let $G$ be the instance of \textsf{$k$-EDSP} on $n$ vertices and $m$ edges, with sources $s_1, \dots, s_k$ and targets $t_1, \dots, t_k$.
    Let $V$ and $E$ be the vertex and edge sets of $G$ respectively.
    
    We construct a graph $G'$ as follows.
    For every vertex $v\in V$, $G'$ has $k$ nodes $v_1, \dots, v_k$ (we call these the \emph{copies} of $v$ in $G'$).
    For every edge $e\in E$, $G'$ has a node $e$.
    For every edge $e = (v,w)\in E$, we include edges in $G'$ from $v_i$ to $e$ and from $e$ to $w_i$ for all $i\in [k]$.
    If $e = (v,w)$ had weight $\ell(v,w)$ in $G$, then the $(v_i,e)$ and $(e,w_i)$ edges in $G$ each have weight $\ell(v,w)$.

    Finally, we introduce new sources $s'_1, \dots, s'_k$ and targets $t'_1, \dots, t'_k$ in $G'$.
    Then for every $i,j\in [k]$, we add edges from $s'_i$ to $(s_i)_j$ and from $(t_i)_j$ to $t'_i$ of weight $1$.

    From this construction, $G'$ has $m+kn + 2k$ nodes and $2km + 2k$ edges as claimed.

    Suppose we have vertex-disjoint $(s'_i,t'_i)$-shortest paths $P'_i$ in $G'$.
    We can assume each $P'_i$ never traverses two copies of the same vertex $v\in V$ (otherwise, we could remove the subpath between two copies of $v$ and obtain a shorter path than $P'_i$).
    We map each $P'_i$ to an $(s_i,t_i)$-shortest path $P_i$ in $G$, by having $P_i$ pass through the vertices $v\in V$ for which $P'_i$ contains a copy of $v$, in the order the copies appear in $P'_i$.

    By construction, if $P_i$ has length $\ell$, then $P'_i$ has length $2\ell + 2$.
    So since the $P'_i$ are shortest paths, the $P_i$ are also shortest paths.
    The $P_i$ are also edge-disjoint, since if some $P_i$ and $P_j$ overlap at an edge $e$, the paths $P'_i$ and $P'_j$ would overlap at node $e$ in $G'$, which would contradict the assumption that the $P'_i$ are vertex-disjoint.

    So any solution to \textsf{$k$-DSP} on $G'$ pulls back to a solution to \textsf{$k$-EDSP} on $G$.

    Conversely, given edge-disjoint $(s_i,t_i)$-shortest paths $P_i$ in $G$ of the form 
    \[P_i = \langle v_{i,1}, \dots, v_{i,\ell_i}\rangle \]
    we can produce $(s'_i,t'_i)$-paths $P'_i$ in $G'$ of the form 
    \[P'_i = \langle s'_i, (v_{i,1})_i, ((v_{i,1})_i, (v_{i,2})_i), (v_{i,2})_i, \dots, (v_{i,\ell_i})_i, t'_i\rangle.\]
    Similar reasoning to the above shows that the $P'_i$ are shortest paths in $G'$.
    These paths are vertex-disjoint because each $P'_i$ only uses the $i^{\text{th}}$ copies of $v\in V$, and the $P_i$ were edge-disjoint, so the $P'_i$ cannot overlap at any nodes of the form $e\in E$.

    So any solution to \textsf{$k$-EDSP} in $G$ has a corresponding solution to \textsf{$k$-DSP} in $G'$.

    This proves the desired result.
\end{proof}

As mentioned in \Cref{sec:intro}, \Cref{app-prop:edge-to-vertex} combined with \Cref{2dsp-undir,2dsp-DAG} implies \Cref{corr:2-EDSP}, and  \Cref{app-prop:edge-to-vertex} combined with the $O(mn^{k-1})$ time algorithm for \textsf{$k$-DSP} in DAGs from  \cite[Theorem 3]{FortuneHopcroftWyllie1980} implies that \textsf{$k$-EDSP} in DAGs can be solved in $O(m^k)$ time for constant $k$.

\begin{proposition}[Disjoint Paths $\le$ Disjoint Shortest Paths]
    \label{app-prop:vertex-to-edge}
    There are reductions from \textsf{$k$-DP} on DAGs  to  \textsf{$k$-DSP} on DAGs and from \textsf{$k$-EDP} on DAGs to \textsf{$k$-EDSP} on DAGs with the same number of nodes and edges respectively.
\end{proposition}
\begin{proof}
    Let $G$ be the DAG which is an instance of \textsf{$k$-DP} or \textsf{$k$-EDP}.
    Take a topological order $v_1, \dots, v_n$ of the vertices in $G$.
    Let $G'$ be the DAG with the same vertex and edge sets as $G$, but where each edge $(v_i,v_j)$ has weight $(j-i)$.
    Then by telescoping, for any vertices $v_i$ and $v_j$ in $G'$, every path from $v_i$ to $v_j$ in $G'$ has total length $(j-i)$.
    
    Consequently, every path in $G$ becomes a shortest path in $G'$.
    
    So solving \textsf{$k$-DSP} and \textsf{$k$-EDSP} in $G'$ corresponds to solving \textsf{$k$-DP} and \textsf{$k$-EDP} respectively in $G$, which proves the desired result.
\end{proof}

As mentioned in \Cref{sec:intro}, \Cref{app-prop:vertex-to-edge} shows that \textsf{$k$-EDSP} in weighted DAGs generalizes the \textsf{$k$-EDP} problem in DAGs.
\end{document}